\documentclass[11pt]{article}
\usepackage[T1]{fontenc}
\usepackage[utf8]{inputenc}
\usepackage[margin=1in]{geometry}
\usepackage{amsmath,amssymb,amsthm,mathtools,microtype,booktabs}
\usepackage{authblk}
\usepackage[dvipsnames]{xcolor}
\usepackage{tikz}
\usepackage{cancel}
\usetikzlibrary{arrows.meta,positioning,decorations.pathreplacing,calc}
\usepackage[colorlinks=true,linkcolor=MidnightBlue,citecolor=MidnightBlue,
 urlcolor=MidnightBlue]{hyperref}
\numberwithin{equation}{section}
\newtheorem{theorem}{Theorem}[section]
\newtheorem{proposition}[theorem]{Proposition}
\newtheorem{lemma}[theorem]{Lemma}
\newtheorem{corollary}[theorem]{Corollary}
\theoremstyle{definition}
\newtheorem{definition}[theorem]{Definition}
\theoremstyle{remark}
\newtheorem{remark}[theorem]{Remark}
\newenvironment{acknowledgments}{\section*{Acknowledgments}}{\par}
\RequirePackage{tcolorbox}
\tcbuselibrary{skins,breakable}

\definecolor{MLSIDefinitionAccent}{HTML}{526F96}
\definecolor{MLSITheoremAccent}{HTML}{8A667E}
\definecolor{MLSIOtherAccent}{HTML}{5B8271}
\colorlet{MLSIDefinitionFill}{MLSIDefinitionAccent!10!white}
\colorlet{MLSIDefinitionBorder}{MLSIDefinitionAccent!55!white}
\colorlet{MLSITheoremFill}{MLSITheoremAccent!10!white}
\colorlet{MLSITheoremBorder}{MLSITheoremAccent!55!white}
\colorlet{MLSIOtherFill}{MLSIOtherAccent!10!white}
\colorlet{MLSIOtherBorder}{MLSIOtherAccent!55!white}

\tcbset{
  mlsi statement/.style={
    enhanced jigsaw,
    breakable,
    lines before break=6,
    boxrule=0.4pt,
    boxsep=0pt,
    left=7pt,
    right=7pt,
    top=5pt,
    bottom=5pt,
    arc=1.4mm,
    grow sidewards by=7.4pt,
    before skip=11pt plus 2pt minus 1pt,
    after skip=11pt plus 2pt minus 1pt,
  },
  mlsi definition/.style={
    mlsi statement,
    colback=MLSIDefinitionFill,
    colframe=MLSIDefinitionBorder,
    borderline west={1.8pt}{0pt}{MLSIDefinitionAccent},
  },
  mlsi theorem/.style={
    mlsi statement,
    colback=MLSITheoremFill,
    colframe=MLSITheoremBorder,
    borderline west={1.8pt}{0pt}{MLSITheoremAccent},
  },
  mlsi other/.style={
    mlsi statement,
    colback=MLSIOtherFill,
    colframe=MLSIOtherBorder,
    borderline west={1.8pt}{0pt}{MLSIOtherAccent},
  },
}

\tcolorboxenvironment{definition}{mlsi definition}
\tcolorboxenvironment{theorem}{mlsi theorem}
\tcolorboxenvironment{proposition}{mlsi other}
\tcolorboxenvironment{lemma}{mlsi other}
\tcolorboxenvironment{corollary}{mlsi other}
 \DeclareMathOperator{\Tr}{Tr}
\DeclareMathOperator{\EP}{EP}
\DeclareMathOperator{\id}{id}
\DeclareMathOperator{\supp}{supp}
\DeclareMathOperator{\Ent}{Ent}
\DeclareMathOperator{\Ran}{Ran}
\newcommand{\one}{\mathbf1}
\newcommand{\cL}{\mathcal L}
\newcommand{\cP}{\mathcal P}
\newcommand{\cG}{\mathcal G}
\newcommand{\norm}[1]{\left\|#1\right\|_\infty}
\newcommand{\hs}[1]{\left\|#1\right\|_2}
\newcommand{\ip}[2]{\langle #1,#2\rangle_{\rm HS}}
\title{Modified logarithmic Sobolev inequality\\ 
for 1D non-commuting Hamiltonians}
\renewcommand{\Affilfont}{\footnotesize\normalfont}
\author[1,2]{Ángela Capel}
\author[3,4]{David Pérez-García}
\author[5]{Matteo Scandi}
\affil[1]{Department of Applied Mathematics and Theoretical Physics,
University of Cambridge, Cambridge CB3~0WA, United Kingdom}
\affil[2]{Fachbereich Mathematik, Universität Tübingen, 72076 Tübingen, Germany}
\affil[3]{Departamento de Análisis Matemático y Matemática Aplicada,
Universidad Complutense de Madrid, 28040 Madrid, Spain}
\affil[4]{Instituto de Ciencias Matemáticas (ICMAT), 28049 Madrid, Spain}
\affil[5]{Instituto de Física Teórica UAM/CSIC,
C. Nicolás Cabrera 13--15, Cantoblanco, 28049 Madrid, Spain}
\makeatletter
\renewcommand{\@author}{%
  \AB@authlist\\[\affilsep]
  \makebox[\textwidth][c]{%
    \parbox{\dimexpr\textwidth+0.8in\relax}{%
      \centering\Affilfont\AB@affillist\par}}}
\makeatother
\date{24 September 2026}

\definecolor{ms}{rgb}{0,.4,1}

\begin{document}
\maketitle
\begin{abstract}
Given a finite-range, non-commuting Hamiltonian on a 1D chain of finite spins, we prove a modified logarithmic Sobolev inequality with constant
 $\Omega(1/\log n)$ for the heat-bath dynamics, a quasi-local Gibbs sampler, and for the regularised heat-bath dynamics.
This implies rapid mixing of both samplers towards the Gibbs state in relative entropy and trace norm, with mixing times of order $\mathcal{O}(\log n\cdot \log(n/\epsilon))$. The interactions may be spatially inhomogeneous. The proof combines weak quasi-factorisation of relative
entropy, decay-of-correlation estimates, a uniform conditional local gap, and a buffered comparison
between conditional entropy and quadratic energy. The resulting
bounded entropy defect is removed using the global gap. To the best of our knowledge, this is the first proof of MLSI in the non-commuting regime.
\end{abstract}

\tableofcontents

\section{Introduction and main result}\label{sec:introduction}
How quickly does a quantum spin system approach thermal equilibrium?
Understanding how this time scales with system size is central to
dissipative many-body dynamics and to Gibbs state preparation.
In the first setting, a quantum Markov semigroup describes a system
coupled to a thermal environment \cite{davies1974}; in the second,
the dynamics is designed to prepare a prescribed Gibbs state
\cite{metropolis2011,brandao-kastoryano2019,ckg2023,ding-li-lin2025}.

We consider a chain of $n$ spins with a finite-range Hamiltonian $H$
whose interaction terms, not necessarily commuting, are uniformly bounded and may vary from
site to site. Throughout, we assume that the local dimension is
$d=2^q$ for a fixed integer $q\ge1$ and work at a fixed inverse
temperature $\beta>0$.
We write $\sigma=e^{-\beta H}/\Tr[e^{-\beta H}]$ for the Gibbs state.
All logarithms are natural, and
$D(\rho\Vert\sigma)=\Tr[\rho(\log\rho-\log\sigma)]$
denotes the relative entropy \cite{umegaki1962}.
For a primitive quantum Markov evolution $e^{t\cL_*}$ preserving $\sigma$,
define the trace-norm and relative-entropy mixing times by
\begin{equation}\label{eq:mixing-times}
 \begin{aligned}
 t_{\rm mix}(\epsilon)
   &:=\inf\{t\ge0:\sup_\rho\|e^{t\cL_*}\rho-\sigma\|_1\le\epsilon\},\\
 t_{{\rm mix},D}(\epsilon)
   &:=\inf\{t\ge0:\sup_\rho D(e^{t\cL_*}\rho\Vert\sigma)\le\epsilon\},
 \end{aligned}
\end{equation}
where the suprema run over all initial states. The trace-norm mixing
time is standard; here we prove a bound for the stronger relative-entropy
notion. Pinsker's inequality gives the comparison
\begin{equation}\label{eq:mixing-comparison}
 \|\rho-\sigma\|_1^2\le2D(\rho\Vert\sigma)
 \quad\Longrightarrow\quad
 t_{\rm mix}(\epsilon)\le t_{{\rm mix},D}(\epsilon^2/2).
\end{equation}

For reversible primitive dynamics, a spectral gap $\lambda>0$
controls the exponential decay of a weighted quadratic
distance to equilibrium. The resulting convergence bound also
depends on the initial value of this distance, which can grow
exponentially with $n$. Consequently, the direct spectral-gap
estimate for the mixing time of the evolution is
$t_{\rm mix}(\epsilon)=\mathcal{O}(\lambda^{-1}[n+\log(1/\epsilon)])$.
To obtain a polylogarithmic bound, we instead study the decay of
relative entropy
\cite{olkiewicz-zegarlinski1999,kastoryano-temme2013}.

 A modified logarithmic Sobolev inequality
(MLSI) with constant $\alpha$ asserts that the relative entropy between the evolved state and the equilibrium decays as
\[
 D(e^{t\cL_*}\rho\Vert\sigma)
 \le e^{-\alpha t}D(\rho\Vert\sigma).
\]
Since the largest initial relative entropy is $\mathcal{O}(n)$ for
the Gibbs state of a local Hamiltonian, this directly gives
\[
 t_{{\rm mix},D}(\epsilon)
 =\mathcal{O}\!\left(\alpha^{-1}
          [\log n+\log(1/\epsilon)]\right).
\]
The same order of bound holds for $t_{\rm mix}(\epsilon)$ by
\eqref{eq:mixing-comparison}.

We prove such an inequality for two Gibbs samplers of $H$. The first,
the \emph{heat-bath sampler}, with generator $\cL_*^{\mathrm{HB}}$,
discards a single site and applies the Petz recovery channel
associated with $\sigma$ \cite{petz-sufficiency}.
This construction is of interest because it combines exact
preservation of the Gibbs state with locality: for bounded
finite-range, non-commuting interactions in one dimension,
$\cL_*^{\mathrm{HB}}$ is a quasi-local Lindbladian at every positive
temperature (cf. Proposition~\ref{loc:heatbath-quasilocal}).
The second, the \emph{regularised heat-bath sampler}, with generator
$\cL_*^{\mathrm{RHB}}$, replaces each heat-bath update $\cP_{i*}$
by the adjoint $\mathbb E_{i*}$ of the conditional expectation
obtained by iterating $\cP_i$ to its limit,
following the construction of \cite{BardetCapelRouze2021NoncommutingAT}.

With each site's update occurring once per unit time on average, our main result,
Theorem~\ref{thm:main}, is
\[
 \alpha_{\rm MLSI}(\cL_*^{\mathrm{HB}}),\,
 \alpha_{\rm MLSI}(\cL_*^{\mathrm{RHB}})
 \ge \frac{c_\beta}{\log n},\qquad n\ge2.
\]
The constant depends on $\beta$, the local dimension, and the
interaction bounds, but is independent of $n$. To the best of our knowledge, this is the first proof of MLSI in the non-commuting regime. In particular,
Corollary~\ref{cor:mixing} gives
\[
 \max\{t_{{\rm mix},D}^a(\epsilon),t_{\rm mix}^a(\epsilon)\}
 \le C_\beta\log n \cdot \log(n/\epsilon),
 \qquad a\in\{\mathrm{HB},\mathrm{RHB}\},
\]
for both dynamics at any positive temperature. For fixed accuracy this is a mixing-time bound
of order $(\log n)^2$, without a high-temperature assumption.

The connection between equilibrium correlations and entropy decay
has a long history in classical spin systems
\cite{martinelli1999,caputo-menz-tetali2015}. Approximate tensorisation
and quasi-factorisation relate global relative entropy to conditional
entropies, which local entropy inequalities connect to dissipation.
Quantum versions of this approach use conditional relative entropy
\cite{CapelLuciaPerezGarcia_ConditionalRelativeEntropy_2018,
BardetCapelRouze2021NoncommutingAT}, within a broader theory of
noncommutative entropy dissipation \cite{carlen-maas2017}.
For commuting Hamiltonians, clustering and functional inequalities
are closely related \cite{kastoryano-brandao2016}, and MLSI for the
heat-bath sampler in one dimension was obtained under
mixing and strong quasi-factorisation assumptions
\cite{BardetCapelLuciaPerezGarciaRouze_HeatBath1D_2021}.
Under assumptions on spatial mixing and entropy decay on finite
blocks, bounds independent of system size were established for
the so-called Schmidt dynamics associated with classical spin systems and quantum systems with commuting
nearest-neighbour interactions \cite{CapelRouzeStilckFranca2020MLSI}.

For Davies dynamics of commuting chains, MLSI bounds of order
$1/\log n$ were proved at every positive temperature for
translation-invariant finite-range interactions
\cite{BardetEtAlEntropyDecayDavies,BardetEtAlRapidThermalization2023}.
These were subsequently strengthened to bounds independent of the
chain length, also allowing spatially inhomogeneous interactions
\cite{KochanowskiAlhambraCapelRouze2024}, and hence logarithmic
mixing times at fixed accuracy. Beyond chains, such bounds hold for
Davies dynamics of translation-invariant CSS codes under
Dobrushin--Shlosman conditions \cite{StengeleEtAl2025CSS}, and at
every positive temperature for two-dimensional Abelian quantum
double models \cite{StengeleEtAl2026AbelianDouble} and Haah's
three-dimensional cubic code \cite{StengeleCahaCapelWarzel2026Haah},
both with periodic boundary conditions.

One dimension supplies strong equilibrium locality even when the
interaction terms do not commute. Imaginary-time estimates
\cite{araki1969}, and their complex-time extensions
\cite{perez-locality2023}, control the influence of a boundary on
distant observables. Related information-theoretic results describe
approximate quantum Markov structure \cite{kato-brandao2019},
decay of mutual information \cite{bcp}, and locality of marginal
logarithms and conditional mutual information
\cite{kuwahara2025clustering}. Recent work also proves a uniform
gap for the canonical purified Hamiltonian of a one-dimensional
Gibbs state \cite{LuciaPerezGarciaPerezHernandez2025Gap}.
For a different Gibbs sampler, the analysis in
\cite{BergamaschiChen2025FastMixing} establishes a uniform gap at
all positive temperatures and develops multiplicative mixing
estimates that we also use here.

The passage to entropy decay requires handling two difficulties
in the noncommuting setting. First, a Gibbs marginal $\sigma_S$ need not be the
Gibbs state of the restricted Hamiltonian $H_S$. Entropy
factorisation involves logarithms of these marginals, while Petz
recovery involves their inverse square roots. Control of ordinary
correlations alone does not supply the operator estimates required
for these expressions. Second, a heat-bath update need not be
idempotent, and its fixed algebra can be smaller than the full
exterior algebra. Consequently the conditional entropy associated
with discarding a region cannot simply be identified with entropy
relative to the fixed algebra of the corresponding dynamics.

Our proof first covers the chain by overlapping intervals of length
$\mathcal{O}(\log n)$. Operator-norm
control of logarithmic marginal combinations, expressed through
the matrix-valued conditional mutual information of
\cite{CapelGondolfKochanowskiRouze2024MCMI}, together with
multiplicative mixing across a traced separator, bounds the global
relative entropy by the sum of the block conditional entropies
and a bounded error.

We next compare each block entropy with the quadratic energy of
updates in a slightly larger, buffered interval. The canonical
projection gap of \cite{LuciaPerezGarciaPerezHernandez2025Gap} transfers to a
conditional gap for the heat-bath sampler by the Kadison--Schwarz
inequality. Araki expansionals and the marginal-logarithm locality
estimate of \cite{kuwahara2025clustering} then show that changes
beyond the buffer are approximately recovered by the Petz map.
This use of recovery complements its general relation to
conditional mutual information \cite{fawzi-renner2015}, while
providing the uniform exterior-change estimate needed here.

Summing the buffered estimates gives an entropy--energy inequality
with a logarithmic coefficient and an additive defect independent
of $n$. The global gap removes this defect through the
norm-to-entropy method of \cite{norms}, yielding MLSI. Comparison
of the quadratic forms makes the argument apply to both samplers.
The proof also gives complete entropy inequalities for individual
updates, valid in the presence of an arbitrary untouched quantum
memory \cite{gao-rouze2022,cp-order,bkm-coercivity}. Global complete
bounds are stated separately: the $1/\log n$ rate in the main
theorem is an ordinary MLSI rate.

\subsection{Model and samplers}\label{sec:model}
Let $\Lambda=\{1,\ldots,n\}$, with $n\ge2$, be a finite chain
with open boundary conditions and local Hilbert space $\mathbb C^d$ at each site. The
Hilbert space of the chain is
$
 \mathcal H_\Lambda=\bigotimes_{i\in\Lambda}\mathbb C^d.
$
An interaction $\Phi$ assigns to each nonempty set $X\subseteq\Lambda$
a self-adjoint operator $\Phi(X)$ supported on $X$, and we set
$\Phi(\varnothing)=0$. The associated
Hamiltonian is
\begin{equation}\label{eq:model}
 H_\Lambda=\sum_{X\subseteq\Lambda}\Phi(X).
\end{equation}
We assume finite interaction range: $\Phi(X)=0$ whenever
$\operatorname{diam}(X)>r_0$. We also assume the uniform bound
\begin{equation}\label{eq:interaction-bound}
 \sup_{i\in\Lambda}\sum_{X\ni i}\norm{\Phi(X)}\le J,
\end{equation}
where $r_0$ and $J$ are independent of $n$. The interaction terms
may vary from site to site and need not commute; scalar terms may
be discarded. Throughout, $d$, $r_0$, $J$, and $\beta$ are fixed.
Constants carrying a subscript $\beta$ may depend on these parameters,
but not on $n$.

For $V\subseteq\Lambda$, write
$H_V=\sum_{X\subseteq V}\Phi(X)$ for the restricted Hamiltonian.
At inverse temperature $\beta>0$, its Gibbs state is
\begin{equation}\label{eq:gibbs-state}
 \sigma^V=\frac{e^{-\beta H_V}}{Z_V(\beta)},
 \qquad \sigma=\sigma^\Lambda,
\end{equation}
where $Z_V(\beta)=\Tr[e^{-\beta H_V}]$ is the partition function.
We reserve subscripts for marginals of the global state $\sigma$:
\[
 \sigma_V=\Tr_{V^c}[\sigma],\qquad V^c=\Lambda\setminus V.
\]
In general, $\sigma_V$ need not coincide with $\sigma^V$. We abbreviate
$H_\Lambda$ to $H$ when the ambient chain is fixed.

For a positive definite operator $\tau$ and $z\in\mathbb C$, set
\begin{equation}\label{eq:gamma-notation}
 \begin{aligned}
  \Gamma_\tau^z(X)=\tau^{z/2}X\tau^{z/2}, \quad
  \Delta_\tau^z(X)=\tau^zX\tau^{-z}.
 \end{aligned}
\end{equation}
These are the standard sandwiching and modular superoperators; our
conventions agree with \cite[Section 2.3]{KochanowskiAlhambraCapelRouze2024}.

For $S\subseteq\Lambda$, let $S^c=\Lambda\setminus S$ and
$\rho_{S^c}=\Tr_S[\rho]$. The heat-bath update in the Schr\"odinger picture is
\begin{equation}\label{eq:petz}
 \cP_{S*}(\rho)=\Gamma_\sigma\!\left(
 \one_S\otimes\Gamma_{\sigma_{S^c}}^{-1}(\rho_{S^c})\right).
\end{equation}
Thus the update discards $S$ and applies the Petz recovery channel to
the remaining marginal \cite{petz-sufficiency}. Let $\cP_S$ denote
its trace adjoint, acting on observables. The corresponding regularised heat-bath
update is
\begin{equation}\label{eq:regularisation}
 \mathbb{E}_S=\lim_{m\to\infty}\cP_S^m.
\end{equation}
The limit exists and is the $\sigma$-preserving conditional expectation
onto the fixed algebra of $\cP_S$; see
\cite[Section 2.4.1, equation (2.6)]{BardetCapelRouze2021NoncommutingAT}
and Proposition~\ref{pre:properties}. Its trace adjoint is denoted by
$\mathbb{E}_{S*}$. For a singleton $S=\{i\}$, we omit the braces in the subscripts.

\begin{definition}[Heat-bath and regularised heat-bath samplers]
\label{def:samplers}
The \emph{heat-bath sampler} and the \emph{regularised heat-bath sampler} associated
with $\sigma$ are the quantum Markov semigroups
with Schr\"odinger-picture generators
\begin{subequations}\label{eq:generators}
 \begin{align}
  \cL_*^{\mathrm{HB}}
    &=\sum_{i\in\Lambda}(\cP_{i*}-\id),
       \label{eq:generator-heatbath}\\
  \cL_*^{\mathrm{RHB}}
    &=\sum_{i\in\Lambda}(\mathbb{E}_{i*}-\id).
       \label{eq:generator-regularised-heatbath}
 \end{align}
\end{subequations}
\end{definition}

Every update uses the same global reference state $\sigma$.
The site label specifies the region discarded in the Petz
construction; the exact update maps may act beyond that site.
For the heat-bath updates $\cP_{i*}$, this dependence decays
exponentially with distance: Proposition~\ref{loc:heatbath-quasilocal}
proves quasi-locality by constructing local channel approximations.

We write $\EP^a$ for the entropy production
of the sampler $a\in\{\mathrm{HB},\mathrm{RHB}\}$, defined in
\eqref{eq:ep} below.
Our MLSI convention is $\EP^a(\rho)\ge\alpha D(\rho\Vert\sigma)$,
without a factor of two.

\begin{theorem}[Modified logarithmic Sobolev inequality]\label{thm:main}
Under the assumptions above, both generators in \eqref{eq:generators}
are primitive and KMS symmetric with invariant state $\sigma$.
The heat-bath generator $\cL_*^{\mathrm{HB}}$ is quasi-local in
the sense of Proposition~\ref{loc:heatbath-quasilocal}.
There is a constant $C_\beta<\infty$, independent of $n$, such that
\begin{equation}\label{eq:mlsi}
 D(\rho\Vert\sigma)\le C_\beta\log n\,\EP^a(\rho),
 \qquad a\in\{\mathrm{HB},\mathrm{RHB}\},
\end{equation}
for every faithful state $\rho$. In particular,
$\alpha_{\rm MLSI}(\cL_*^a)\ge c_\beta/\log n$ at every fixed
positive temperature.
\end{theorem}

The MLSI bounds both mixing times in \eqref{eq:mixing-times}; a
superscript $a$ indicates the corresponding sampler.

\begin{corollary}[Entropy decay and rapid mixing]\label{cor:mixing}
For either sampler, all initial states $\rho$, and $t\ge0$,
\begin{equation}\label{eq:entropy-decay}
 D(e^{t\cL_*^a}\rho\Vert\sigma)
 \le e^{-c_\beta t/\log n}D(\rho\Vert\sigma).
\end{equation}
For $0<\epsilon<1$, the relative-entropy and trace-norm mixing times satisfy
\begin{equation}\label{eq:mixing}
 \max\{t_{{\rm mix},D}^a(\epsilon),t_{\rm mix}^a(\epsilon)\}
 \le C_\beta'\log n\,[\log n+\log\epsilon^{-1}].
\end{equation}
\end{corollary}

For completeness, when $n=1$ both samplers reduce to the reset generator
$\cL_*(\rho)=\sigma-\rho$. For faithful $\rho$,
$\EP(\rho)=D(\rho\Vert\sigma)+D(\sigma\Vert\rho)\ge D(\rho\Vert\sigma)$,
so this case has MLSI constant at least one.

%
%

\subsection{Structure of the proof: heat-bath dynamics}\label{sec:proof-overview}
We focus on the heat-bath sampler. Each update erases one site and
reconstructs it with the Petz recovery channel associated with $\sigma$
\cite{petz-sufficiency}. This connects thermalisation to recoverability
in quantum information, where small conditional mutual information
guarantees the existence of an accurate recovery channel
\cite{fawzi-renner2015}. Moreover, the updates
are quasi-local: they admit local channel approximations with errors
that decay exponentially with the buffer width
(Proposition~\ref{loc:heatbath-quasilocal}). Recovery and locality
are essential for the key entropy estimate below.

\paragraph{Quantities and scales.}
Fix a chain $\Lambda$ of size $n\ge2$ and its Gibbs state $\sigma=\sigma^\Lambda$.
All dynamics below are heat-bath. For $S\subseteq\Lambda$, define
\[
 D_S(\rho\Vert\sigma)=D(\rho\Vert\sigma)-D(\rho_{S^c}\Vert\sigma_{S^c}),
 \qquad S[b]=\{i\in\Lambda:\operatorname{dist}(i,S)\le b\}.
\]
 $D_S(\rho\Vert\sigma)$ is the conditional relative entropy in $S$ \cite{CapelLuciaPerezGarcia_ConditionalRelativeEntropy_2018}.
The nonnegative, additive quadratic energy is
\[
 \mathcal Q_T(\rho)=\sum_{i\in T}\mathcal Q_i(\rho)
 =-\bigl\langle\sqrt\rho,
 \Gamma_\sigma^{1/2}\cL_T^{\mathrm{HB}}
 \Gamma_\sigma^{-1/2}(\sqrt\rho)\bigr\rangle_{\rm HS},
 \qquad \cL_T^{\mathrm{HB}}=\sum_{i\in T}(\cP_i-\id).
\]
Here $T$ selects update centres and
$\EP(\rho)=-\Tr[\cL_*^{\mathrm{HB}}(\rho)(\log\rho-\log\sigma)]$
is the entropy production. We seek $D(\rho\Vert\sigma)\le C_\beta\log n\,\EP(\rho)$,
with constants independent of $n$.

Consider the cover $\mathcal G=\{A_i\}_i\cup\{B_j\}_j$ in
Figure~\ref{fig:summary-cover}, with $A=\bigcup_i A_i$, $B=\bigcup_j B_j$.
Choose $L\asymp\log n$, $b\asymp L$ so that $|S|\le6L$,
$|\mathcal G|\le2n$ and
$\nu:=\sup_i|\{S\in\mathcal G:i\in S[b]\}|\le C_\beta$.
For $n<4L$, take $A=\Lambda$, $B=\varnothing$ and $\mathcal G=\{\Lambda\}$.
We first split the entropy between $A$ and $B$, then split their
conditional entropies into the interval components. This gives
\begin{equation}\label{eq:road-qf}
 D(\rho\Vert\sigma)
 \le D_A(\rho\Vert\sigma)+D_B(\rho\Vert\sigma)+1
 \le\sum_{S\in\mathcal G}\left(D_S(\rho\Vert\sigma)
                              +\frac{C_\beta}{n}\right)+1.
\end{equation}
The $C_\beta/n$ bound per segment requires Araki's expansional
estimates \cite{araki1969,bcp} and effective-Hamiltonian
approximations \cite{kuwahara2025clustering} to control the logarithms
of Gibbs marginals (Section~\ref{stat:section}). The approximation
errors decay exponentially in $L$, giving this accuracy at $L\asymp\log n$.

\medskip
%
\noindent\begin{minipage}{\linewidth}
\makeatletter\def\@captype{figure}\makeatother
\centering
\setlength{\abovecaptionskip}{4pt}
\setlength{\belowcaptionskip}{0pt}
\vspace{0.5cm}
\begin{tikzpicture}[x=.55cm,y=1cm,font=\small]
  \definecolor{qfball}{RGB}{102,0,102}
  \definecolor{qfcyan}{RGB}{50,165,210}
  \definecolor{qfpink}{RGB}{175,55,160}
  \foreach \a/\z/\lab in {1.5/7.5/A_1,9.5/15.5/{S=A_2},17.5/23.5/A_3}{
    \fill[qfcyan,opacity=.28] (\a,.82) rectangle (\z,1.22);
    \draw[thick,qfcyan!80!black,decorate,
      decoration={brace,amplitude=3pt}]
      (\a,1.32)--(\z,1.32);
    \node[text=qfcyan!65!black] at ({(\a+\z)/2},1.65) {$\lab$};
  }
  \foreach \a/\z/\lab in {-.5/3.5/B_1,5.5/11.5/B_2,13.5/19.5/B_3,21.5/23.5/B_4}{
    \fill[qfpink,opacity=.26] (\a,.82) rectangle (\z,1.22);
    \draw[thick,qfpink!80!black,decorate,
      decoration={brace,amplitude=3pt,mirror}]
      (\a,.72)--(\z,.72);
    \node[text=qfpink!70!black] at ({(\a+\z)/2},.4) {$\lab$};
  }
  \foreach \x/\lab in {.5/X_0,8.5/X_1,16.5/X_2}
    \node[text=black!65] at (\x,1.65) {$\lab$};
  \foreach \x/\lab in {4.5/Y_1,12.5/Y_2,20.5/Y_3}
    \node[text=black!65] at (\x,.4) {$\lab$};
  \draw[black!30] (0,1.02)--(23,1.02);
  \foreach \x in {0,...,23}
    \shade[shading=ball,ball color=qfball] (\x,1.02) circle (.12cm);
  \node[anchor=east] at (-1,1.02) {$\Lambda$};
  \node[anchor=east,text=qfcyan!65!black] at (-1,1.65) {$A$};
  \node[anchor=east,text=qfpink!70!black] at (-1,.4) {$B$};
  \foreach \x in {9.5,15.5}
    \draw[dashed,qfcyan!70!black] (\x,.17)--(\x,-.15);
  \node[text=qfcyan!70!black] at (12.5,-.01) {$S[b]$};
  \fill[black!5] (-.5,-.55) rectangle (23.5,-.15);
  \fill[qfcyan,opacity=.16] (6.5,-.55) rectangle (18.5,-.15);
  \fill[qfcyan,opacity=.30] (9.5,-.55) rectangle (15.5,-.15);
  \draw[black!30] (0,-.35)--(23,-.35);
  \foreach \x in {0,...,23}
    \shade[shading=ball,ball color=qfball] (\x,-.35) circle (.12cm);
  \node[anchor=east] at (-1,-.35) {$\Lambda$};
  \foreach \a/\z/\lab in {6.5/9.5/b,9.5/15.5/S,15.5/18.5/b}{
    \draw[thick,qfcyan!70!black,decorate,
      decoration={brace,amplitude=3pt,mirror}]
      (\a,-.65)--(\z,-.65);
    \node at ({(\a+\z)/2},-.97) {$\lab$};
  }
\end{tikzpicture}
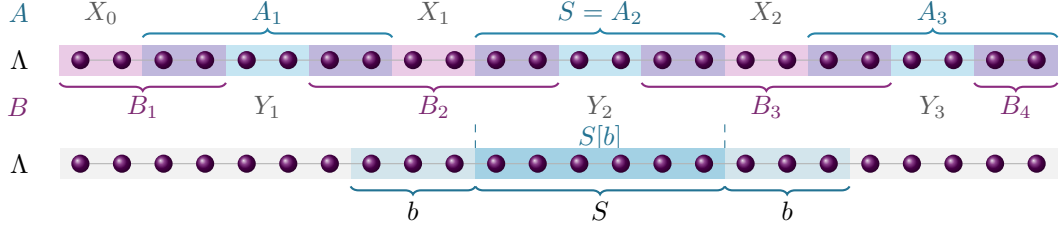
\caption{Overlapping blocks $A=\bigcup_i A_i$, $B=\bigcup_j B_j$
and a buffered core. The gaps $X_i,Y_j$ form
$\Lambda\setminus A,\Lambda\setminus B$. Each site lies in at most
$C_\beta$ enlarged intervals $S[b]$, independently of $n$.
The intervals specify update centres;
the maps have decaying tails.}
\label{fig:summary-cover}
\vspace{0.5cm}
\end{minipage}\par

 
%
%

Following the block-entropy strategy in
\cite{BardetEtAlEntropyDecayDavies}, one would next compare block
entropies with local entropy production. Even with such comparisons,
the additive terms in \eqref{eq:road-qf} survive summation, yielding
only a \emph{weak MLSI}, of the form
$D(\rho\Vert\sigma)\le A_n\EP(\rho)+B_\beta$.
The defect $B_\beta$ stays positive as $\rho$ approaches $\sigma$,
so this bound alone only guarantees decay to a nonzero entropy level.
We instead control entropy by the quadratic energy
$\mathcal Q_\Lambda(\rho)$. This intermediate bound allows
the global gap and weighted norm interpolation to remove the defect
and yield MLSI.

\paragraph{The buffered block estimate: recovery and gap.}
The canonical purified Hamiltonian introduced in
\cite{LuciaPerezGarciaPerezHernandez2025Gap}
yields a uniform conditional GNS gap (Theorem~\ref{gap:conditional}).
To control the KMS energy $\mathcal Q_{S[b]}(\rho)$, we transfer this
estimate by averaging over modular twists (Lemma~\ref{lem:average}).
A unit Hilbert--Schmidt amplitude $v$ represents the state $vv^\dagger$.
The twisted exterior spaces in \eqref{eq:exterior-spaces} give
\[
 \mathcal V_{S[b]}(t)=\{\Delta_\sigma^{-it}(O)\sigma^{1/2}:
                         O\in\mathcal A_{(S[b])^c}\},
 \qquad x_t(\rho)=\operatorname{dist}_2(\sqrt\rho,\mathcal V_{S[b]}(t))^2,
\]
where $\mathcal A_{(S[b])^c}$ consists of observables outside $S[b]$.
Here $t$ is a modular parameter and $x_t(\rho)$ measures the component
outside the twisted exterior space. Recovery and modular averaging give
\begin{align*}
 D_S(\rho\Vert\sigma)
 &\le C_\beta Lx_t(\rho)+C_\beta/n \quad (|t|\le1)
 &&\text{superposition, Lem.~\ref{lem:superposition}},\\
 \int_{\mathbb R} f(t)x_t(\rho)\,dt
 &\le C_\beta\mathcal Q_{S[b]}(\rho),
 \qquad f(t)=2\operatorname{sech}(2\pi t)
 &&\text{averaging, Lem.~\ref{lem:average}}.
\end{align*}
For the first estimate, Theorem~\ref{loc-exterior} shows that the Petz
update almost fixes exterior changes. Lemma~\ref{lem:kernel-entropy}
and Corollary~\ref{main:shielding} give
$D_S(vv^\dagger\Vert\sigma)\le C_\beta/n$ for every unit
$v\in\mathcal V_{S[b]}(t)$.
Lemma~\ref{lem:superposition} bounds the cost of the remaining
component by $C_\beta Lx_t(\rho)+C_\beta/n$.
The second estimate controls averaged exterior distances; it yields
a genuine KMS spectral gap when the update region is the whole chain.
Since $\int_{-1}^1 f(t)\,dt$ is a fixed positive constant, integrating
the first estimate and using the second proves
\begin{equation}\label{eq:road-buffer}
 \boxed{D_S(\rho\Vert\sigma)\le C_\beta L\,\mathcal Q_{S[b]}(\rho)+C_\beta/n.}
 \qquad\text{(Theorem~\ref{ent:buffered})}
\end{equation}
The logarithmic coefficient comes from the block size and the
$1/n$ accuracy needed to sum over all blocks.

\paragraph{From blocks to the whole chain.}
Combining quasi-factorisation (Theorem~\ref{stat:weak-at})
with the energy estimate \eqref{eq:road-buffer} gives
\begin{align}
 D(\rho\Vert\sigma)
 &\le C_\beta L\sum_{S\in\mathcal G}\mathcal Q_{S[b]}(\rho)
            +\frac{2C_\beta|\mathcal G|}{n}+1
 &&\text{energy estimate, Thm.~\ref{ent:buffered}}
 \nonumber\\[3pt]
 &\le C_\beta'\log n\,\mathcal Q_\Lambda(\rho)+B_\beta
 &&\text{overlap, Prop.~\ref{global:defective}}.
 \label{eq:road-defect}
\end{align}
Here, we are using that by \eqref{global:overlap},
$\sum_S\mathcal Q_{S[b]}(\rho)\le\nu\mathcal Q_\Lambda(\rho)$;
$|\mathcal G|\le2n$ makes $B_\beta$ independent of $n$.

\paragraph{Remove the bounded error.}
Proposition~\ref{pre:properties} gives primitivity and KMS symmetry;
Lemma~\ref{lem:average} gives the global KMS gap $\gamma_\beta>0$
in \eqref{eq:kms-gap}. Following the norm-to-entropy method in
\cite{norms}, Theorem~\ref{main:tightening} gives the implication
below for $C>0$ and $B\ge0$, with both bounds valid for every state $\rho$:
\[
 D(\rho\Vert\sigma)\le C\mathcal Q_\Lambda(\rho)+B
 \quad\Longrightarrow\quad
 D(\rho\Vert\sigma)\le
 \left(\frac C2+\frac{2+B}{2\gamma_\beta}\right)\EP(\rho).
\]
Section~\ref{app:tightening} proves this step. Section~\ref{sec:conclusion} completes the proof of
Theorem~\ref{thm:main}: taking $C=C_\beta'\log n$ and $B=B_\beta$ from \eqref{eq:road-defect} gives
\begin{equation}\label{eq:road-end}
 D(\rho\Vert\sigma)\le C_\beta''\log n\,\EP(\rho),
 \qquad
 \boxed{\alpha_{\rm MLSI}(\cL_*^{\mathrm{HB}})
                 \ge\frac{c_\beta}{\log n}.}
\end{equation}
 
\section{Preliminaries and properties of the dynamics}\label{sec:prelim}
\subsection{Update maps and weighted inner products}\label{pre:maps}
We use the model, update maps, and normalisation of
Section~\ref{sec:model}. Observables on a region $T$ are identified
with their identity extensions in
$\mathcal A_T=\mathcal B(\mathcal H_T)\otimes\one_{T^c}$.
For a map $Q$ on observables, its trace adjoint satisfies
$\Tr[Q_*(\rho)X]=\Tr[\rho Q(X)]$.
The observable-picture heat-bath update, adjoint to \eqref{eq:petz}, is
\begin{equation}\label{eq:single-site-updates}
 \cP_i(X)
   =\one_i\otimes\left(
      \sigma_{i^c}^{-1/2}\,
      \Tr_i[\sigma^{1/2}X\sigma^{1/2}]\,
      \sigma_{i^c}^{-1/2}\right).
\end{equation}
Tensor factors are placed in their original order. The map $\cP_i$
is the \emph{Accardi--Cecchini generalised conditional expectation}
associated with $\sigma$ and $\mathcal A_{i^c}$
\cite{accardi-cecchini1982}; see also
\cite[Section 2.4.1]{BardetCapelRouze2021NoncommutingAT}.
It is completely positive, unital, and $\sigma$-preserving, with
range contained in $\mathcal A_{i^c}$, but it need not be idempotent.

In \eqref{eq:regularisation}, \emph{conditional expectation} means a
unital completely positive projection $\mathbb{E}$ onto a subalgebra
$\mathcal N$, satisfying
$\mathbb{E}(AXB)=A\mathbb{E}(X)B$ for $A,B\in\mathcal N$.
It preserves $\sigma$ when
$\Tr[\sigma\mathbb{E}(X)]=\Tr[\sigma X]$; see
\cite[Section 2.3]{BardetCapelRouze2021NoncommutingAT}.
For $\mathbb{E}_S$, the algebra is
$\mathcal N_S=\{X:\cP_S(X)=X\}\subseteq\mathcal A_{S^c}$,
as proved in Proposition~\ref{pre:properties}.

The operator norm is $\|\cdot\|_\infty$, while $\|\cdot\|_2$ denotes
the Hilbert--Schmidt norm. The KMS and GNS inner products are
\begin{equation}\label{eq:inner-products}
 \begin{aligned}
 \langle X,Y\rangle_{\mathrm{KMS},\sigma}
     &=\Tr[X^\dagger\Gamma_\sigma(Y)],\\
 \langle X,Y\rangle_{\mathrm{GNS},\sigma}
     &=\Tr[\sigma X^\dagger Y].
 \end{aligned}
\end{equation}
Their induced norms are denoted by $\|\cdot\|_{r,\sigma}$,
where $r\in\{\mathrm{GNS},\mathrm{KMS}\}$; the subscript $\sigma$
on an inner product is omitted when the reference is clear.
KMS symmetry and GNS symmetry mean self-adjointness in the respective
inner products; see \cite{fagnola-umanita2010,carlen-maas2017}.
In particular, $\|X\|_{\mathrm{KMS},\sigma}
=\|\Gamma_\sigma^{1/2}(X)\|_2$ is the weighted $L_2$ norm used in
\cite{olkiewicz-zegarlinski1999,kastoryano-temme2013}.

For any set of update centres $S\subseteq\Lambda$, define
\begin{equation}\label{pre:restricted-generators}
 \cL_S^{\mathrm{HB}}=\sum_{i\in S}(\cP_i-\id),\qquad
 \cL_S^{\mathrm{RHB}}=\sum_{i\in S}(\mathbb{E}_i-\id),\qquad
 \cL^a=\cL_\Lambda^a.
\end{equation}
Here $S$ selects update centres, with $\sigma$ fixed; the maps may
act beyond $S$. The block update $\cP_S-\id$ should be distinguished
from this sum. An omitted sampler superscript means $a=\mathrm{HB}$,
unless an arbitrary generator has explicitly been fixed.

We also regard these generators as families on increasing volumes.
Fix a finite-range interaction $\Phi$ on $\mathbb Z$ satisfying
\eqref{eq:interaction-bound}, and take increasing finite intervals
$\Lambda_m\uparrow\mathbb Z$. For each $m$, set
$\sigma^{(m)}:=e^{-\beta H_{\Lambda_m}}/Z_{\Lambda_m}(\beta)$
and use this Gibbs state in \eqref{eq:generators} to define the two families
\[
 \bigl\{\cL_{\Lambda_m,*}^{a}\bigr\}_{m\ge1},
 \qquad a\in\{\mathrm{HB},\mathrm{RHB}\}.
\]
All updates for a given $m$ use $\sigma^{(m)}$, so the update
maps depend on the ambient interval.

\subsection{Variance and Dirichlet forms}\label{pre:quadratic-objects}
Write $\langle X\rangle_\sigma=\Tr[\sigma X]$ and
$E_\sigma(X)=\langle X\rangle_\sigma\one$.
The \emph{variance} in either weighted inner product is
\begin{equation}\label{pre:variance}
 \operatorname{Var}_\sigma^r(X)
 :=\|X-E_\sigma(X)\|_{r,\sigma}^2
 =\|X\|_{r,\sigma}^2-|\Tr[\sigma X]|^2,
 \qquad r\in\{\mathrm{GNS},\mathrm{KMS}\}.
\end{equation}
This is the usual squared distance from the scalar observables;
the KMS convention appears, for example, in
\cite[Section II]{kastoryano-temme2013}.
We abbreviate $\operatorname{Var}_\sigma
=\operatorname{Var}_\sigma^{\mathrm{GNS}}$.

The spatial \emph{conditional variance} used in the gap argument is
\begin{equation}\label{gap:variance}
 \operatorname{Var}_S(X)
 :=\inf_{O\in\mathcal A_{S^c}}
          \|X-O\|_{\mathrm{GNS},\sigma}^2.
\end{equation}
Thus $\operatorname{Var}_\Lambda(X)=\operatorname{Var}_\sigma(X)$,
and $\operatorname{Var}_i=\operatorname{Var}_{\{i\}}$.
Under the isometry $X\mapsto X\sigma^{1/2}$, this is the squared
distance to $\mathcal A_{S^c}\sigma^{1/2}$, the subspace used in the
canonical purified-Hamiltonian construction of
\cite[Section 1.2]{LuciaPerezGarciaPerezHernandez2025Gap}.
The minimisation is over the full exterior algebra. Since
$\mathcal N_S$ may be strictly smaller than $\mathcal A_{S^c}$,
one cannot in general replace \eqref{gap:variance} by
$\|X-\mathbb{E}_S(X)\|_{\mathrm{GNS},\sigma}^2$.

For $a\in\{\mathrm{HB},\mathrm{RHB}\}$, the \emph{Dirichlet form}
and its \emph{conditional} version on $S$ are the quadratic forms
\begin{equation}\label{pre:dirichlet-forms}
 \begin{aligned}
 \mathcal E_r^a(X)
   &:=-\operatorname{Re}\langle X,\cL^a(X)\rangle_{r,\sigma},\\
 \mathcal E_{S,r}^a(X)
   &:=-\operatorname{Re}\langle X,\cL_S^a(X)\rangle_{r,\sigma},
 \qquad r\in\{\mathrm{GNS},\mathrm{KMS}\}.
 \end{aligned}
\end{equation}
These are the quadratic dissipation forms associated with the
weighted norms \cite{olkiewicz-zegarlinski1999,kastoryano-temme2013,
carlen-maas2017}. They are nonnegative and additive in the update
centres:
$\mathcal E_{S,r}^a=\sum_{i\in S}\mathcal E_{i,r}^a$ and
$\mathcal E_r^a=\mathcal E_{\Lambda,r}^a$.
For the conditional-gap estimates we use the abbreviations
\begin{equation}\label{gap:energy}
 \mathcal E_S^a:=\mathcal E_{S,\mathrm{GNS}}^a,\qquad
 \mathcal E^a:=\mathcal E_{\mathrm{GNS}}^a,\qquad
 \mathcal E_S:=\mathcal E_S^{\mathrm{HB}},\qquad
 \mathcal E:=\mathcal E^{\mathrm{HB}}.
\end{equation}
The real part in the GNS form is needed because the heat-bath
generator is only assumed KMS symmetric.

For a state $\rho$, the \emph{quadratic energy} used in the entropy
argument is the KMS Dirichlet form evaluated at its weighted
square-root amplitude:
\begin{equation}\label{eq:quadratic}
 \begin{aligned}
 \mathcal Q_S^a(\rho)
 &:=\mathcal E_{S,\mathrm{KMS}}^a
                 (\Gamma_\sigma^{-1/2}(\sqrt\rho))\\
 &=-\ip{\sqrt\rho}
       {\Gamma_\sigma^{1/2}\cL_S^a\Gamma_\sigma^{-1/2}(\sqrt\rho)}.
 \end{aligned}
\end{equation}
In particular, $\mathcal Q_S^a=\sum_{i\in S}\mathcal Q_i^a$,
$\mathcal Q^a=\mathcal Q_\Lambda^a$, and $\mathcal Q_S$
denotes the heat-bath energy. This is the same weighted
$L_2$ form used in the norm-to-entropy argument of \cite{norms}.

\subsection{Relative entropy and entropy production}\label{pre:entropy}
For the relative entropy introduced in Section~\ref{sec:introduction}, we
use the convention $D(\rho\Vert\tau)=+\infty$ unless
$\supp\rho\subseteq\supp\tau$. Our Gibbs reference $\sigma$ is faithful.
The spatial \emph{conditional relative entropy} in a region $S$ is
\begin{equation}\label{eq:conditional-entropy}
 D_S(\rho\Vert\sigma)
 :=D(\rho\Vert\sigma)-D(\rho_{S^c}\Vert\sigma_{S^c}),
\end{equation}
following \cite{CapelLuciaPerezGarcia_ConditionalRelativeEntropy_2018}.
It is nonnegative by data processing and satisfies
$D_\Lambda(\rho\Vert\sigma)=D(\rho\Vert\sigma)$.
The fixed-algebra relative entropy $D(\rho\Vert \mathbb{E}_{S*}\rho)$ is a
different quantity; their comparison is given in
Proposition~\ref{cmlsi:local}.

For a faithful state $\rho$, the \emph{entropy production} of
$\cL^a$ is \cite{spohn1978,kastoryano-temme2013}
\begin{equation}\label{eq:ep}
 \EP^a(\rho)
 :=-\Tr[\cL_*^a(\rho)(\log\rho-\log\sigma)]
 =-\left.\frac{d}{dt}D(e^{t\cL_*^a}\rho\Vert\sigma)\right|_{t=0}.
\end{equation}
Its \emph{conditional entropy production} on $S$ is
\begin{equation}\label{pre:conditional-ep}
 \begin{aligned}
 \EP_S^a(\rho)
 &:=-\Tr[\cL_{S*}^a(\rho)(\log\rho-\log\sigma)]\\
 &=-\left.\frac{d}{dt}
          D(e^{t\cL_{S*}^a}\rho\Vert\sigma)\right|_{t=0}.
 \end{aligned}
\end{equation}
This is the local entropy-production convention used for heat-bath
dynamics in \cite{BardetCapelLuciaPerezGarciaRouze_HeatBath1D_2021}.
Here ``conditional'' refers to restricting the update centres to
$S$, while keeping the global reference $\sigma$. It is the
dissipation of the global relative entropy under that restricted
evolution, rather than the derivative of the spatial quantity $D_S$.
Since every site update preserves $\sigma$, data processing gives
\begin{equation}\label{pre:ep-additivity}
 \EP_S^a(\rho)=\sum_{i\in S}\EP_i^a(\rho)\ge0,
 \qquad \EP^a=\EP_\Lambda^a.
\end{equation}
We write $\EP_S=\EP_S^{\mathrm{HB}}$ and $\EP=\EP^{\mathrm{HB}}$.
The same trace formula defines $\EP_{\cG}$ for any explicitly
specified generator $\cG$ and its stated faithful invariant reference,
including generators tensored with an untouched memory.
Formulas with logarithms are first used for faithful states;
singular states are treated by faithful approximation where needed.

\subsection{Stationarity and symmetry}\label{pre:stationarity}
We now verify the properties of the update maps and the comparison
of quadratic forms used to transfer estimates from the heat-bath
sampler to the regularised heat-bath sampler.

\begin{proposition}[Stationarity, symmetry, and conditional expectations]
\label{pre:properties}
The maps $\cP_{S*}$ are channels preserving $\sigma$. Their adjoints
are KMS self-adjoint with spectrum in $[0,1]$. The limit $\mathbb{E}_S$ is
the $\sigma$-preserving conditional expectation onto an algebra
$\mathcal N_S\subseteq\mathcal A_{S^c}$ and is the KMS orthogonal
projection onto that algebra. In particular, both generators in
\eqref{eq:generators} are primitive, and
\begin{equation}\label{eq:energy-comparison}
 0\le-\Gamma_\sigma^{1/2}\cL_T^{\mathrm{HB}}\Gamma_\sigma^{-1/2}
   \le-\Gamma_\sigma^{1/2}\cL_T^{\mathrm{RHB}}\Gamma_\sigma^{-1/2}.
\end{equation}
The inequalities are in the operator order on Hilbert--Schmidt space.
\end{proposition}
\begin{proof}
The formula for $\cP_{S*}$ is a composition of completely positive
maps. Writing $C=S^c$, cyclicity and partial trace give
\[
 \Tr[\cP_{S*}(\rho)]
 =\Tr[\sigma_C\sigma_C^{-1/2}\rho_C\sigma_C^{-1/2}]
 =\Tr[\rho],\qquad \cP_{S*}(\sigma)=\sigma.
\]
Thus it is a channel preserving $\sigma$. Its adjoint is
\[
 \cP_S(X)=\one_S\otimes\Gamma_{\sigma_{S^c}}^{-1}
             \bigl(\Tr_S[\Gamma_\sigma(X)]\bigr).
\]
Putting $T_X=\Tr_S[\Gamma_\sigma(X)]$, we have
\[
 \langle X,\cP_S X\rangle_{\mathrm{KMS}}
 =\Tr[T_X^\dagger\Gamma_{\sigma_{S^c}}^{-1}(T_X)]\ge0.
\]
Polarisation gives KMS self-adjointness, and unital positivity bounds
the spectral radius by one. Thus the spectrum lies in $[0,1]$, and
the powers converge to the orthogonal projection onto the eigenvalue-one
space.

If $X$ is fixed, Schwarz gives
$\cP_S(X^\dagger X)-X^\dagger X\ge0$. Its expectation in the faithful
state $\sigma$ is zero, so equality holds. The same applies to
$XX^\dagger$. The equality case of Schwarz places $X$ in the
multiplicative domain of $\cP_S$: for every $Y$,
$\cP_S(XY)=X\cP_S(Y)$ and $\cP_S(YX)=\cP_S(Y)X$.
The fixed space is consequently an algebra, and the limiting map is
bimodular over it. Every fixed point belongs to $\mathcal A_{S^c}$
by the formula for $\cP_S$.

For a single site, the intersection of these algebras over all sites
is scalar. Positivity of the individual KMS forms shows that the
kernel of either sum of forms is precisely this intersection. Thus
both generators are primitive with invariant state $\sigma$.
Finally $1-\lambda\le1-\mathbf1_{\{1\}}(\lambda)$ for
$0\le\lambda\le1$. Applying this to each $\cP_i$ proves
\eqref{eq:energy-comparison}.
\end{proof}

%
\subsection{Complete inequalities for individual updates}
\label{cmlsi:section}

We distinguish the complete estimates used for individual updates from
the ordinary MLSI proved for the sum of the single-site generators.
Let $\mathbb{E}$ be a state-preserving conditional expectation and let $M$ be an
arbitrary finite-dimensional memory. Write
\[
 \tau=(\mathbb{E}_*\otimes\mathrm{id}_M)(\rho),\qquad
 D_{\mathbb{E}}(\rho)=D(\rho\Vert\tau).
\]
For a generator $\mathcal L$ with limiting expectation $\mathbb{E}$, CMLSI with rate $\alpha$
means
\begin{equation}\label{cmlsi:definition}
 \operatorname{EP}_{\mathcal L\otimes\mathrm{id}_M}(\rho)
 \ge\alpha D(\rho\Vert\tau)
 \qquad\text{for every }M\text{ and every faithful }\rho,
\end{equation}
with a constant independent of the memory. For a primitive generator,
$\tau=\sigma\otimes\rho_M$. We use the rate convention in
\eqref{cmlsi:definition}, without an additional factor of two.

\begin{proposition}[Complete entropy estimates for an update]
\label{cmlsi:local}
Every regularised heat-bath update generator $\mathbb{E}_S-\mathrm{id}$ satisfies CMLSI with
rate at least one. More precisely, for
$\tau=(\mathbb{E}_{S*}\otimes\mathrm{id}_M)(\rho)$,
\begin{equation}\label{cmlsi:rhb}
 \operatorname{EP}_{(\mathbb{E}_S-\mathrm{id})\otimes\mathrm{id}_M}(\rho)
 =D(\rho\Vert\tau)+D(\tau\Vert\rho)
 \ge D(\rho\Vert\tau).
\end{equation}
The usual spatial conditional relative entropy is bounded by this
fixed-algebra relative entropy:
\begin{equation}\label{cmlsi:conditional-comparison}
 \begin{aligned}
 D_S(\rho\Vert\sigma\otimes\rho_M)
 &:=D(\rho\Vert\sigma\otimes\rho_M)
     -D(\rho_{S^cM}\Vert\sigma_{S^c}\otimes\rho_M)\\
 &\le D(\rho\Vert\tau).
 \end{aligned}
\end{equation}
Here the complement of $S$ includes the memory $M$.
For both the heat-bath and regularised heat-bath updates, $Q_S=\mathcal P_S$
or $Q_S=\mathbb{E}_S$, one also has
\begin{equation}\label{cmlsi:spatial}
 \operatorname{EP}_{(Q_S-\mathrm{id})\otimes\mathrm{id}_M}(\rho)
 \ge D_S(\rho\Vert\sigma\otimes\rho_M).
\end{equation}
\end{proposition}
\begin{proof}
Fix a faithful reference $\Omega=\sigma\otimes\nu_M$, with $\nu_M>0$.
In the finite-dimensional decomposition of the conditional-expectation
algebra, the invariant states have blocks
$\omega_{j,L}\otimes\theta_{j,R}$, where the right factors
$\theta_{j,R}$ are fixed. Hence
$\log\tau-\log\Omega$ belongs to the expectation algebra, and
$\rho-\tau$ has zero trace pairing with it. This is the conditional
expectation chain rule; see \cite[Appendix A, equation (A.1)]{BardetCapelRouze2021NoncommutingAT}.
Consequently
\[
 \Tr\bigl[(\rho-\tau)(\log\rho-\log\Omega)\bigr]
 =\Tr\bigl[(\rho-\tau)(\log\rho-\log\tau)\bigr],
\]
which proves \eqref{cmlsi:rhb}. Equivalently, the semigroup is explicit:
\[
 e^{t(\mathbb{E}_{S*}-\mathrm{id})}\otimes\mathrm{id}_M(\rho)
 =e^{-t}\rho+(1-e^{-t})\tau.
\]
Convexity of relative entropy gives contraction of $D_{\mathbb{E}}$ by $e^{-t}$,
also for singular states on the support of their memory marginal.

For any channel $Q_*$ preserving $\Omega$, direct expansion gives
\[
 \operatorname{EP}_{Q-\mathrm{id}}(\rho)
 =D(\rho\Vert\Omega)-D(Q_*\rho\Vert\Omega)
                              +D(Q_*\rho\Vert\rho).
\]
For $Q_* = \mathcal P_{S*}\otimes\mathrm{id}_M$, data processing
under the Petz recovery channel yields
\[
 D(Q_*\rho\Vert\Omega)
 \le D(\rho_{S^cM}\Vert\sigma_{S^c}\otimes\nu_M).
\]
The same inequality holds for every positive power of this channel,
and then for its limit $\mathbb{E}_{S*}\otimes\mathrm{id}_M$ by continuity.
For that limit the same trace-pairing identity gives the chain rule
\[
 D(\rho\Vert\Omega)=D(\rho\Vert\tau)+D(\tau\Vert\Omega).
\]
Combining it with the preceding data-processing bound proves
\eqref{cmlsi:conditional-comparison}.
Dropping the last nonnegative term proves \eqref{cmlsi:spatial},
since the difference of the two relative entropies is independent
of $\nu_M$. Faithful approximation gives the usual extended-value
statement for singular $\rho$.
\end{proof}

\begin{remark}[Local heat-bath CMLSI and the conditional gap] 
A uniform local CMLSI relative to the exact fixed algebra of
$\mathcal P_i$ cannot replace the conditional-gap argument. For example,
on two qubits let $\sigma_\varepsilon\propto
e^{-\varepsilon Z_1Z_2}$, where $X_i,Z_i$ are Pauli matrices and
$\varepsilon>0$. Its second-site marginal is $\one/2$, and direct
evaluation of the Petz formula gives
\[
 \mathcal P_1(\one\otimes X_2)
       =\operatorname{sech}(\varepsilon)(\one\otimes X_2),
 \qquad \mathbb{E}_1(\one\otimes X_2)=0.
\]
By KMS symmetry, $V=\Gamma_{\sigma_\varepsilon}(\one\otimes X_2)$
is the corresponding Schr\"odinger eigenvector. For small real $u$,
$\rho_u=\sigma_\varepsilon+uV$ is a faithful density matrix and
$\mathbb{E}_{1*}\rho_u=\sigma_\varepsilon$. The second-order expansion of
relative entropy at its faithful minimum gives
\[
 \lim_{u\to0}
 \frac{\EP_{\mathcal P_1-\mathrm{id}}(\rho_u)}
      {D(\rho_u\Vert \mathbb{E}_{1*}\rho_u)}
       =2[1-\operatorname{sech}(\varepsilon)]\longrightarrow0
       \quad(\varepsilon\downarrow0).
\]
Indeed the numerator is $(1-\operatorname{sech}\varepsilon)u$
times the derivative of $D(\rho_u\Vert\sigma_\varepsilon)$.
Thus even an ordinary local MLSI constant can vanish within a family
of uniformly bounded interactions. The ambient conditional gap and
the global sum of site updates are different statements.
\end{remark}

\begin{remark}[Scope of the final inequality]\label{cmlsi:scope}
The spatial comparison \eqref{cmlsi:spatial} for a single heat-bath update
is not, by itself, a CMLSI relative to its fixed-point algebra:
its right-hand side need not equal $D_{\mathbb{E}_S}(\rho)$. Likewise,
\eqref{cmlsi:rhb} for each regularised heat-bath update does not supply the
factorisation needed for a complete inequality for their sum.
The defect-removal argument used below proves ordinary MLSI.
It centres observables around the scalar expectation
$X\mapsto\Tr[\sigma X]\mathbf1$ and uses a primitive
spectral gap, as in \cite[Lemma 2.3 and Theorem 2.4]{norms}.
After tensoring with an untouched memory, the fixed algebra is
$\mathbf1\otimes\mathcal B(\mathcal H_M)$ and this scalar-centred
gap vanishes: every nonzero $\mathbf1\otimes Z$ with
$\Tr[\nu_MZ]=0$ is fixed. Thus that proof cannot be
amplified verbatim. A memory-uniform conditional defect-removal
argument would be needed to conclude the same global CMLSI rate;
that rate remains unproved in the complete setting. Global CMLSI
with weaker volume dependence follows separately from the GNS gap,
as shown in Corollary~\ref{cmlsi:global}. The distinction between
primitive hypercontractivity and the conditional norm setting is
discussed in \cite[Section 2.2]{norms}.
\end{remark}
 %
%
%
\definecolor{darkblue}{RGB}{0,0,153}
\definecolor{darkpurple}{RGB}{102,0,102}
\definecolor{midgreen}{RGB}{0,204,0}
\definecolor{midpink}{RGB}{153,0,153}

\section{Weak quasi-factorisation of the Gibbs state}
\label{stat:section}

The entropy decomposition follows the overlapping-cover strategy of
\cite{BardetEtAlEntropyDecayDavies}.  Differently from the commuting setting, the noncommuting case
introduces an additive error, which we control in the operator norm.
All states and marginals in this section refer to the same Gibbs state
$\sigma$ on $\Lambda=[1,n]$.  A superscript distinguishes a restricted
Gibbs state, $\sigma^V=e^{-\beta H_V}/\Tr[e^{-\beta H_V}]$, from a marginal
$\sigma_V=\Tr_{V^c}[\sigma]$.  Every marginal logarithm is lifted to the
full tensor product when added to other operators.

\begin{theorem}[Weak quasi-factorisation]
\label{stat:weak-at}
We fix the interaction parameters and $\beta>0$ as in
Section~\ref{sec:model}.  There are constants $L_\beta,B_\beta,a_\beta>0$
such that, for every integer $L\geq L_\beta$, we can choose a cover
$\mathcal G=\{A_i\}_i\cup\{B_j\}_j$, illustrated in
Figure~\ref{fig:summary-cover}, with the following properties.
If $n\geq4L$, every member is an interval of length between $L$ and $6L$,
each site belongs to at most two members, and
$|\mathcal G|\leq n/L+2$.  If $n<4L$, we take $\mathcal G=\{\Lambda\}$.
For every state $\rho$,
\begin{equation}
 D(\rho\Vert\sigma)
 \leq\sum_{S\in\mathcal G}D_S(\rho\Vert\sigma)
       +B_\beta n e^{-a_\beta L}.
 \label{stat:weak-at-equation}
\end{equation}
When $n<4L$, the choice $\mathcal G=\{\Lambda\}$ gives
$D_\Lambda(\rho\Vert\sigma)=D(\rho\Vert\sigma)$, so the additive
error can be omitted and the inequality becomes an equality.
In particular, a constant
$K_\beta$ can be chosen so that $L=\lceil K_\beta\log n\rceil$ makes
the error at most one for every $n\ge2$.
\end{theorem}

We first reduce the entropy inequality to logarithmic marginal errors.
We then establish two connected correlation estimates and combine them
in an overlapping cover.  The two estimates distinguish
a retained separator, controlled by conditional MCMI, from a traced-out
separator, controlled by a multiplicative comparison with a product state.

\subsection{Reduction to logarithmic marginal errors}

The conditional relative entropy as introduced in
\cite{CapelLuciaPerezGarcia_ConditionalRelativeEntropy_2018} is
\[
 D_S(\rho\Vert\sigma)=D(\rho\Vert\sigma)
                 -D(\rho_{S^c}\Vert\sigma_{S^c}).
\]
We use the matrix-valued conditional mutual information notation of
\cite[Definition~2.3]{CapelGondolfKochanowskiRouze2024MCMI}:
\begin{align}
 \mathbf H_\sigma(X:Z\mid Y)
 &=\log\sigma_{XYZ}+\log\sigma_Y
       -\log\sigma_{XY}-\log\sigma_{YZ},
 &H_\sigma(X:Z\mid Y)&=\norm{\mathbf H_\sigma(X:Z\mid Y)}.
 \label{stat:mcmi-definition}
\end{align}
The regions in this definition are disjoint and they need not cover the whole
chain.  We set $\sigma_\varnothing=1$.  The corresponding scalar
conditional mutual information is
$I_\rho(X:Z\mid Y)=S(\rho_{XY})+S(\rho_{YZ})-S(\rho_Y)-S(\rho_{XYZ})$.

Figure~\ref{stat:entropy-geometries} illustrates the two set
decompositions used in the following lemma.

\begin{lemma}[Entropy identities for a cover and for a union]
\label{stat:entropy-identities}
If $A\cup B=\Lambda$, we set $X=\Lambda\setminus A$ and
$Y=\Lambda\setminus B$.  Then
\begin{equation}
 D(\rho\Vert\sigma)
 \leq D_A(\rho\Vert\sigma)+D_B(\rho\Vert\sigma)
 +\norm{\log\sigma_{XY}-\log\sigma_{X}-\log\sigma_{Y}}.
 \label{stat:two-set}
\end{equation}
For pairwise disjoint nonempty sets $A_1,\ldots,A_m$, we set
$A=\bigcup_i A_i$ and $X=\Lambda\setminus A$, and we define
\begin{equation}
 K_\sigma(A):=(m-1)\log\sigma+\log\sigma_{X}
                    -\sum_{i=1}^m\log\sigma_{A_i^c}.
 \label{stat:multipartite-defect}
\end{equation}
Then
\begin{equation}
 D_A(\rho\Vert\sigma)\leq\sum_{i=1}^mD_{A_i}(\rho\Vert\sigma)
                       +\norm{K_\sigma(A)}.
 \label{stat:multipartite}
\end{equation}
\end{lemma}
\begin{proof}
For the first statement, data processing and entropy subadditivity give
\begin{align*}
 D(\rho\Vert\sigma)-D_A(\rho\Vert\sigma)-D_B(\rho\Vert\sigma)
 &\leq-D(\rho_{XY}\Vert\sigma_{XY})
        +D(\rho_{X}\Vert\sigma_{X})
        +D(\rho_{Y}\Vert\sigma_{Y})\\
 &=-I_\rho(X:Y)
    +\Tr\bigl[\rho\bigl(\log\sigma_{XY}-\log\sigma_{X}
                                  -\log\sigma_{Y}\bigr)\bigr],
\end{align*}
which implies \eqref{stat:two-set}.
For the second statement, we set $A_{1:j}=\bigcup_{i\leq j}A_i$ and
$C_j=\Lambda\setminus A_{1:j}$.  Expanding and telescoping gives
\begin{align}
 D_A(\rho\Vert\sigma)-\sum_iD_{A_i}(\rho\Vert\sigma)
 &=\Tr[\rho K_\sigma(A)]
     -\sum_{j=2}^m I_\rho(A_{1:j-1}:A_j\mid C_j).
 \label{stat:entropy-remainder}
\end{align}
By strong subadditivity, each conditional mutual information
$I_\rho(A_{1:j-1}:A_j\mid C_j)$ is nonnegative.  Dropping these
terms from the right-hand side of \eqref{stat:entropy-remainder}
and using that $\rho$ is a state gives
\[
 D_A(\rho\Vert\sigma)-\sum_{i=1}^m D_{A_i}(\rho\Vert\sigma)
 \leq \Tr[\rho K_\sigma(A)]
 \leq \|\rho\|_1\norm{K_\sigma(A)}
 =\norm{K_\sigma(A)},
\]
which is \eqref{stat:multipartite}.  For $m=1$, we have $A=A_1$
and $X=A_1^c$, so $K_\sigma(A)=0$ and the sum of conditional mutual
informations in \eqref{stat:entropy-remainder} is empty.  Thus
\eqref{stat:multipartite} is an equality in this case.
\end{proof}

\begin{remark}[Comparison with multiplicative quasi-factorisation]
\label{stat:multiplicative-alternative}
For the two-set estimate \eqref{stat:two-set}, we could instead use
the superadditivity result of
\cite[Theorem 1]{CapelLuciaPerezGarcia2018Superadditivity} and its
quasi-factorisation reformulation
\cite[Theorem 5]{CapelLuciaPerezGarcia_ConditionalRelativeEntropy_2018}.
With $X=A^c$ and $Y=B^c$, we set
\[
 h(\sigma_{XY})
 :=(\sigma_X\otimes\sigma_Y)^{-1/2}\sigma_{XY}
                    (\sigma_X\otimes\sigma_Y)^{-1/2}
                      -\one_{XY}.
\]
The superadditivity estimate reads
\[
 D(\rho_X\Vert\sigma_X)+D(\rho_Y\Vert\sigma_Y)
 \le\bigl(1+2\|h(\sigma_{XY})\|_\infty\bigr)
                  D(\rho_{XY}\Vert\sigma_{XY}).
\]
By data processing and the definition of $D_A,D_B$, this gives,
whenever $\|h(\sigma_{XY})\|_\infty<1/2$,
\begin{equation}\label{stat:two-set-multiplicative}
 D(\rho\Vert\sigma)
 \le\frac{D_A(\rho\Vert\sigma)+D_B(\rho\Vert\sigma)}
              {1-2\|h(\sigma_{XY})\|_\infty}.
\end{equation}
Here the defect enters multiplicatively. It measures the same
departure from a product state as the additive logarithmic defect
in \eqref{stat:two-set}. More precisely, a comparison
$e^{-\varepsilon}\sigma_X\otimes\sigma_Y\le\sigma_{XY}
\le e^{\varepsilon}\sigma_X\otimes\sigma_Y$ implies
\[
 \|\log\sigma_{XY}-\log\sigma_X-\log\sigma_Y\|_\infty
       \le\varepsilon,
 \qquad \|h(\sigma_{XY})\|_\infty\le e^\varepsilon-1.
\]
Thus the same product-comparison estimates control both defects,
with comparable bounds for small $\varepsilon$. We retain the
additive form \eqref{stat:two-set}, since the multipartite
decomposition and the buffered entropy--energy comparison below
introduce further additive errors in any case.
\end{remark}

%
\begin{figure}[htbp]
\centering
\begin{tikzpicture}[x=.55cm,y=.55cm,font=\small,
  region/.style={thick,decorate,decoration={brace,mirror,amplitude=4pt}}]
 \node[anchor=west] at (-.5,4.5) {(a) Overlapping cover: $A\cup B=\Lambda$};
 \foreach \a/\b/\col/\lab in
  {-.5/6.5/darkblue!30!cyan/X,6.5/15.5/midgreen/A\cap B,
   15.5/22.5/midpink/Y}{
  \fill[\col,opacity=.24] (\a,2.05) rectangle (\b,2.95);
  \node at ({(\a+\b)/2},3.55) {$\lab$};
 }
 \draw[darkpurple!35] (0,2.5)--(22,2.5);
 \foreach \x in {0,...,22}
  \shade[shading=ball,ball color=darkpurple] (\x,2.5) circle (.24);
 \draw[region,midpink] (6.5,1.8)--(22.5,1.8)
  node[midway,below=5pt] {$A$};
 \draw[region,darkblue!70!black] (-.5,.55)--(15.5,.55)
  node[midway,below=5pt] {$B$};

 \node[anchor=west] at (-.5,-1.5)
  {(b) Disjoint union: $A=A_1\sqcup A_2\sqcup A_3$};
 \foreach \a/\b/\lab in
  {-.5/2.5/X_0,6.5/9.5/X_1,13.5/16.5/X_2,20.5/22.5/X_3}{
  \fill[black!8] (\a,-3.95) rectangle (\b,-3.05);
  \node at ({(\a+\b)/2},-2.55) {$\lab$};
 }
 \foreach \a/\b/\lab in
  {2.5/6.5/A_1,9.5/13.5/A_2,16.5/20.5/A_3}{
  \fill[darkblue!30!cyan,opacity=.24] (\a,-3.95) rectangle (\b,-3.05);
  \node at ({(\a+\b)/2},-2.55) {$\lab$};
 }
 \draw[darkpurple!35] (0,-3.5)--(22,-3.5);
 \foreach \x in {0,...,22}
  \shade[shading=ball,ball color=darkpurple] (\x,-3.5) circle (.24);

 \foreach \a/\b in {-.5/2.5,6.5/9.5,13.5/22.5}{
  \fill[midgreen,opacity=.24] (\a,-5.85) rectangle (\b,-4.85);
  \node at ({(\a+\b)/2},-5.35) {$C_2$};
 }
 \fill[darkblue!30!cyan,opacity=.24] (2.5,-5.85) rectangle (6.5,-4.85);
 \node at (4.5,-5.35) {$P_1=A_1$};
 \fill[orange,opacity=.28] (9.5,-5.85) rectangle (13.5,-4.85);
 \node at (11.5,-5.35) {$A_2$};
 \foreach \x in {2.5,6.5,9.5,13.5}
  \draw[black!35,dashed] (\x,-4.1)--(\x,-4.75);
 \node at (11,-6.7)
  {$X=X_0\sqcup X_1\sqcup X_2\sqcup X_3,\qquad C_2=X\cup A_3$};
\end{tikzpicture}
\caption{The two decompositions in Lemma~\ref{stat:entropy-identities}.
(a) The complements are $X=B\setminus A$ and $Y=A\setminus B$;
tracing out the overlap $A\cap B$ leaves the marginal on $XY$ in
\eqref{stat:two-set}. (b) For $m=3$ and the merger $j=2$ in
\eqref{stat:entropy-remainder}, the systems are $P_1=A_1$, $A_2$,
and $C_2=X\cup A_3$. The green pieces form the generally disconnected
conditioning set $C_2$. The interval arrangements are illustrative:
the lemma holds for arbitrary sets.}
\label{stat:entropy-geometries}
\end{figure}
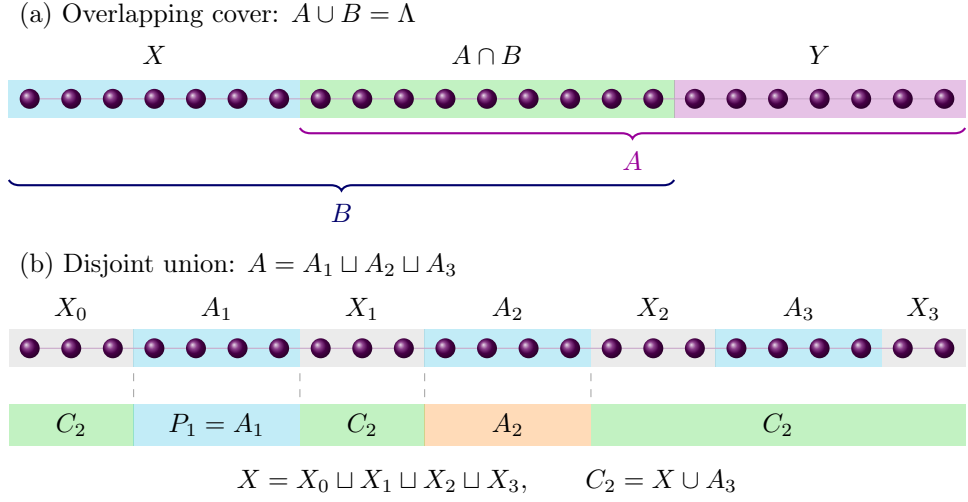
 
We first establish MCMI estimates for connected regions and then
extend them to the disconnected regions appearing in $K_\sigma(A)$
through a cut decomposition that keeps the original physical regions
fixed.

\subsection{Two connected correlation estimates}
\label{stat:connected-estimates}

The entropy identities above lead to two correlation estimates across a
separator: one with the separator retained as a conditioning region, and
one after it is traced out.  For the first, we follow the local
interpolation and effective-Hamiltonian approach of
\cite[Supplemental Material]{kuwahara2025clustering}.  The idea is to
cut the Hamiltonian at two points inside the separator and restore the
removed interactions approximately by operators localised near the cuts.
The five-piece decomposition below provides a buffer around each cut
and a central region separating the two corrections.  This produces an
approximating state with controlled MCMI; a relative operator bound then
transfers this control to the original Gibbs state, since it survives
partial traces and bounds the logarithms of the marginals.

\subsubsection{Retaining the separator}
We partition $\Lambda=X\sqcup Y\sqcup Z$ into three consecutive intervals
and write $\ell=|Y|$.  The following construction uses only the finite-range
interaction bounds.

We split $Y$ into five consecutive pieces $Y_1,\ldots,Y_5$ of lengths
differing by at most one, as illustrated in
Figure~\ref{stat:connected-five-piece}, and we set
\[
 W=e^{-\beta H_\Lambda},\qquad
 H_0=H_{XY_1}+H_{Y_2Y_3Y_4}+H_{Y_5Z}.
\]
Thus $H_0$ is obtained from $H_\Lambda$ by removing the interactions
across the cuts between $Y_1,Y_2$ and between $Y_4,Y_5$.
The local interpolation construction in
\cite[Supplemental Material]{kuwahara2025clustering} supplies
invertible operators $T_{12}$ and $T_{45}$, supported on $Y_1Y_2$ and
$Y_4Y_5$, and the positive operator
\begin{equation}
 \widetilde W=T_{45}T_{12}e^{-\beta H_0}T_{12}^*T_{45}^*,
 \qquad \widetilde\sigma=\frac{\widetilde W}{\Tr[\widetilde W]}.
 \label{stat:five-piece-approximation}
\end{equation}
The corresponding relative approximation and effective-Hamiltonian estimates in
\cite[Supplemental Material]{kuwahara2025clustering} give
\begin{align}
 H_{\widetilde\sigma}(X:Z\mid Y)
 &\leq a_\beta(\ell),
 &a_\beta(\ell)&=A_\beta(1+\ell)^{s_\beta}e^{-b_\beta \ell},
 \label{stat:approximate-connected-mcmi}\\
 \sup_{\psi\ne0}
 \frac{|\langle\psi,(W-\widetilde W)\psi\rangle|}
      {\langle\psi,W\psi\rangle}
 &\leq\delta_\beta(\ell),
 &\delta_\beta(\ell)&=e^{\xi_\beta-\ell/c_\beta}.
 \label{stat:global-relative-approximation}
\end{align}
Here $A_\beta,b_\beta,c_\beta>0$ and $s_\beta,\xi_\beta\geq0$ are
constants depending only on the inverse temperature $\beta$, the local
dimension $d$, the interaction range $r_0$, and the interaction bound
$J$.  The estimates hold for $\ell\geq\ell_\beta$, where the threshold
$\ell_\beta>0$ has the same parameter dependence.  In particular, these
constants are independent of the chain length and of the sizes and
positions of $X,Y,Z$; the dependence on the separator length $\ell$
is explicit in $a_\beta(\ell)$ and $\delta_\beta(\ell)$.
Both estimates concern the same approximant $\widetilde W$, and the relative error in
\eqref{stat:global-relative-approximation} controls the \emph{whole}
unnormalised operator $W$.

%
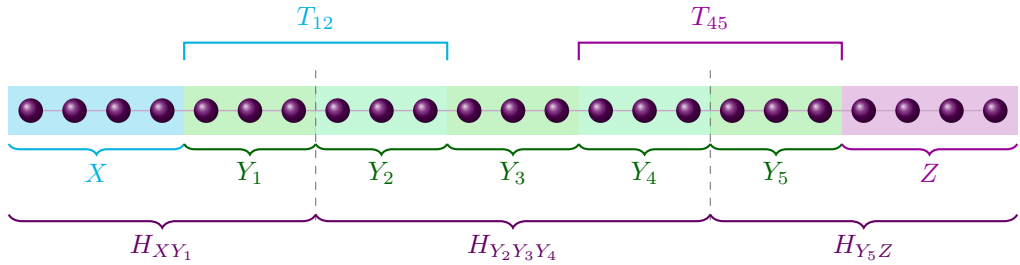
\begin{figure}[htbp]
\centering
\begin{tikzpicture}[x=.58cm,y=.58cm,font=\small]
 \fill[darkblue!30!cyan,opacity=.28] (-.5,.45) rectangle (3.5,1.55);
 \foreach \a/\b/\col in
  {3.5/6.5/midgreen,6.5/9.5/midgreen!65!cyan,
   9.5/12.5/midgreen,12.5/15.5/midgreen!65!cyan,15.5/18.5/midgreen}{
  \fill[\col,opacity=.23] (\a,.45) rectangle (\b,1.55);
 }
 \fill[midpink,opacity=.23] (18.5,.45) rectangle (22.5,1.55);
 \draw[darkpurple!35,line width=.5pt] (0,1)--(22,1);
 \foreach \x in {0,...,22}
  \shade[shading=ball,ball color=darkpurple] (\x,1) circle (.28);

 \foreach \a/\b/\lab/\col in
  {-.5/3.5/X/darkblue!30!cyan,3.5/6.5/Y_1/midgreen!50!black,
   6.5/9.5/Y_2/midgreen!50!black,9.5/12.5/Y_3/midgreen!50!black,
   12.5/15.5/Y_4/midgreen!50!black,15.5/18.5/Y_5/midgreen!50!black,
   18.5/22.5/Z/midpink}{
  \draw[thick,\col,decorate,decoration={brace,mirror,amplitude=4pt}]
   (\a,.24)--(\b,.24);
  \node[text=\col] at ({(\a+\b)/2},-.43) {$\lab$};
 }

 \foreach \a/\b/\lab/\col in
  {3.5/9.5/T_{12}/darkblue!30!cyan,12.5/18.5/T_{45}/midpink}{
  \draw[thick,\col] (\a,2.15)--(\a,2.55)--(\b,2.55)--(\b,2.15);
  \node[text=\col,above] at ({(\a+\b)/2},2.62) {$\lab$};
 }
 \foreach \x in {6.5,15.5}
  \draw[black!55,dashed] (\x,1.92)--(\x,-1.5);
 \foreach \a/\b/\lab in
  {-.5/6.5/H_{XY_1},6.5/15.5/H_{Y_2Y_3Y_4},15.5/22.5/H_{Y_5Z}}{
  \draw[darkpurple,thick,decorate,decoration={brace,mirror,amplitude=5pt}]
   (\a,-1.37)--(\b,-1.37);
  \node[text=darkpurple] at ({(\a+\b)/2},-2.12) {$\lab$};
 }
\end{tikzpicture}
\caption{The five-piece approximation on the spin chain. Each ball is a
site; the coloured bands and braces specify $X,Y_1,\ldots,Y_5,Z$.
The two dashed cuts split $H_0$ into the three Hamiltonians shown below.
The local operators $T_{12}$ and $T_{45}$ act on the indicated pairs of
separator pieces. Both estimates \eqref{stat:approximate-connected-mcmi}
and \eqref{stat:global-relative-approximation} concern this construction.}
\label{stat:connected-five-piece}
\end{figure}
 
\begin{lemma}[Connected MCMI from relative approximation]
\label{stat:connected-mcmi-relative}
In the setting above, if $\delta_\beta(\ell)<1$, then
\begin{equation}
 H_\sigma(X:Z\mid Y)
 \leq a_\beta(\ell)-4\log(1-\delta_\beta(\ell)).
 \label{stat:connected-operator-bound}
\end{equation}
Consequently, there are $C_\beta,c'_\beta,\ell'_\beta>0$, independent of
$|\Lambda|$, such that
\begin{equation}
 H_\sigma(X:Z\mid Y)\leq F_\beta(\ell):=C_\beta e^{-c'_\beta \ell}
 \qquad(\ell\geq \ell'_\beta).
 \label{stat:connected-mcmi-exponential}
\end{equation}
\end{lemma}
\begin{proof}
We compare the Gibbs state $\sigma$ of the original Hamiltonian
$H_\Lambda$ with the approximating state
$\widetilde\sigma=\widetilde W/\Tr[\widetilde W]$ constructed from the
truncated Hamiltonian $H_0$ and the local corrections $T_{12},T_{45}$.
The aim is to compare their MCMIs while keeping track of the additive
error introduced by replacing each logarithmic marginal.
We write $\delta=\delta_\beta(\ell)$.  Since
\eqref{stat:global-relative-approximation} holds for every vector,
\[
 (1-\delta)W\leq\widetilde W\leq(1+\delta)W.
\]
Since the partial trace is positive, writing
$W_V=\Tr_{V^c}[W]$ and $\widetilde W_V=\Tr_{V^c}[\widetilde W]$, we have
\begin{equation}
 (1-\delta)W_V\leq\widetilde W_V\leq(1+\delta)W_V
 \qquad\text{for every }V\subseteq\Lambda.
 \label{stat:marginal-relative-order}
\end{equation}

By operator monotonicity of the logarithm,
\[
 \log(1-\delta)I
 \leq\log\widetilde W_V-\log W_V
 \leq\log(1+\delta)I.
\]
Since $\log(1+\delta)\leq-\log(1-\delta)$, it follows that
\begin{equation}
 \|\log\widetilde W_V-\log W_V\|_\infty\leq-\log(1-\delta).
 \label{stat:relative-log-comparison}
\end{equation}

The normalising constants cancel in the four-logarithm MCMI expression.
The triangle inequality therefore gives
\begin{align*}
 H_\sigma(X:Z\mid Y)
 &\leq H_{\widetilde\sigma}(X:Z\mid Y)
   +\sum_{V\in\{XYZ,Y,XY,YZ\}}
     \|\log W_V-\log\widetilde W_V\|_\infty\\
 &\leq a_\beta(\ell)-4\log(1-\delta_\beta(\ell)).
\end{align*}
This proves \eqref{stat:connected-operator-bound}.

For sufficiently large $\ell$, $\delta_\beta(\ell)\leq1/2$ and
$-\log(1-\delta_\beta(\ell))\leq2\delta_\beta(\ell)$.
Absorbing the polynomial prefactor into the exponential proves
\eqref{stat:connected-mcmi-exponential}, with
$c'_\beta=\min\{b_\beta/2,1/c_\beta\}$.
\end{proof}

\subsubsection{Tracing out the separator}
The second estimate compares the remaining marginal with a product
state in operator order, which will also allow further partial traces.
It is an operator-order consequence of \cite[Proposition~8.1]{bcp},
whose exponential uniform clustering assumption is supplied for
inhomogeneous chains by
\cite[Corollary~I.2]{BergamaschiChen2025FastMixing}.
We give a derivation below from local indistinguishability and Araki
expansional bounds.

\begin{lemma}[Multiplicative comparison across a connected gap]
\label{stat:inhomogeneous-connected-mixing}
There are constants $c_\beta,a_\beta>0$ and $L_0$, independent of the volume,
such that every consecutive partition $\Lambda=X\sqcup Y\sqcup Z$ with
$|Y|\geq L\geq L_0$ satisfies
\begin{equation}
 e^{-g_\beta(L)}\,\sigma_X\otimes\sigma_Z
 \leq\sigma_{XZ}\leq
 e^{g_\beta(L)}\,\sigma_X\otimes\sigma_Z,
 \qquad g_\beta(L)=c_\beta e^{-a_\beta L}.
\label{stat:connected-mix-input}
\end{equation}
Consequently
$H_\sigma(X:Z\mid\varnothing)\leq g_\beta(L)$.
\end{lemma}

We first record the two known one-dimensional estimates used in the
proof.  Blocking the fixed interaction range and rescaling the strength
gives a nearest-neighbour chain.  The constants below depend only on
$\beta,d,r_0,J$ and may change under this fixed blocking; they are
uniform in all finite restrictions, their lengths and their positions.

\begin{lemma}[Exponential local indistinguishability]
\label{stat:local-indistinguishability}
There are constants $b_\beta,c_\beta>0$ such that, for consecutive
intervals $U,V,W$,
\begin{equation}
 \| (\sigma^{UVW})_U-(\sigma^{UV})_U\|_1
 \leq b_\beta e^{-c_\beta|V|}.
\label{stat:mix-foundation-li}
\end{equation}
Here $\sigma^R=e^{-\beta H_R}/\Tr[e^{-\beta H_R}]$ denotes the
Gibbs state of the Hamiltonian restricted to $R$.
\end{lemma}
\begin{proof}
This follows from the interval covariance estimate in
\cite{BergamaschiChen2025FastMixing}, followed by the local
indistinguishability implication in \cite[Proposition~7.1]{bcp}; its
superexponential term is absorbed into the exponential above.
\end{proof}

\begin{lemma}[Uniform bounds and locality of expansionals]
\label{stat:expansional-estimates}
For consecutive intervals $U,V,W$, we set
\[
 E_{U,V}:=e^{-\beta H_{UV}/2}
               e^{\beta(H_U+H_V)/2}.
\]
These operators are called \emph{expansionals}.
There are constants $Q\geq1$ and $b'_\beta,c'_\beta>0$, and a
decaying function $r_\beta(t)\leq b'_\beta e^{-c'_\beta t}$, such that
\begin{align}
 \|E_{U,V}\|_\infty,\ \|E_{U,V}^{-1}\|_\infty&\leq Q,
 \label{stat:mix-foundation-expansional-norm}\\
 \|E_{U,VW}-E_{U,V}\|_\infty&\leq r_\beta(|V|).
 \label{stat:mix-foundation-expansional-locality}
\end{align}
\end{lemma}
\begin{proof}
These follow from \cite[Proposition~3.2 and Corollary~3.4]{bcp}.
Writing $E_{U,V}(z)=e^{-zH_{UV}}e^{z(H_U+H_V)}$, the inverse is
controlled by $E_{U,V}(z)^{-1}=E_{U,V}(-z)^*$ for real $z$.
The two expansionals share the same truncation on
collars of radius less than $|V|$.  Their factorial tails imply the
displayed exponential bound.
\end{proof}

\begin{proof}[Proof of Lemma~\ref{stat:inhomogeneous-connected-mixing}]
We combine Lemmas~\ref{stat:local-indistinguishability}
and~\ref{stat:expansional-estimates} to obtain a comparison in relative
operator order.  Both estimates also apply after folding, which keeps
the local dimension a fixed power of two.

The geometry has two stages.  We first compare the marginals on $U$
for the consecutive chains $UV_1V_2W$ and $UV_1V_2$: the buffer
$V=V_1V_2$ separates $U$ from the region $W$ that we remove, as
illustrated in Figure~\ref{stat:folded-gap-removal}(a).
We then apply this comparison to the folded chain in
Figure~\ref{stat:folded-gap-removal}(c), where the retained region is
$X\cup Z$, the buffer is $Y_1\cup Y_3$, and the removed region is
$Y_2$.  In the original chain, this removes the middle of the separator
and leaves two independent intervals.

We fix $U,V,W$ and split $V=V_1V_2$, with $V_1$ adjacent to $U$ and both
pieces of length at least $(|V|-1)/2$.  We set
\[
 F_{U,V}=E_{U,V}^{*}E_{U,V},\qquad
 T_U^{UV}=\Tr_V[\sigma^V F_{U,V}],\qquad
 T_U^{UVW}=\Tr_{VW}[\sigma^{VW}F_{U,VW}].
\]
The superscript records the region on which the operator before the
partial trace is defined, while the subscript records the retained
region. Thus $T_U^{UV}$ and $T_U^{UVW}$ both act on $U$, but arise
from operators on $UV$ and $UVW$, respectively. This follows the
same convention as the Gibbs marginals $(\sigma^{UV})_U$ and
$(\sigma^{UVW})_U$.
Weighted partial trace is unital and completely positive, since
$\Tr_V[\sigma^V F]=\Tr_V[(\sigma^V)^{1/2}F(\sigma^V)^{1/2}]$.  Hence
\eqref{stat:mix-foundation-expansional-norm} gives
\begin{equation}
 Q^{-2}\one\leq T_U^{UV},T_U^{UVW}\leq Q^2\one.
\label{stat:positive-boundary-bounds}
\end{equation}
Moreover, multiplication of the exponentials gives the exact identity
\[
 F_{U,V}=e^{\beta(H_U+H_V)/2}
          e^{-\beta H_{UV}}e^{\beta(H_U+H_V)/2},
 \qquad
 \Tr_V[e^{-\beta H_{UV}}]
 =Z_V e^{-\beta H_U/2}T_U^{UV}e^{-\beta H_U/2}.
\]
Here $Z_V=\Tr[e^{-\beta H_V}]$. With $S=e^{-\beta H_U/2}$, both reduced Gibbs states have the
same outside factors:
\begin{equation}
 (\sigma^{UV})_U=\frac{ST_U^{UV}S}{\Tr[ST_U^{UV}S]},\qquad
 (\sigma^{UVW})_U=\frac{ST_U^{UVW}S}{\Tr[ST_U^{UVW}S]}.
\label{stat:positive-boundary-marginals}
\end{equation}

Next, for any two expansionals bounded by $Q$,
\[
 \|E^*E-\widetilde E^*\widetilde E\|_\infty
 \leq2Q\|E-\widetilde E\|_\infty.
\]
We localise both $F_{U,VW}$ and $F_{U,V}$ to $F_{U,V_1}$.
By hypothesis, $r_\beta(t)\leq b'_\beta e^{-c'_\beta t}$, and hence the two
localisation errors sum to at most $4Qr_\beta(|V_1|)$.
For the remaining difference, we set
$D=(\sigma^{VW})_{V_1}-(\sigma^V)_{V_1}$.  Trace duality gives
\[
 \|\Tr_{V_1}[D F_{U,V_1}]\|_\infty
 \leq\|D\|_1\|F_{U,V_1}\|_\infty.
\]
Combining \eqref{stat:mix-foundation-li} for $V_1,V_2,W$ with the two
localisation errors yields
\begin{equation}
 \|T_U^{UVW}-T_U^{UV}\|_\infty
 \leq4Qr_\beta(|V_1|)+Q^2b_\beta e^{-c_\beta|V_2|}
 =:\omega_\beta(|V|)
 \leq b''_\beta e^{-c''_\beta|V|}.
\label{stat:positive-boundary-difference}
\end{equation}
We set $\epsilon=Q^2\omega_\beta(|V|)$.  Positivity and
\eqref{stat:positive-boundary-bounds} imply
\[
 (1-\epsilon)T_U^{UV}\leq T_U^{UVW}\leq(1+\epsilon)T_U^{UV}.
\]
For $\epsilon<1$, congruence by $S$ and normalisation give
\begin{equation}
 \frac{1-\epsilon}{1+\epsilon}(\sigma^{UV})_U
 \leq(\sigma^{UVW})_U\leq
 \frac{1+\epsilon}{1-\epsilon}(\sigma^{UV})_U.
\label{stat:proved-strong-li}
\end{equation}
Thus the comparison factors are $e^{\pm s_\beta(|V|)}$, where
$s_\beta(|V|)=\log[(1+\epsilon)/(1-\epsilon)]
\leq4\epsilon\leq Ce^{-a|V|}$ for sufficiently large $|V|$.

We split $Y=Y_1Y_2Y_3$ so that $|Y_1|=|Y_3|=\lfloor|Y|/3\rfloor$.
The large-gap threshold ensures that $Y_2$ exceeds the interaction
range.  Removing it leaves two noninteracting intervals, and hence
\[
 \sigma^{XY_1\sqcup Y_3Z}
 =\sigma^{XY_1}\otimes\sigma^{Y_3Z}.
\]
We fold the chain through the midpoint of $Y_2$, as shown in
Figure~\ref{stat:folded-gap-removal}(b)--(c).  Paired sites become a
single site, the two outer regions $X,Z$ form an outer interval, the
two equal collars $Y_1,Y_3$ form its buffer, and $Y_2$ forms the inner
interval.  The folded chain has doubled local qubit number and bounded
nearest-neighbour strength; shorter legs can be padded with decoupled
spins.  Thus the preceding relative-order estimate applies to this chain too, with constants depending
only on the original fixed parameters.  This is the folding geometry
used in \cite{BergamaschiChen2025FastMixing}.

%
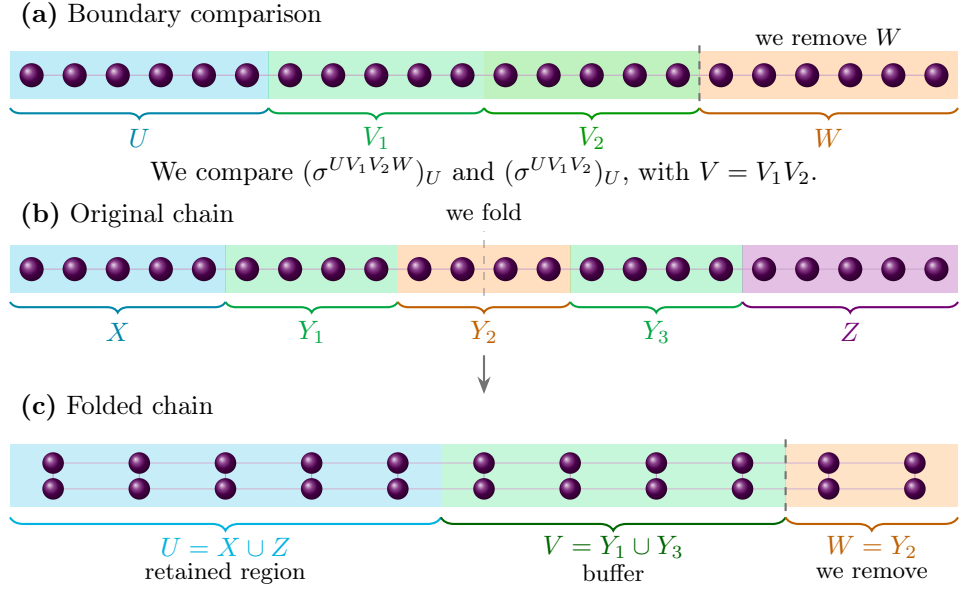
\begin{figure}[!htb]
\centering
\begin{tikzpicture}[x=.57cm,y=.57cm,font=\small,>=Stealth]
 \node[anchor=west] at (-.5,7.9) {\textbf{(a)} Boundary comparison};
 \foreach \a/\b/\lab/\col in
  {-.5/5.5/U/darkblue!30!cyan,5.5/10.5/V_1/midgreen!65!cyan,
   10.5/15.5/V_2/midgreen,15.5/21.5/W/orange}{
  \fill[\col,opacity=.25] (\a,5.95) rectangle (\b,7.05);
  \draw[thick,\col!75!black,decorate,
    decoration={brace,mirror,amplitude=4pt}] (\a,5.74)--(\b,5.74);
  \node[text=\col!75!black] at ({(\a+\b)/2},5.05) {$\lab$};
 }
 \draw[darkpurple!35] (0,6.5)--(21,6.5);
 \foreach \x in {0,...,21}
  \shade[shading=ball,ball color=darkpurple] (\x,6.5) circle (.28);
 \draw[black!55,dashed,thick] (15.5,5.9)--(15.5,7.35);
 \node[font=\footnotesize] at (18.5,7.4) {we remove $W$};
 \node at (10.5,4.25)
  {We compare $(\sigma^{UV_1V_2W})_U$ and $(\sigma^{UV_1V_2})_U$,
   with $V=V_1V_2$.};

 \node[anchor=west] at (-.5,3.25) {\textbf{(b)} Original chain};
 \foreach \a/\b/\lab/\col in
  {-.5/4.5/X/darkblue!30!cyan,4.5/8.5/Y_1/midgreen!65!cyan,
   8.5/12.5/Y_2/orange,12.5/16.5/Y_3/midgreen!65!cyan,
   16.5/21.5/Z/midpink}{
  \fill[\col,opacity=.25] (\a,1.45) rectangle (\b,2.55);
  \draw[thick,\col!75!black,decorate,
    decoration={brace,mirror,amplitude=4pt}] (\a,1.24)--(\b,1.24);
  \node[text=\col!75!black] at ({(\a+\b)/2},.55) {$\lab$};
 }
 \draw[darkpurple!35] (0,2)--(21,2);
 \foreach \x in {0,...,21}
  \shade[shading=ball,ball color=darkpurple] (\x,2) circle (.28);
 \draw[black!45,dashed] (10.5,1.38)--(10.5,2.88);
 \node[anchor=south,font=\footnotesize] at (10.5,2.91) {we fold};
 \draw[->,thick,black!55] (10.5,.0)--(10.5,-.9);

 \node[anchor=west] at (-.5,-1.2) {\textbf{(c)} Folded chain};
 \foreach \a/\b/\col in
  {-.5/9.5/darkblue!30!cyan,9.5/17.5/midgreen!65!cyan,
   17.5/21.5/orange}{
  \fill[\col,opacity=.22] (\a,-3.52) rectangle (\b,-2.08);
 }
 \draw[darkpurple!30] (.5,-2.5)--(20.5,-2.5);
 \draw[darkpurple!30] (.5,-3.1)--(20.5,-3.1);
 \foreach \x in {.5,2.5,...,20.5}{
  \draw[darkpurple!40] (\x,-2.5)--(\x,-3.1);
  \shade[shading=ball,ball color=darkpurple] (\x,-2.5) circle (.25);
  \shade[shading=ball,ball color=darkpurple] (\x,-3.1) circle (.25);
 }
 \foreach \a/\b/\lab/\col in
  {-.5/9.5/{U=X\cup Z}/darkblue!30!cyan,
   9.5/17.5/{V=Y_1\cup Y_3}/midgreen!50!black,17.5/21.5/{W=Y_2}/orange!75!black}{
  \draw[thick,\col,decorate,decoration={brace,mirror,amplitude=5pt}]
   (\a,-3.76)--(\b,-3.76);
  \node[text=\col] at ({(\a+\b)/2},-4.45) {$\lab$};
 }
 \draw[black!55,dashed,thick] (17.5,-1.95)--(17.5,-3.65);
 \node[font=\footnotesize] at (4.5,-5.05) {retained region};
 \node[font=\footnotesize] at (13.5,-5.05) {buffer};
 \node[font=\footnotesize] at (19.5,-5.05) {we remove};
\end{tikzpicture}
\caption{The two geometries in the proof of
Lemma~\ref{stat:inhomogeneous-connected-mixing}.
In (a), we compare the marginals on $U$ before and after removing $W$
from the Hamiltonian. The buffer $V=V_1V_2$ separates $U$ from $W$;
we localise the boundary operators to $UV_1$ and use local
indistinguishability across $V_2$.
Panels (b) and (c) show the original chain and its folding.
Each vertical pair of balls in (c) is one
folded site, pairing original sites at equal distance from the midpoint.
The collars $Y_1,Y_3$ have equal lengths. The labels $U,V,W$ in (c)
identify the regions to which we apply the comparison from (a).
In the folded geometry,
removing $Y_2$ is a boundary-removal operation; in the original geometry,
it leaves independent left and right chains. Their marginal on
$XZ$ is $(\sigma^{XY_1})_X\otimes(\sigma^{Y_3Z})_Z$, which is compared with $\sigma_{XZ}$.}
\label{stat:folded-gap-removal}
\end{figure}
 
We now apply \eqref{stat:proved-strong-li} to the folded chain, with
\[
 U=X\cup Z,\qquad V=Y_1\cup Y_3,\qquad W=Y_2.
\]
Then the two states compared in \eqref{stat:proved-strong-li} are
precisely the states in the following inequality.  We write their
comparison factors as $e^{\pm s}$.  The folded buffer has
$|Y_1|=|Y_3|=\lfloor|Y|/3\rfloor$ sites, so the exponential estimate
following \eqref{stat:proved-strong-li}, with constants adjusted for
folding, gives
\begin{equation}
 e^{-s}(\sigma^{XY_1})_X\otimes(\sigma^{Y_3Z})_Z\leq\sigma_{XZ}
 \leq e^s(\sigma^{XY_1})_X\otimes(\sigma^{Y_3Z})_Z,
 \qquad s\leq Ce^{-a\lfloor|Y|/3\rfloor}.
\label{stat:middle-gap-removal}
\end{equation}
Taking the partial traces of this same inequality gives
$e^{-s}(\sigma^{XY_1})_X\leq\sigma_X\leq e^s(\sigma^{XY_1})_X$ and
$e^{-s}(\sigma^{Y_3Z})_Z\leq\sigma_Z\leq e^s(\sigma^{Y_3Z})_Z$.
Tensoring and combining with \eqref{stat:middle-gap-removal} proves
\[
 e^{-3s}\sigma_X\otimes\sigma_Z\leq\sigma_{XZ}
 \leq e^{3s}\sigma_X\otimes\sigma_Z.
\]
After adjusting the constants, this is
\eqref{stat:connected-mix-input}.

Operator monotonicity of the logarithm gives the asserted MCMI bound.
\end{proof}

\subsection{From connected cuts to the overlapping cover}

We now use the preceding estimates for regions that are unions of
connected components; to this end, we first establish the following
lemma on the factorisation of separated marginals.

\subsubsection{Factorisation of separated marginals}

For ordered, pairwise disjoint regions $R_1<\cdots<R_k$, we define
\begin{equation}
 \mathcal E_\sigma(R_1,\ldots,R_k)
 :=\log\sigma_{\cup_iR_i}-\sum_{i=1}^k\log\sigma_{R_i}.
 \label{stat:disconnected-factorisation-error}
\end{equation}
Here $R_i<R_{i+1}$ means that every site of $R_i$ lies to the left of
every site of $R_{i+1}$; the regions need not be connected.
Empty regions are omitted, and the defect is zero for at most one
region. For two regions,
$\mathcal E_\sigma(U,V)=\mathbf H_\sigma(U:V\mid\varnothing)$.

\begin{lemma}[Factorisation of separated marginals]
\label{stat:cut-retain}
\label{stat:separated-interval-factorisation}
If $k\geq1$ and every consecutive pair $R_i,R_{i+1}$ is separated by
at least $L$ sites, with $L$ above the threshold in
Lemma~\ref{stat:inhomogeneous-connected-mixing}, then
\begin{equation}
 \|\mathcal E_\sigma(R_1,\ldots,R_k)\|_\infty
 \leq(k-1)g_\beta(L).
 \label{stat:separated-log-bound}
\end{equation}
\end{lemma}
\begin{proof}
For a consecutive partition $\Lambda=U\sqcup V\sqcup W$ with
$|V|\geq L$, positivity of partial trace transfers
\eqref{stat:connected-mix-input} to any $U'\subseteq U$, $W'\subseteq W$:
\begin{equation}
 e^{-g_\beta(L)}\sigma_{U'}\otimes\sigma_{W'}
 \leq\sigma_{U'W'}\leq
 e^{g_\beta(L)}\sigma_{U'}\otimes\sigma_{W'}.
 \label{stat:retained-order}
\end{equation}
Operator monotonicity of the logarithm therefore gives
\begin{equation}
 H_\sigma(U':W'\mid\varnothing)\leq g_\beta(L).
 \label{stat:retained-mcmi}
\end{equation}
This applies to disconnected retained regions because the partial trace
is taken in the operator-order estimate, before taking logarithms.

We now telescope:
\[
 \mathcal E_\sigma(R_1,\ldots,R_k)
 =\sum_{j=1}^{k-1}
 \mathbf H_\sigma(R_1\cup\cdots\cup R_j:R_{j+1}\mid\varnothing).
\]
For each term, we partition the full chain at the gap between $R_j$ and
$R_{j+1}$, retaining the preceding regions from the connected prefix
and $R_{j+1}$ from the connected suffix.
Equation~\eqref{stat:retained-mcmi} bounds every term by $g_\beta(L)$,
proving \eqref{stat:separated-log-bound}.
\end{proof}

\subsubsection{The multipartite defect}
We write the chain in its spatial order:
\begin{equation}
 \Lambda=X_0A_1X_1A_2\cdots X_{m-1}A_mX_m,
 \qquad A=\bigcup_{i=1}^mA_i,\qquad
 X=\Lambda\setminus A=\bigcup_{i=0}^mX_i.
 \label{stat:ordered-components-gaps}
\end{equation}
For $1\leq i<m$, the gap $X_i$ lies between $A_i$ and $A_{i+1}$; the endpoint gaps
$X_0,X_m$ may be empty. We set
\[
 L_i=X_0A_1X_1\cdots A_{i-1}X_{i-1},\qquad
 R_i=X_iA_{i+1}X_{i+1}\cdots A_mX_m.
\]
These are the connected tails on either side of $A_i$, so
$A_i^c=L_iR_i$, $L_1=X_0$, and $R_m=X_m$.
For $B=\bigcup_jB_j$, we use the same convention, with $Y_j$ the gaps
between its adjacent blocks.

Figure~\ref{stat:connected-and-traced-cuts} shows the three contributions
below: a gap retained as a conditioning region, a block traced out
between its two tails, and further traces retaining only selected gaps.

%
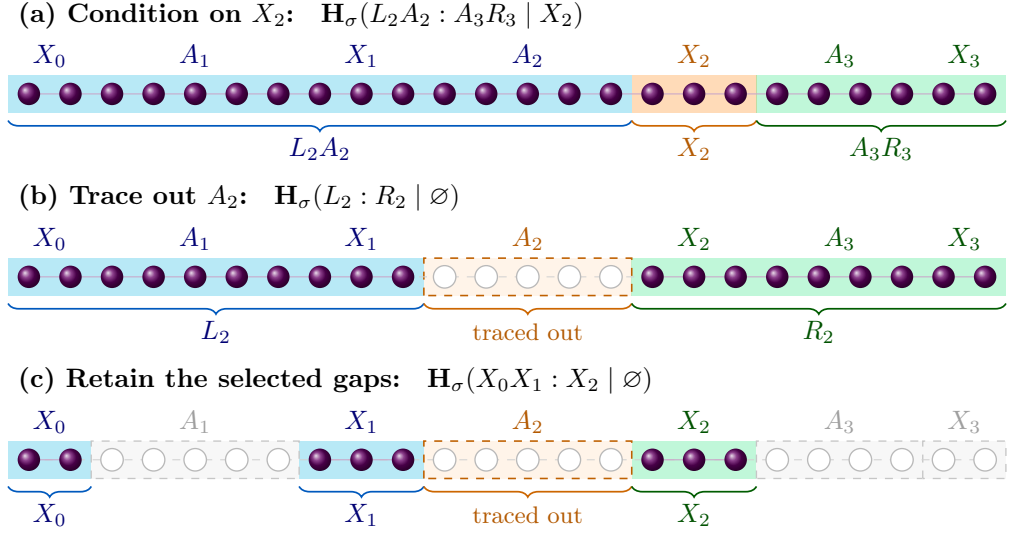
\begin{figure}[ht]
 \centering
 \begin{tikzpicture}[x=.55cm,y=.55cm,font=\small,
  qfbrace/.style={line width=.65pt,decorate,
   decoration={brace,mirror,amplitude=4pt}},
  qfghost/.style={draw=gray!55,fill=white,line width=.5pt}]
  \node[anchor=west] at (-.5,1.9)
   {\textbf{(a) Condition on $X_2$:}\quad
    $\mathbf H_\sigma(L_2A_2:A_3R_3\mid X_2)$};
  \fill[darkblue!30!cyan,opacity=.28]
   (-.5,-.45) rectangle (14.5,.45);
  \fill[orange,opacity=.28]
   (14.5,-.45) rectangle (17.5,.45);
  \fill[midgreen!60!cyan,opacity=.25]
   (17.5,-.45) rectangle (23.5,.45);
  \draw[darkpurple!35,line width=.4pt] (0,0)--(23,0);
  \foreach \x in {0,...,23}
   \shade[shading=ball,ball color=darkpurple] (\x,0) circle (.27);
  \foreach \x/\lab in {.5/X_0,4/A_1,8/X_1,12/A_2}
   \node[text=darkblue!75!black] at (\x,.91) {$\lab$};
  \node[text=orange!70!black] at (16,.91) {$X_2$};
  \foreach \x/\lab in {19.5/A_3,22.5/X_3}
   \node[text=midgreen!40!black] at (\x,.91) {$\lab$};
  \draw[qfbrace,darkblue!65!cyan] (-.5,-.64)--(14.5,-.64);
  \draw[qfbrace,orange!80!black] (14.5,-.64)--(17.5,-.64);
  \draw[qfbrace,midgreen!45!black] (17.5,-.64)--(23.5,-.64);
  \node[text=darkblue!75!black] at (7,-1.32) {$L_2A_2$};
  \node[text=orange!70!black] at (16,-1.32) {$X_2$};
  \node[text=midgreen!40!black] at (20.5,-1.32) {$A_3R_3$};

  \begin{scope}[yshift=-2.42cm]
   \node[anchor=west] at (-.5,1.9)
    {\textbf{(b) Trace out $A_2$:}\quad
     $\mathbf H_\sigma(L_2:R_2\mid\varnothing)$};
   \fill[darkblue!30!cyan,opacity=.28]
    (-.5,-.45) rectangle (9.5,.45);
   \fill[midgreen!60!cyan,opacity=.25]
    (14.5,-.45) rectangle (23.5,.45);
   \fill[orange,opacity=.09]
    (9.5,-.45) rectangle (14.5,.45);
   \draw[orange!75!black,dashed,line width=.6pt]
    (9.5,-.45) rectangle (14.5,.45);
   \draw[darkpurple!35,line width=.4pt] (0,0)--(9,0);
   \draw[gray!35,dashed,line width=.4pt] (9,0)--(15,0);
   \draw[darkpurple!35,line width=.4pt] (15,0)--(23,0);
   \foreach \x in {0,...,9}
    \shade[shading=ball,ball color=darkpurple] (\x,0) circle (.27);
   \foreach \x in {15,...,23}
    \shade[shading=ball,ball color=darkpurple] (\x,0) circle (.27);
   \foreach \x in {10,...,14}
    \draw[qfghost] (\x,0) circle (.27);
   \foreach \x/\lab in {.5/X_0,4/A_1,8/X_1}
    \node[text=darkblue!75!black] at (\x,.91) {$\lab$};
   \node[text=orange!70!black] at (12,.91) {$A_2$};
   \foreach \x/\lab in {16/X_2,19.5/A_3,22.5/X_3}
    \node[text=midgreen!40!black] at (\x,.91) {$\lab$};
   \draw[qfbrace,darkblue!65!cyan] (-.5,-.64)--(9.5,-.64);
   \draw[qfbrace,orange!80!black] (9.5,-.64)--(14.5,-.64);
   \draw[qfbrace,midgreen!45!black] (14.5,-.64)--(23.5,-.64);
   \node[text=darkblue!75!black] at (4.5,-1.32) {$L_2$};
   \node[text=orange!70!black,font=\footnotesize]
    at (12,-1.32) {traced out};
   \node[text=midgreen!40!black] at (19,-1.32) {$R_2$};
  \end{scope}

  \begin{scope}[yshift=-4.84cm]
   \node[anchor=west] at (-.5,1.9)
    {\textbf{(c) Retain the selected gaps:}\quad
     $\mathbf H_\sigma(X_0X_1:X_2\mid\varnothing)$};
   \foreach \a/\b in {-.5/1.5,6.5/9.5}
    \fill[darkblue!30!cyan,opacity=.28]
     (\a,-.45) rectangle (\b,.45);
   \fill[midgreen!60!cyan,opacity=.25]
    (14.5,-.45) rectangle (17.5,.45);
   \foreach \a/\b in {1.5/6.5,17.5/21.5,21.5/23.5}{
    \fill[gray!6] (\a,-.45) rectangle (\b,.45);
    \draw[gray!45,dashed,line width=.5pt]
     (\a,-.45) rectangle (\b,.45);
   }
   \fill[orange,opacity=.09]
    (9.5,-.45) rectangle (14.5,.45);
   \draw[orange!75!black,dashed,line width=.6pt]
    (9.5,-.45) rectangle (14.5,.45);
   \draw[gray!35,dashed,line width=.4pt] (0,0)--(23,0);
   \foreach \a/\b in {0/1,7/9,15/17}
    \draw[darkpurple!35,line width=.4pt] (\a,0)--(\b,0);
   \foreach \x in {0,1,7,8,9,15,16,17}
    \shade[shading=ball,ball color=darkpurple] (\x,0) circle (.27);
   \foreach \x in {2,...,6}
    \draw[qfghost] (\x,0) circle (.27);
   \foreach \x in {10,...,14}
    \draw[qfghost] (\x,0) circle (.27);
   \foreach \x in {18,...,23}
    \draw[qfghost] (\x,0) circle (.27);
   \foreach \x/\lab in {.5/X_0,8/X_1}
    \node[text=darkblue!75!black] at (\x,.91) {$\lab$};
   \foreach \x/\lab in {4/A_1,19.5/A_3,22.5/X_3}
    \node[text=gray!75] at (\x,.91) {$\lab$};
   \node[text=orange!70!black] at (12,.91) {$A_2$};
   \node[text=midgreen!40!black] at (16,.91) {$X_2$};
   \foreach \a/\b in {-.5/1.5,6.5/9.5}
    \draw[qfbrace,darkblue!65!cyan] (\a,-.64)--(\b,-.64);
   \draw[qfbrace,orange!80!black] (9.5,-.64)--(14.5,-.64);
   \draw[qfbrace,midgreen!45!black] (14.5,-.64)--(17.5,-.64);
   \node[text=darkblue!75!black] at (.5,-1.32) {$X_0$};
   \node[text=darkblue!75!black] at (8,-1.32) {$X_1$};
   \node[text=orange!70!black,font=\footnotesize]
    at (12,-1.32) {traced out};
   \node[text=midgreen!40!black] at (16,-1.32) {$X_2$};
  \end{scope}
 \end{tikzpicture}
 \caption{The three contributions to the defect on the same chain
 $X_0A_1X_1A_2X_2A_3X_3$.  Purple balls are retained sites; hollow
 grey balls are traced-out sites.  In (a), the orange interval $X_2$
 is retained as the conditioning system, so the connected conditional
 MCMI estimate applies.  In (b), the orange interval $A_2$ is traced
 out, leaving the two tails $L_2$ and $R_2$.  In (c), further partial
 traces retain only $X_0\cup X_1$ and $X_2$, giving one merger in the
 gap marginal.  The multiplicative estimate across $A_2$ controls both
 (b) and (c).}
 \label{stat:connected-and-traced-cuts}
\end{figure}
 
\begin{lemma}[Connected-cut expansion of the multipartite defect]
\label{stat:connected-cut-expansion}
For the decomposition \eqref{stat:ordered-components-gaps},
\begin{equation}
 \begin{aligned}
 K_\sigma(A)
 &=\sum_{i=1}^{m-1}
   \mathbf H_\sigma(L_iA_i:A_{i+1}R_{i+1}\mid X_i)+\mathcal E_\sigma(X_0,\ldots,X_m)
       -\sum_{i=1}^m\mathcal E_\sigma(L_i,R_i).
 \end{aligned}
 \label{stat:connected-cut-identity}
\end{equation}
If every block $A_i$ and every internal gap $X_i$ has length at least
$L$, with $L$ above the thresholds of the connected estimates, then
\begin{equation}
 \|K_\sigma(A)\|_\infty
 \leq(m-1)F_\beta(L)+2m g_\beta(L).
 \label{stat:multipartite-connected-bound}
\end{equation}
Here $F_\beta$ and $g_\beta$ are the exponentially decaying errors in
\eqref{stat:connected-mcmi-exponential} and
\eqref{stat:connected-mix-input}.
\end{lemma}
\begin{proof}
For each internal gap,
\[
 \mathbf H_\sigma(L_iA_i:A_{i+1}R_{i+1}\mid X_i)
 =\log\sigma+\log\sigma_{X_i}
  -\log\sigma_{L_{i+1}}-\log\sigma_{R_i}.
\]
Summing and expanding the $\mathcal E_\sigma$ terms proves
\eqref{stat:connected-cut-identity}: the boundary logarithms cancel
because $L_1=X_0$ and $R_m=X_m$.

Each of the $m-1$ conditional MCMI terms is bounded by $F_\beta(L)$.
Since the gaps $X_i$ are separated by the blocks $A_i$, and the tails
$L_i,R_i$ are separated by $A_i$,
Lemma~\ref{stat:separated-interval-factorisation} gives
\[
 \|\mathcal E_\sigma(X_0,\ldots,X_m)\|_\infty\leq m g_\beta(L),
 \qquad
 \sum_{i=1}^m\|\mathcal E_\sigma(L_i,R_i)\|_\infty
 \leq m g_\beta(L).
\]
Empty gaps or tails contribute zero. The triangle inequality proves
\eqref{stat:multipartite-connected-bound}; for $m=1$ the exact identity
gives $K_\sigma(A)=0$.
\end{proof}

\subsubsection{The four-cell cover}

We now prove Theorem~\ref{stat:weak-at}.  The geometry in
Figure~\ref{stat:four-cell-cover} makes both the components and the gaps
large enough for the preceding estimates.

%
\begin{figure}[htbp]
\centering
\begin{tikzpicture}[x=.55cm,y=.55cm,font=\small,
  component/.style={thick,black!65,decorate,
    decoration={brace,amplitude=5pt,mirror}}]
  \definecolor{qfball}{RGB}{102,0,102}
  \definecolor{qfcyan}{RGB}{50,165,210}
  \definecolor{qfgreen}{RGB}{0,170,65}
  \definecolor{qfpink}{RGB}{175,55,160}
  \foreach \x/\col/\lab in {
    0/qfcyan/X_0,3/qfgreen/O,6/qfpink/Y_1,9/qfgreen/O,
    12/qfcyan/X_1,15/qfgreen/O,18/qfpink/Y_2,21/qfgreen/O}{
    \fill[\col,opacity=.27] (\x-.5,1.95) rectangle (\x+2.5,3.05);
    \node[text=\col!70!black] at (\x+1,3.7) {$\lab$};
  }
  \draw[black!30] (0,2.5)--(23,2.5);
  \foreach \x in {0,...,23}
    \shade[shading=ball,ball color=qfball] (\x,2.5) circle (.27);
  \node[anchor=east] at (-1,2.5) {$\Lambda$};
  \foreach \a/\z/\lab in {2.5/11.5/A_1,14.5/23.5/A_2}{
    \draw[component] (\a,1.5)--(\z,1.5);
    \node at ({(\a+\z)/2},.65) {$\lab$};
  }
  \node[anchor=east] at (-1,.65) {$A$};
  \foreach \a/\z/\lab in {-.5/5.5/B_1,8.5/17.5/B_2,20.5/23.5/B_3}{
    \draw[component] (\a,-.15)--(\z,-.15);
    \node at ({(\a+\z)/2},-1) {$\lab$};
  }
  \node[anchor=east] at (-1,-1) {$B$};
\end{tikzpicture}
\caption{The four-cell cover, drawn with three sites per cell for
illustration. The braces display the two families of intervals
$A=\bigcup_iA_i$ and $B=\bigcup_jB_j$, shifted by two cells relative
to each other. Their overlaps $A\cap B$ are the green $O$-cells.
We then define $X=X_0\cup X_1=\Lambda\setminus A$ and
$Y=Y_1\cup Y_2=\Lambda\setminus B$; the endpoint gaps
$X_2,Y_0,Y_3$ are empty. The resulting cell types repeat in the order
$X,O,Y,O$, and the actual cell lengths lie in $[L,2L]$.
Both the components and their separating
gaps have length at least $L$, and $A\cup B=\Lambda$.}
\label{stat:four-cell-cover}
\end{figure}
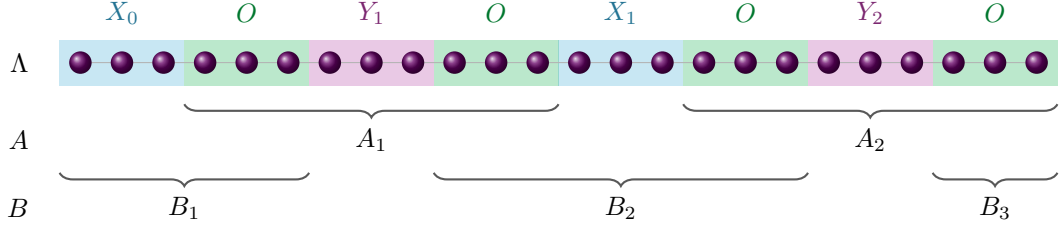

\begin{proof}[Proof of Theorem~\ref{stat:weak-at}]
If $n<4L$, the single block $\Lambda$ gives equality.  Otherwise we set
$N=\lfloor n/L\rfloor$.  We partition the chain into $N$ consecutive cells
$C_1,\ldots,C_N$:
the first $N-1$ have length $L$, and the last has length in $[L,2L)$.
We construct two families of intervals by taking three consecutive
cells at a time, shifting by four cells within each family and by two
cells between the families, as illustrated in Figure~\ref{fig:summary-cover}.
More precisely, with $C_k=\varnothing$
for $k\notin\{1,\ldots,N\}$, we define
\[
 A_i=\bigcup_{k=4i-2}^{4i}C_k,\qquad
 B_j=\bigcup_{k=4j-4}^{4j-2}C_k
 \qquad(i,j\geq1),
\]
and omit empty intervals.  We then set
\[
 A=\bigcup_i A_i,\qquad B=\bigcup_j B_j,\qquad
 \mathcal G=\{A_i\}_i\cup\{B_j\}_j.
\]
The intervals within each family are disjoint, while successive
intervals from the two families overlap in one cell.  Thus
$A\cup B=\Lambda$ and the cover multiplicity is at most two.
Every internal component
contains three cells; every nonempty endpoint component contains one,
two, or three cells.  Thus their lengths belong to $[L,6L]$ and all
their internal gaps have length at least $L$.

We now define the complementary regions
\[
 X=\Lambda\setminus A,\qquad Y=\Lambda\setminus B.
\]
The region $X$ consists of the cells $C_1,C_5,C_9,\ldots$, and $Y$
consists of $C_3,C_7,C_{11},\ldots$, with indices at most $N$.
The remaining cells form the overlaps $A\cap B$, so the cell types
occur in the order $X,O,Y,O$, where $O$ denotes an overlap cell.

We number the gaps by the adjacent blocks as above: the $X$-cells are
$X_0,\ldots,X_{s-1}$ and the $Y$-cells are $Y_1,\ldots,Y_t$,
where $s$ and $t$ are their respective numbers.  The left endpoint
gap $Y_0$ is empty.  In the notation
of \eqref{stat:disconnected-factorisation-error}, the first error is
\begin{align*}
 \log\sigma_{XY}-\log\sigma_{X}
                  -\log\sigma_{Y}
 &=\mathcal E_\sigma(X_0,Y_1,X_1,Y_2,\ldots)-\mathcal E_\sigma(X_0,X_1,\ldots)
        -\mathcal E_\sigma(Y_1,Y_2,\ldots).
\end{align*}
Lemma~\ref{stat:separated-interval-factorisation} bounds its norm by
$2(s+t)g_\beta(L)$.  Hence \eqref{stat:two-set} gives
\[
 D(\rho\Vert\sigma)\leq D_A(\rho\Vert\sigma)+D_B(\rho\Vert\sigma)
                         +2(s+t)g_\beta(L).
\]
Removing $s$ cells creates at most $s+1$ components of $A$; similarly
$B$ has at most $t+1$ components.  As $s+t\leq N$, there are at most
$N+2\leq2n$ blocks.  We apply \eqref{stat:multipartite} and
Lemma~\ref{stat:connected-cut-expansion} to $A$ and $B$.  The resulting
estimate is
\begin{equation}
 D(\rho\Vert\sigma)
 \leq\sum_{S\in\mathcal G}D_S(\rho\Vert\sigma)
 +4\left(\frac nL+1\right)\bigl(F_\beta(L)+g_\beta(L)\bigr).
 \label{stat:explicit-defect}
\end{equation}
The connected estimates provide fixed constants $M_\beta,a_\beta,L_\beta$
with $F_\beta(L)+g_\beta(L)\leq M_\beta e^{-a_\beta L}$ for
$L\geq L_\beta$.  Since $n/L+1\leq2n$ for $L\geq1$, choosing
$B_\beta=\max\{1,8M_\beta\}$ proves \eqref{stat:weak-at-equation}.
Finally, we take
\[
 K_\beta\geq\max\left\{1,\frac{L_\beta}{\log2},
                  \frac{1+\log B_\beta/\log2}{a_\beta}\right\}.
\]
For $L=\lceil K_\beta\log n\rceil$,
we have $L\ge L_\beta$ and
\[
 B_\beta n e^{-a_\beta L}
 \le B_\beta n^{1-a_\beta K_\beta}
 \le B_\beta 2^{1-a_\beta K_\beta}\le1.
\]
\end{proof}
 %
\section{A uniform conditional gap}
\label{gap:section}

In this section we show how the canonical purified Hamiltonian from~\cite{LuciaPerezGarciaPerezHernandez2025Gap} can be used to prove a uniform conditional gap for the heat-bath sampler (and consequently for the regularised generator).

In particular, we are able to prove that:
\begin{theorem}[Uniform conditional gap]
	\label{gap:conditional}
	For a bounded finite-range interaction on an open chain, at every
	fixed $\beta>0$, there is $\gamma_\beta>0$, depending only on
	$\beta,d,r_0,J$, 
	such that every interval $S\subseteq\Lambda$ satisfies
	\begin{equation}
		\mathcal E_S^a(X)\ge\gamma_\beta\operatorname{Var}_S(X)
		\qquad(X\in\mathcal B(\mathcal H_\Lambda))\,,
		\label{gap:conditional-equation}
	\end{equation}
	where $\mathcal E_S^a(X)$ is the conditional GNS Dirichlet form on $S$ defined in Eq.~\eqref{gap:energy} with the global Gibbs state $\sigma$.
	Here $a\in\{\mathrm{HB},\mathrm{RHB}\}$ denotes the chosen sampler.
\end{theorem}

\begin{proof}
	We begin by proving that we can lower bound the conditional Dirichlet form in terms of the on-site conditional variance as:
	\begin{align}
		\mathcal E_S^a(X) \geq \frac{1}{2}\sum_{i\in S} \operatorname{Var}_i(X) \,,\label{eq:4.2}
	\end{align}
	both for the heat-bath and the regularised sampler. Since $\mathcal E_S^a(X)  = \sum_{i\in S} \mathcal E_i^a(X) $, this is equivalent to proving:
	\begin{align}
		\operatorname{Re}
		\langle X,-\cL_i^a(X)\rangle_{\mathrm{GNS},\sigma}\geq \frac{1}{2}\operatorname{Var}_i(X) \qquad a\in\{\mathrm{HB},\mathrm{RHB}\}.\label{eq:4.3}
	\end{align}
	Let us start with the case of the heat-bath sampler $\cL_i^{\mathrm{HB}} := -(\id-\cP_i)$.  First, notice that by applying the Kadison--Schwarz inequality to $\cP_i$ it follows that:
	\[
	\|\cP_iX\|_{\mathrm{GNS},\sigma}^2 = \Tr[\sigma \cP_i(X)^\dagger\cP_i(X)]
	\le\Tr[\sigma\cP_i(X^\dagger X)]
	=\Tr[\sigma X^\dagger X]
	=\|X\|_{\mathrm{GNS},\sigma}^2,
	\]
	where we implicitly used the invariance $\cP_{i*}(\sigma) = \sigma$. Then, Eq.~\eqref{eq:4.3} follows from the inequality:
	\[
	\begin{aligned}
		2\operatorname{Re}\langle X,(\id-\cP_i)X\rangle_{\mathrm{GNS},\sigma}
		&=\|X-\cP_iX\|_{\mathrm{GNS},\sigma}^2
		+\|X\|_{\mathrm{GNS},\sigma}^2
		-\|\cP_iX\|_{\mathrm{GNS},\sigma}^2\\
		&\ge\|X-\cP_iX\|_{\mathrm{GNS},\sigma}^2
		\ge\operatorname{Var}_i(X),
	\end{aligned}
	\]
	where we used $(\|X\|_{\mathrm{GNS},\sigma}^2
	-\|\cP_iX\|_{\mathrm{GNS},\sigma}^2)\ge0$, and the last inequality follows from $\cP_iX\in\mathcal A_{i^c}$. This proves Eq.~\eqref{eq:4.2} in the case of the heat-bath sampler.  For the regularised generator the argument is even simpler, since $(\id - \mathbb E_i)$ is a GNS orthogonal projection. Then, it follows that:
	\begin{align}
		\operatorname{Re}\langle X,(\id-\mathbb E_i)X\rangle_{\mathrm{GNS},\sigma} = \operatorname{Re}\langle (\id-\mathbb E_i)X,(\id-\mathbb E_i)X\rangle_{\mathrm{GNS},\sigma} =\|X-\mathbb E_iX\|_{\mathrm{GNS},\sigma}^2
		\ge\operatorname{Var}_i(X)\,.
	\end{align}
	Then, the result follows from the application of the following lemma:
	\begin{lemma}
		\label{gap:canonical}
		For every interval $S\subseteq \Lambda$, there exists a constant $\gamma_\beta>0$ uniform in $S$ and $\Lambda$ such that:
		\begin{equation}\label{gap:variance-factorisation}
			\sum_{i\in S}\operatorname{Var}_i(X)
			\ge2\gamma_\beta\operatorname{Var}_S(X).
		\end{equation}
	\end{lemma}	
	We defer the proof of this fact to the next subsection.
\end{proof}

\begin{remark}\label{gap:global-memory}
	Theorem~\ref{gap:conditional} directly implies a lower bound on the KMS spectral gap:
	\begin{equation}\label{gap:global-spectral}
		\lambda_{\mathrm{KMS}}(\cL_\Lambda^a)\ge\gamma_\beta.\qquad a\in\{\mathrm{HB},\mathrm{RHB}\}
	\end{equation}
	Indeed, since $\cL_\Lambda^a$ is KMS self-adjoint, its spectrum is real. Then, for any eigenoperator $X$ it holds that $-\cL_\Lambda^a X=\lambda X$ (with
	$\lambda\ne0$), which implies that 
	\[
	\lambda\|X\|_{\mathrm{GNS},\sigma}^2
	=\mathcal E_\Lambda^a(X)
	\ge\gamma_\beta\|X\|_{\mathrm{GNS},\sigma}^2,
	\]
	where in the last step we applied Theorem~\ref{gap:conditional}. Since the KMS gap is the smallest positive eigenvalue, this directly proves Eq.~\eqref{gap:global-spectral}.

	Interestingly, by tensoring $\sigma$ with any full-rank state $\tau$, we can prove the same lower bound for $\cL_\Lambda^a\otimes \id$ (since the Kadison--Schwarz inequality is unchanged by this step). In the opposite scenario, if $S$ is a proper interval, Theorem~\ref{gap:conditional} is not enough to bound the kernel of $\cL_S^a$, and more sophisticated methods are needed (see Lemma~\ref{lem:average}). 
\end{remark}

\begin{remark}\label{cmlsi:global-section}
	The global GNS gap has a further consequence for quantum memories.
	For a primitive generator $\cL$, let $\alpha_{\rm CMLSI}$ be the largest
	constant such that, for every memory $M$,
	\[
	\EP_{\cL\otimes\id_M}(\rho_{\Lambda M})
	\ge\alpha_{\rm CMLSI}
	D(\rho_{\Lambda M}\Vert\sigma\otimes\rho_M).
	\]
	Define the complete index of the equilibrium expectation
	$E_\sigma(X)=\langle\one,X\rangle_{\mathrm{GNS},\sigma}\one$ as
	\[
	C_{\rm cb}(E_\sigma)
	:=\inf\{C>0:C E_\sigma-\id\text{ is completely positive}\}.
	\]
	We use the GNS gap
	\[
	\lambda_{\rm GNS}(\cL)
	:=\inf_{\substack{X\ne0\\\Tr[\sigma X]=0}}
	\frac{-\operatorname{Re}\langle X,\cL(X)\rangle_{\mathrm{GNS},\sigma}}
	{\|X\|_{\mathrm{GNS},\sigma}^{2}}.
	\]
	This is the gap of the GNS symmetrisation; it coincides with
	the spectral gap of $-\cL$ when $\cL$ is GNS symmetric.
	Then from \cite[Corollary 5.5]{bkm-coercivity} and \cite[Theorem 1.1 and Corollary 4.13]{cp-order} we have the two estimates:
	\begin{align}
		\alpha_{\rm CMLSI}(\cL)
		&\ge\frac{\lambda_{\rm GNS}(\cL)}{C_{\rm cb}(E_\sigma)},
		\label{cmlsi:general-gns}\\
		\alpha_{\rm CMLSI}(\cL)
		&\ge\frac{\lambda_{\rm GNS}(\cL)}
		{\log(10C_{\rm cb}(E_\sigma))}
		\quad\text{if $\cL$ is GNS symmetric}.
		\label{cmlsi:symmetric-gns}
	\end{align}
	The first comparison requires no detailed balance. Both
	references use $\EP\ge2\alpha D$; the constants above have been
	multiplied by two to match our convention $\EP\ge\alpha D$.
	Then, putting this ingredient together with Theorem~\ref{gap:conditional} we obtain:
	\begin{corollary}[Complete bounds at finite volume]\label{cmlsi:global}
		Set $Z_\Lambda(t)=\Tr[e^{-tH_\Lambda}]$. For either sampler,
		\begin{equation}\label{cmlsi:global-bound}
			\alpha_{\rm CMLSI}(\cL^a)
			\ge\frac{\gamma_\beta}{Z_\Lambda(\beta)Z_\Lambda(-\beta)}
			\ge\gamma_\beta e^{-2n(\log d+\beta J)}>0,
			\qquad a\in\{\mathrm{HB},\mathrm{RHB}\}.
		\end{equation}
		For the regularised heat-bath sampler one has the stronger bound
		\begin{equation}\label{cmlsi:rhb-global-bound}
			\alpha_{\rm CMLSI}(\cL^{\mathrm{RHB}})
			\ge\frac{\gamma_\beta}
			{\log\!\left(10Z_\Lambda(\beta)Z_\Lambda(-\beta)\right)}
			\ge\frac{\gamma_\beta}{\log10+2n(\log d+\beta J)}.
		\end{equation}
	\end{corollary}
	\begin{proof}
		Set $S= \Lambda$. Then, Theorem~\ref{gap:conditional} shows that $\gamma_\beta$ is a lower bound on the GNS gap, which we can plug directly into Eq.~\eqref{cmlsi:general-gns}. Moreover, notice that the regularised heat-bath sampler satisfies GNS detailed balance, so one can also use Eq.~\eqref{cmlsi:symmetric-gns} to lower bound the CMLSI constant.
		Indeed, each $\mathbb E_i$ is a $\sigma$-preserving conditional
		expectation, hence a GNS orthogonal projection, as shown in the proof of
		Proposition~\ref{lem:rhb-direct-comparisons}. Thus their generator sum
		is GNS symmetric. The Petz heat-bath sampler is generally only KMS
		symmetric, so the stronger comparison does not apply to it in general.

		Then, it remains to prove that $C_{\rm cb}(E_\sigma) \leq d^{2n}e^{2\beta J n}$.
		Denote by
		$D=\dim\mathcal H_\Lambda=d^n$ and
		$|\Omega\rangle=\sum_{j=1}^D|j\rangle\otimes|j\rangle$.
		The Choi matrix of $C E_\sigma-\id$ is $C\sigma^{\mathsf T}\otimes\one-|\Omega\rangle\langle\Omega|$, which is positive semidefinite if and only if $C\geq \Tr[\sigma^{-1}]$. Indeed, conjugating by $(\sigma^{\mathsf T})^{-1/2}\otimes\one$ shows that
		\[
		C\ge\langle\Omega|(\sigma^{\mathsf T})^{-1}\otimes\one
		|\Omega\rangle=\Tr[\sigma^{-1}].
		\]
		This implies that $C_{\rm cb}(E_\sigma)=\Tr[\sigma^{-1}] =Z_\Lambda(\beta)Z_\Lambda(-\beta)$. 
		Since $\|H_\Lambda\|_\infty\le Jn$,
		$Z_\Lambda(\pm\beta)\le d^n e^{\beta Jn}$. This proves the claim.
	\end{proof}

	These complete bounds hold uniformly in the memory dimension. Their
	dependence on the chain length is weaker than the ordinary
	$c_\beta/\log n$ rate in Theorem~\ref{thm:main}. In particular,
	$\Tr[\sigma^{-1}]\ge d^{2n}$, so the index $C_{\rm cb}(E_\sigma)$
	cannot be bounded independently of the volume even at infinite temperature.

	It should also be noticed that at the local level, controlling the conditional gap is not enough to give an estimate on the CMLSI: indeed, the first depends on the distance to $\mathcal A_{S^c}$, whereas for either sampler $a$ the local CMLSI accounts for the
	distance to $\operatorname{Fix}\cL_S^a$. The former might be strictly smaller, so there are directions that are ignored by the conditional gap, but that are relevant for the CMLSI. The two-qubit example in Section~\ref{cmlsi:section} exhibits this
	distinction explicitly.
\end{remark}

\subsection{The canonical purified Hamiltonian estimate}\label{sec:purifiedHamiltonian}

In this section we discuss how one can use the canonical purified Hamiltonian from~\cite{LuciaPerezGarciaPerezHernandez2025Gap} to derive Lemma~\ref{gap:canonical}.

To this end, first it should be noticed that the map $X\rightarrow \widetilde{X}:=X\sigma^{1/2}$ is an isometry from the GNS scalar product to the Hilbert--Schmidt one, since:
\[
\|X\|_{\mathrm{GNS},\sigma}^2 = \Tr[\sigma X^\dagger X] = \Tr[(X\sigma^{1/2})^\dagger (X\sigma^{1/2})] = \|\widetilde{X}\|_{2}^2.
\]
Then, we denote by $\mathcal V_S(0)$ the image of $\mathcal{A}_{S^c}$ under this mapping, that is:
\begin{equation}\label{gap:canonical-spaces}
	\mathcal V_S(0)=\{O\sigma^{1/2}:O\in\mathcal A_{S^c}\}.
\end{equation}
Let us also introduce the orthogonal projection on $\mathcal V_S(0)$, which we denote by $\Pi_{S,0}$. In this context, in~\cite{LuciaPerezGarciaPerezHernandez2025Gap} the canonical purified Hamiltonian associated to $\sigma$ was defined as:
\[
\mathbf{H}_S := \sum_{i\in S} (\id - \Pi_{i,0})\,.
\]
This is positive semidefinite, frustration-free and has kernel $\ker\!\left(\mathbf{H}_S\right)=\mathcal V_S(0)$. Moreover, in~\cite{LuciaPerezGarciaPerezHernandez2025Gap} it was shown that in the case $S=\Lambda$, for one-dimensional systems $\mathbf{H}_\Lambda$ has a gap which is independent of the system size (Theorem 38 therein). The proof of Lemma~\ref{gap:canonical} hinges on this fact, which we extend to the case in which $S\subset\Lambda$.

\begin{proof}[Proof of Lemma~\ref{gap:canonical}]
	First, it should be noticed that Eq.~\eqref{gap:variance-factorisation} is equivalent to the fact that there exists a $\gamma_\beta>0$ uniform in $S$ and $\Lambda$ such that:
	\begin{equation}\label{gap:canonical-equation}
		\mathbf{H}_S\ge2\gamma_\beta(\id-\Pi_{S,0})\,.
	\end{equation}
	Since $(1-\Pi_{S,0})$ has the same kernel as $\mathbf{H}_S$, this is the same as saying that the canonical purified Hamiltonian has a gap. To show that Eq.~\eqref{gap:canonical-equation} is indeed equivalent to Eq.~\eqref{gap:variance-factorisation}, notice that for any $S\subseteq \Lambda$ it holds that:
	\begin{align}
		\operatorname{Var}_S(X)&:=\inf_{O\in\mathcal A_{S^c}}
		\|X-O\|_{\mathrm{GNS},\sigma}^2 = \inf_{\widetilde{O}\in\mathcal{V}_S(0)}
		\|(X\sigma^{1/2})-\widetilde{O}\|_{2}^2 = \|(\id-\Pi_{S,0})\widetilde{X}\|_{2}^2, 
	\end{align}
	since $\Pi_{S,0}$ is an orthogonal projection. Then, assume that Eq.~\eqref{gap:canonical-equation} holds. We have:
	\[
	\sum_{i\in S}\operatorname{Var}_i(X) = \sum_{i\in S} \|(\id-\Pi_{i,0})\widetilde{X}\|_{2}^2 = \Tr[\widetilde{X}^\dagger \mathbf{H}_S \widetilde{X}]
	\ge2\gamma_\beta\Tr[\widetilde{X}^\dagger (\id-\Pi_{S,0}) \widetilde{X}]=2\gamma_\beta\operatorname{Var}_S(X).
	\]
	For this reason it is sufficient to prove Eq.~\eqref{gap:canonical-equation}. In order to do so, we first show how to estimate the gap of the canonical purified Hamiltonian on a ring, and then move on to show how this bound can be lifted to an open chain.

	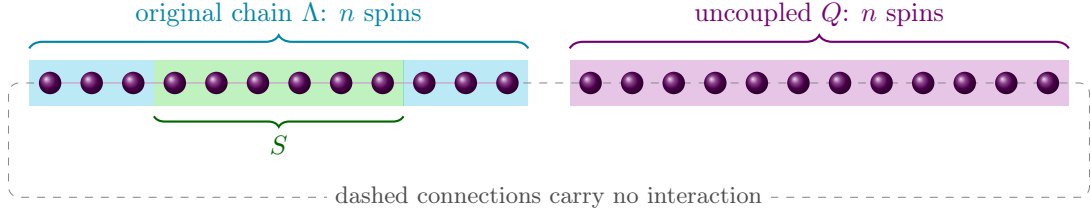
\begin{figure}
		\centering
		\begin{tikzpicture}[x=.55cm,y=.55cm,font=\small]
			\foreach \a/\b in {-.5/2.5,8.5/11.5}
			\fill[darkblue!30!cyan,opacity=.27] (\a,.45) rectangle (\b,1.55);
			\fill[midgreen,opacity=.25] (2.5,.45) rectangle (8.5,1.55);
			\fill[midpink,opacity=.23] (12.5,.45) rectangle (24.5,1.55);

			\draw[darkpurple!35,line width=.5pt] (0,1)--(11,1);
			\draw[black!45,dashed] (11,1)--(24,1);
			\draw[black!45,dashed,rounded corners=4pt]
			(0,1)--(-1,1)--(-1,-1.75)--(25,-1.75)--(25,1)--(24,1);
			\foreach \x in {0,...,11}{
				\shade[shading=ball,ball color=darkpurple] (\x,1) circle (.27);
				\shade[shading=ball,ball color=darkpurple] ({\x+13},1) circle (.27);
			}

			\foreach \a/\b/\lab/\col in
			{-.5/11.5/{original chain $\Lambda$: $n$ spins}/darkblue!30!cyan,
				12.5/24.5/{uncoupled $Q$: $n$ spins}/midpink}{
				\draw[thick,\col!75!black,decorate,decoration={brace,amplitude=5pt}]
				(\a,1.83)--(\b,1.83);
				\node[text=\col!75!black] at ({(\a+\b)/2},2.65) {\lab};
			}
			\draw[thick,midgreen!50!black,decorate,
			decoration={brace,mirror,amplitude=4pt}] (2.5,.2)--(8.5,.2);
			\node[text=midgreen!50!black] at (5.5,-.5) {$S$};
			\node[fill=white,inner sep=3pt,font=\footnotesize,text=black!70]
			at (12,-1.75) {dashed connections carry no interaction};
		\end{tikzpicture}
		\caption{Embedding the open chain into a ring with uncoupled spins.
			The green band marks $S\subseteq\Lambda$. The connecting lines indicate
			the ring geometry; all interactions involving $Q$ vanish.
			Every original interval occupies at most half the ring. The added
			spins tensor its projections with the identity, so the imported
			gap applies to the original Gibbs state and its entire exterior.}
		\label{gap:fig-padding}
	\end{figure}

	Let us begin by considering a ring of $n$ sites. In the proof of
	\cite[Theorem 25]{LuciaPerezGarciaPerezHernandez2025Gap} it is shown that
	for intervals satisfying $|I|\leq \lceil n/2\rceil+1$, we can lower bound the gap of the canonical purified Hamiltonian as:
	\begin{align}
		{\rm gap}(\mathbf{H}_I)\geq \gamma_0\,\prod_{k\ge1}
		\frac{1-\delta_k}{1+\lfloor(4/3)^k\rfloor^{-1}} \qquad\delta_k=C\exp\!\left[-a
		\sqrt{\left\lfloor\frac\mu9(9/8)^k\right\rfloor}\right]\,,\label{eq:414}
	\end{align}
	where $\gamma_0$ is the minimal gap on a region of size at most $\mu$, which is chosen to satisfy
	$\mu=\mu_\beta\ge18$ and $\delta_1\le1/2$, and we used \cite[Theorem 36]{LuciaPerezGarciaPerezHernandez2025Gap} to estimate $\delta_k$.
	Here $C,a>0$ depend only on $\beta,d,r_0,J$.

	Such an integer $\mu_\beta$ exists, because
	$\delta_1=C\exp[-a\sqrt{\lfloor\mu_\beta/8\rfloor}]\to0$ as
	$\mu_\beta\to\infty$. Then $0\le\delta_k\le1/2$ for all $k\ge1$,
	and both $\sum_k\delta_k$ and $\sum_k\lfloor(4/3)^k\rfloor^{-1}$
	converge. Using $\log(1-x)\ge-2x$ for $0\le x\le1/2$ gives
	\[
	\prod_{k\ge1}\frac{1-\delta_k}{1+\lfloor(4/3)^k\rfloor^{-1}}
	\ge\exp\!\left[-2\sum_{k\ge1}\delta_k
	-\sum_{k\ge1}\lfloor(4/3)^k\rfloor^{-1}\right]>0.
	\]
	This lower bound depends only on the interaction parameters
	and $\beta$, as in \cite[Remark 26]{LuciaPerezGarciaPerezHernandez2025Gap}.
	For $|I|\le\mu_\beta$, the base gap $\gamma_0$ already implies
	\eqref{eq:414}, since the product is at most one.
	On the other hand, Theorem 22 and Proposition 37 from \cite{LuciaPerezGarciaPerezHernandez2025Gap} imply that $\gamma_0>0$ has a lower bound that is also independent of the system size. This shows that, for any interval smaller than half the ring, we can bound the gap of the canonical purified Hamiltonian independently of the system size.

	At this point, consider an open chain of size $n$. Append to it a region $Q$, composed of $n$ $d$-dimensional spins. We set all interactions involving $Q$ to zero (see Fig.~\ref{gap:fig-padding}). Then, the corresponding Gibbs state is
	\[
	\widetilde\sigma=\sigma\otimes\tau_Q,
	\qquad \tau_Q=\one_Q/d^n,
	\]
	and its interaction range and strength obey the same bounds $r_0$ and $J$. Moreover,
	\[
	\widetilde{\mathcal V}_S(0)
	=\mathcal V_S(0)\otimes\mathcal B(\mathcal H_Q),
	\]
	since the multiplication by $\tau_Q^{1/2}$ maps
	$\mathcal B(\mathcal H_Q)$ onto itself. Thus, setting
	$\id_{\mathrm{HS},Q}$ to be the identity on the auxiliary Hilbert--Schmidt
	space,
	\[
	\widetilde\Pi_{S,0}=\Pi_{S,0}\otimes\id_{\mathrm{HS},Q},
	\qquad
	\widetilde{\mathbf{H}}_S
	=\mathbf{H}_S\otimes\id_{\mathrm{HS},Q}.
	\]
	Then, since $|\Lambda\cup Q| = 2n$, we can apply the estimate in Eq.~\eqref{eq:414} to any $S\subseteq \Lambda$. Denote by $2\gamma_\beta$ the resulting uniform lower bound. This proves
	\eqref{gap:canonical-equation} on the original system.
\end{proof}

\subsection{Amplitudes and the conditional gap}

As discussed in Remark~\ref{gap:global-memory}, Theorem~\ref{gap:conditional} is not enough to prove
 a uniform local KMS gap.

Modular averaging instead gives a lower bound involving an
average of exterior projections, sufficient for the buffered entropy
argument in Section~\ref{sec:buffered}. This is not a uniform local spectral
gap relative to the exact fixed algebra: the local heat-bath gap in the
two-qubit example of Section~\ref{cmlsi:section} tends to zero. For
$S=\Lambda$, the averaged estimate does yield the uniform global KMS gap.
To this end, let us introduce the rotated version of the algebra in Eq.~\eqref{gap:canonical-spaces}
\begin{equation}\label{eq:exterior-spaces}
	\mathcal V_{S}(t)=
	\{\Delta_\sigma^{-it}(O)\sigma^{1/2}:O\in\mathcal A_{S^c}\},
\end{equation}
and we define $\Pi_{S,t} = \Delta_\sigma^{-it} \Pi_{S,0} \Delta_\sigma^{it} $ to be the corresponding orthogonal projection. Also notice that for either sampler $a$, which satisfies KMS detailed balance, the operator $-\Gamma_\sigma^{1/2}\cL_{S}^a\Gamma_\sigma^{-1/2}
=\sum_{i\in S}
\bigl(-\Gamma_\sigma^{1/2}\cL_i^a\Gamma_\sigma^{-1/2}\bigr)$ is self-adjoint with respect to the Hilbert--Schmidt inner product, with the local terms satisfying
\[
-\Gamma_\sigma^{1/2}\cL_i^a\Gamma_\sigma^{-1/2}\ge0.
\]

With this notation, we have that
\begin{lemma}[Exact modular averaging]\label{lem:average}
	For a bounded finite-range interaction on an open chain, at every
	fixed $\beta>0$, there is $\gamma_\beta>0$, depending only on
	$\beta,d,r_0,J$, such that every interval $S\subseteq\Lambda$ satisfies
	\begin{equation}
		 -\Gamma_\sigma^{1/2}\cL_{S}^a\Gamma_\sigma^{-1/2}
		\ge\gamma_\beta\int_{\mathbb R}2\operatorname{sech}(2\pi t)
		(1-\Pi_{S,t})\,dt \qquad a\in\{\mathrm{HB},\mathrm{RHB}\}.                \label{eq:average}
	\end{equation}
	In particular, for $S=\Lambda$, this gives the KMS gap
	\begin{equation}
		-\Gamma_\sigma^{1/2}\cL_\Lambda^a\Gamma_\sigma^{-1/2}
		\ge\gamma_\beta(1-|\sigma^{1/2}\rangle
		\langle\sigma^{1/2}|)\qquad a\in\{\mathrm{HB},\mathrm{RHB}\}.
		\label{eq:kms-gap}
	\end{equation}
\end{lemma}
\begin{proof}
	First, notice that thanks to Prop.~\ref{pre:properties}, it is sufficient to prove the bound for the heat-bath sampler. Let us introduce the mapping $X\rightarrow\psi_{X} := \Gamma_\sigma^{1/2}(X)$. Then, since $\Delta_\sigma^{-1/4}(\psi_X)=X\sigma^{1/2} = \widetilde{X}$ (as defined in Sec.~\ref{sec:purifiedHamiltonian}), we can express the GNS norm in terms of the Hilbert--Schmidt one as
	\[
	\|X\|_{\mathrm{GNS},\sigma}^2 = \|\widetilde{X}\|_{2}^2 = \|\Delta_\sigma^{-1/4}(\psi_X)\|_{2}^2
	\]
	We can rewrite Eq.~\eqref{gap:conditional-equation} in this notation:
	\begin{align}
		-\operatorname{Re}\Tr[\psi_X^\dagger(\Delta_\sigma^{-1/2} \Gamma_\sigma^{1/2}\cL_S\Gamma_\sigma^{-1/2})(\psi_X)]\geq \gamma_\beta \Tr[\psi_X^\dagger(\Delta_\sigma^{-1/4}(1-\Pi_{S,0})\Delta_\sigma^{-1/4})(\psi_X)]\,.\label{eq:418}
	\end{align}
	Define $T:= -\operatorname{Re}(\Delta_\sigma^{-1/2} \Gamma_\sigma^{1/2}\cL_S\Gamma_\sigma^{-1/2})$ and $B=\Delta_\sigma^{1/4}T\Delta_\sigma^{1/4}$. Since Eq.~\eqref{eq:418} holds for any $\psi_{X}$, it implies the operator inequality 
	\[
	\begin{aligned}
		B\ge\gamma_\beta(1-\Pi_{S,0}).
	\end{aligned}
	\]
	This is the first ingredient of the proof. At this point, we need to relate $B$ to the original Lindbladian. To this end, denote by $\{|\omega_i\rangle\}$ an eigenbasis of $\Delta_\sigma$ with eigenvalues $\Delta_\sigma|\omega_i\rangle= e^{\omega_i}|\omega_i\rangle$, and for any operator $X$, we denote its coordinates by $X_{ij} := \langle\omega_i|X|\omega_j\rangle$. Then, direct calculation gives:
	\[
	B_{ij}=-\cosh\!\left(\frac{\omega_i -\omega_j}{4}\right)(\Gamma_\sigma^{1/2}\cL_{S}\Gamma_\sigma^{-1/2}\bigr)_{ij}.
	\]
	Inverting the relation above and using the Fourier identity
	\[
	2\int_{\mathbb R}\operatorname{sech}(2\pi t)e^{itx}\,dt
	=\operatorname{sech}(x/4)
	\]
	gives the exact relation
	\[
	\begin{aligned}
	-\Gamma_\sigma^{1/2}\cL_S\Gamma_\sigma^{-1/2}
	&=2\int_{\mathbb R}\operatorname{sech}(2\pi t)\Delta_\sigma^{-it}B\Delta_\sigma^{it}\,dt\\
	&\geq 2\gamma_\beta\int_{\mathbb R}\operatorname{sech}(2\pi t)(1-\Pi_{S,t})\,dt,
	\end{aligned}
	\]
	where in the last step we used the fact that conjugation by
	$\Delta_\sigma^{-it}$
	preserves positivity. This proves
	\eqref{eq:average}. Additionally, for $S=\Lambda$, the range is
	$\mathbb C\sigma^{1/2}$ and is fixed by all modular unitaries.
	The density $2\operatorname{sech}(2\pi t)$ integrates to one,
	which proves \eqref{eq:kms-gap}.
\end{proof}
 %
%
\section{Locality of exterior changes}\label{sec:loc}

In this section we use the special locality properties of one-dimensional short-range systems to prove that applying operations outside a buffered interval has little effect on the Petz sampler inside the interval. In particular, we are able to show the following:
\begin{theorem}[Exterior changes are almost fixed]\label{loc-exterior}
	Fix an interval $S$, a buffer $S[b]$ and $|t|\le1$, and consider an operator $v\in\mathcal{V}_{S[b]}(t)$, i.e., an operator that can be written as
	\[
	v=\Delta_\sigma^{-it}(O)\sigma^{1/2},\qquad
	\text{for some } O\in\mathcal A_{S[b]^c},\qquad S[b] :=\{x:\operatorname{dist}(x,S)\le b\}.
	\]
	Then, there exist constants $C,c,\kappa>0$, depending on the interaction parameters and the temperature (and independent
	of the chain length), for which it holds that
\begin{equation}\label{loc-exterior-fixed}
 \frac12\bigl\|\cP_{S*}(vv^\dagger)-vv^\dagger\bigr\|_1
       \le C e^{c|S|-\kappa b}\hs{v}^2.
\end{equation}
If $S[b]=\Lambda$, the left side is zero.
\end{theorem}

The proof of this fact is based on the concept of Araki expansionals~\cite{bcp} and on the work on effective Hamiltonians from~\cite{kuwahara2025clustering}. Both these ingredients are peculiar to one-dimensional systems. 

Before passing to the proof, let us fix the notation and state the three main lemmas used in the derivation. First, for simplicity, we absorb the inverse temperature into the Hamiltonian, which we denote by $K=\beta H=\sum_Z k_Z$, and write
\[
 \operatorname{diam}(Z)\le r_0,\qquad
 \sup_x\sum_{Z\ni x}\norm{k_Z}\le j,\qquad
 K_D=\sum_{Z\subseteq D}k_Z,\qquad
 \sigma=e^{-K}/Z_\Lambda.
\]
For a set $D$, we introduce the tracial conditional expectation
\[
  E^{\mathrm{tr}}_D(X)=\Tr_{D^c}[X]\otimes
                  \frac{\one_{D^c}}{\dim\mathcal H_{D^c}}.
\]
This is contractive and bimodular over $\mathcal A_D$. We also define the Araki expansional; for any $D$ this is given by
\begin{align}
	F_D(z) := e^{zK} e^{-z(K_D+K_{D^c})}\,, \label{eq:ArakiExpansionals}
\end{align}
where $z\in\mathbb{C}$. Finally, the last concept that we need is the conditional amplitude, which is given by
\begin{align}
	W := \sigma^{1/2} (\mathbf{1}_S \otimes \sigma_{S^c})^{-1/2}.\label{eq:conditionalAmplitude}
\end{align}

Direct calculation shows that we can write the Petz sampler as $\cP_{S*}(\rho)=W(\one_S\otimes\rho_{S^c})W^\dagger$.

With this notation in hand, we can state the three lemmas needed for the proof of Thm.~\ref{loc-exterior}. 

\begin{lemma}[Araki estimates]
\label{loc-input-expansionals}
For an interval $D$, let
$F_D(z)$ be the Araki expansional defined in Eq.~\eqref{eq:ArakiExpansionals}. Let $\mathcal B_D$ be the
region containing all the interaction terms
removed at the endpoints of $D$. For $|z|\le2$ and every fixed $\nu>0$, there exist constants $G,C_\nu>0$ independent of the chain length such that
\begin{equation}\label{loc-boundary-cocycle-tail}
 \norm{F_D(z)^{\pm1}}\le G,\qquad
 \norm{(\id- E^{\mathrm{tr}}_{\mathcal B_D[m]})F_D(z)^{\pm1}}
       \le C_\nu e^{-\nu m}.
\end{equation}

\end{lemma}

\begin{proof}
	First, let us report the following result from~\cite[Proposition 3.2]{bcp}: for $X$ supported on an interval $I$ and $|z|\le2$,
	\begin{equation}\label{loc-input-imaginary}
		\norm{e^{zK}Xe^{-zK}}\le G^{|I|}\norm X,\qquad
		\norm{(\id- E^{\mathrm{tr}}_{I[m]})e^{zK}Xe^{-zK}}
		\le C G^{|I|}\frac{G^m}{(\lfloor m/r_0\rfloor+1)!}\norm X.
	\end{equation}

	Then, define $K_0=K_D+K_{D^c}$, so that $F_D(z) := e^{zK} e^{-zK_0}$.
	The perturbation $(K-K_0)$ is a bounded sum of finite-range terms
	supported at no more than two cuts. Also introduce 
	$B(z)=e^{zK_0}(K-K_0)e^{-zK_0}$. Thanks to Eq.~\eqref{loc-input-imaginary}, there exist constants $\tilde{G}$ and $\tilde C$ such that
	\[
		\|B(z)\|_\infty\leq \tilde{G},\qquad\qquad \|(\id -  E^{\mathrm{tr}}_{\mathcal B_D[m]})B(z)\|_\infty \leq \varepsilon_m := \tilde{C} \frac{\tilde{G}^m}{(\lfloor m/r_0\rfloor+1)!}.
	\]
	Moreover, for any $0\le s\le1$,
	the expansional satisfies
	$\frac{d}{ds}F_D(sz)=zF_D(sz)B(sz)$, $F_D(0)=\one$. Then, by applying Grönwall’s inequality together with the bound on $\|B(z)\|_\infty$, we get $\|F_D(z)\| \leq e^{|z|\tilde{G}}$. The same argument can also be made for $ F_D(z)^{-1}$, since $\frac{d}{ds}F_D(sz)^{-1}=-zB(sz)F_D(sz)^{-1}$.  This proves the first inequality in Eq.~\eqref{loc-boundary-cocycle-tail} with $G = e^{2 \tilde{G}}$.

	At this point, let us introduce the notation $B_m(s) :=  E^{\mathrm{tr}}_{\mathcal B_D[m]}B(sz)$, and define the operator $F_m(s)$ to solve the differential equation $\frac{d}{ds}F_m(s)=zF_m(s)B_m(s)$, $F_m(0)=\one$. Since conditional expectations are contractive, $\|B_m(s)\|_\infty\leq \tilde{G}$, and  therefore $\|F_m(s)^{\pm1}\|_\infty \leq G$. For the inverse, use $\frac{d}{ds}F_m(s)^{-1}=-zB_m(s)F_m(s)^{-1}$ and $F_m(0)^{-1}=\one$. Writing $D_m(s)=\norm{F_D(sz)-F_m(s)}$, integration of the differential equation gives
	\begin{align*}
		D_m(s)
		&\le |z|\int_0^s\!\bigl(\tilde G D_m(t)+G\varepsilon_m\bigr)\,dt\\
		&\le2\tilde G\int_0^s D_m(t)\,dt+2G\varepsilon_m.
	\end{align*}
	Then, applying Grönwall’s inequality one more time gives that $\|F_D(sz)-F_m(s)\|_\infty$ decays as $\varepsilon_m$. Then, we can use this to bound the right-hand side of Eq.~\eqref{loc-boundary-cocycle-tail}:
	\begin{align}
		\norm{(\id- E^{\mathrm{tr}}_{\mathcal B_D[m]})F_D(z)} = \norm{(\id- E^{\mathrm{tr}}_{\mathcal B_D[m]})(F_D(z)-F_m(1))} \leq \tilde{G}' \varepsilon_m,
	\end{align}
	where we implicitly used the fact that $ E^{\mathrm{tr}}_{\mathcal B_D[m]}(F_m(1))=F_m(1)$, and that conditional expectations are contractive.
	Finally, the factorial tail can be bounded by
	$C_\nu e^{-\nu m}$ for every $\nu>0$, by choosing $C_\nu$ sufficiently large. 

	The case for $\norm{(\id- E^{\mathrm{tr}}_{\mathcal B_D[m]})F_D(z)^{-1}}$ follows from the decomposition
	\[
	F_m(1)^{-1}-F_D(z)^{-1}
	=F_D(z)^{-1}(F_D(z)-F_m(1))F_m(1)^{-1},
	\]
	together with the bounds on the norms of
	$ F_D(z)^{-1},F_m(1)^{-1}$.
	Since $F_m(1)^{-1}$ is supported on $\mathcal B_D[m]$, its localization error vanishes. Hence the inverse tail is bounded by
	\[
	\norm{(\id-E^{\mathrm{tr}}_{\mathcal B_D[m]})F_D(z)^{-1}}
	\le2G^2\norm{F_D(z)-F_m(1)}.
	\]
\end{proof}

\begin{lemma}[A common approximation in two norms]\label{loc-two-norms}
	Let $S$ be an interval and $Y=\Delta_\sigma^{-1/2}(X)$. Suppose
	that there is a constant $a>0$, independent of $r$, such that for every $r\ge0$
	\[
	\norm X,\norm Y\le M,\qquad
	\max_{Z\in\{X,Y\}}\norm{(\id -  E^{\mathrm{tr}}_{S[r]})Z}\le Me^{-ar}.
	\]
	Then $J= E^{\mathrm{tr}}_{S[b]}X$ satisfies
	\begin{equation}\label{loc-two-norms-conclusion}
		\max\{\norm{X-J},\norm{\Delta_\sigma^{-1/2}(X-J)}\}
		\le CMe^{-\kappa b},
	\end{equation}
	where $C,\kappa$ depend only on $a$ and the interaction bounds.
\end{lemma}
\begin{figure}
	\centering
	\begin{tikzpicture}[x=.55cm,y=.55cm,font=\small]
		\fill[darkblue!30!cyan,opacity=.14] (3.5,.45) rectangle (20.5,1.55);
		\fill[midgreen,opacity=.20] (7.5,.45) rectangle (16.5,1.55);
		\fill[midgreen,opacity=.24] (9.5,.45) rectangle (14.5,1.55);
		\foreach \a/\b in {1.5/5.5,18.5/22.5}
		\fill[midpink,opacity=.25] (\a,.45) rectangle (\b,1.55);

		\draw[darkpurple!35,line width=.5pt] (0,1)--(24,1);
		\foreach \x in {0,...,24}
		\shade[shading=ball,ball color=darkpurple] (\x,1) circle (.27);
		\node[anchor=east] at (-1,1) {$\Lambda$};
		\foreach \x in {3.5,20.5}
		\draw[black!55,dashed] (\x,.3)--(\x,1.7);

		\foreach \a/\b/\lab/\col in
		{1.5/5.5/F/midpink,7.5/16.5/{S[m]}/midgreen!50!black,
			18.5/22.5/F/midpink}{
			\draw[thick,\col,decorate,decoration={brace,amplitude=4pt}]
			(\a,1.85)--(\b,1.85);
			\node[text=\col] at ({(\a+\b)/2},2.65) {$\lab$};
		}
		\draw[thick,midgreen!50!black,decorate,
		decoration={brace,mirror,amplitude=4pt}] (9.5,.2)--(14.5,.2);
		\node[text=midgreen!50!black] at (12,-.5) {$S$};
		\draw[thick,darkblue!30!cyan,decorate,
		decoration={brace,mirror,amplitude=5pt}] (3.5,-1.05)--(20.5,-1.05);
		\node[text=darkblue!30!cyan] at (12,-1.9) {$D=S[b]$};
	\end{tikzpicture}
	\caption{Schematic localization regions for $X$,
		$Y=\Delta_\sigma^{-1/2}(X)$, and the expansional $F$.
		The green region $S[m]$, with $m=\lfloor b/3\rfloor$, is disjoint
		from the pink collars around the dashed boundary cuts of $D$.
		The corresponding local approximants commute, bounding $[F,Y]$.
		Identity~\eqref{loc-weighted-expectation} then controls
		the weighted error of the same local approximation $ E^{\mathrm{tr}}_DX$.}
	\label{loc-figure}
\end{figure}
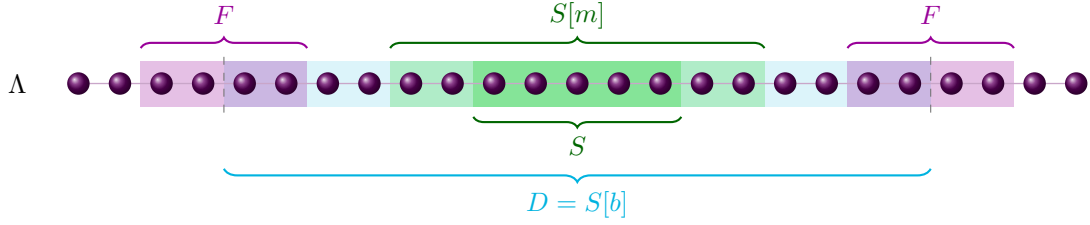
\begin{proof}
	Set $D=S[b]$, $K_0=K_D+K_{D^c}$ and $F:=F_D(1/2)=e^{K/2}e^{-K_0/2}$. Direct computation shows that $ E^{\mathrm{tr}}_D$
	commutes with the conjugation by $e^{K_0/2}$. Hence
	\begin{equation}\label{loc-weighted-expectation}
		\Delta_\sigma^{-1/2}(J) = \Delta_\sigma^{-1/2} E^{\mathrm{tr}}_{D}(\Delta_\sigma^{1/2}(Y)) = \Delta_\sigma^{-1/2}\Delta_{\sigma_0}^{1/2} E^{\mathrm{tr}}_{D}(\Delta_{\sigma_0}^{-1/2}\Delta_\sigma^{1/2}(Y))
		=F E^{\mathrm{tr}}_D(F^{-1}YF)F^{-1} ,
	\end{equation}
	where $\sigma_0$ is the Gibbs state associated to $K_0$. At this point, let us use the identity $F^{-1}YF=Y+F^{-1}[Y,F]$, together with the bounds from \eqref{loc-boundary-cocycle-tail}, to obtain
	\begin{equation}\label{loc-weighted-cut-bound}
		\norm{\Delta_\sigma^{-1/2}(X-J)}
		\le  C_G\bigl(\norm{(\id - E^{\mathrm{tr}}_D)Y}+\norm{[F,Y]}\bigr).
	\end{equation}
	By assumption, the first term contributes at most $Me^{-ab}$. For the second, let us approximate
	$Y$ by $E^{\mathrm{tr}}_{S[m]}Y$, and $F$ on $\mathcal B_D[m]$ by $E^{\mathrm{tr}}_{\mathcal B_D[m]}F$, with
	$m=\lfloor b/3\rfloor$. For sufficiently large $b$ these supports
	are disjoint, so the approximants commute and we obtain
	\[
	\norm{[F,Y]}\le CM(e^{-am}+e^{-m}).
	\]
	This proves the weighted bound. The ordinary bound, on the other hand, follows from the assumption. Still, to cover the case in which $b$ is small, it is sufficient to use $\norm{[F,Y]}\le2\norm F M$ in
	\eqref{loc-weighted-cut-bound} and to enlarge $C$ so as to cover this case as well.
\end{proof}

\begin{lemma}[Locality of the conditional amplitude]\label{loc-ordinary} Let $S$ be an interval. Let $W$ be the matrix defined in Eq.~\eqref{eq:conditionalAmplitude}. Then, it holds that
	\begin{equation}\label{loc-amplitude-bounds}
		\norm W\le C,\qquad \norm{W^{-1}}\le Ce^{c|S|},
		\qquad
		\norm{(\id- E^{\mathrm{tr}}_{S[b]})W}
		\le Ce^{c|S|-ab}.
	\end{equation}
	The same estimates hold for $W^\dagger$.
\end{lemma}
We defer the proof of this fact to Sec.~\ref{sec:locCondAmpl}.

\begin{proof}[Proof of Theorem~\ref{loc-exterior}]

	First, notice that if $S[b]=\Lambda$, then $O$ is scalar (say $O=c\one$) and $vv^\dagger=|c|^2\sigma$. Since $\cP_{S*}(\sigma)=\sigma$, this proves that in this case the norm is zero. For this reason, we henceforth assume that the exterior is nonempty.

	We begin by proving the case $t=0$. Since $O$ acts on $S^c$ only, the operators
	\[
	a=O\sigma^{1/2},\qquad a'=WO(\one_S\otimes\sigma_{S^c})^{1/2}
	\]
	are amplitudes of $\omega=O\sigma O^\dagger$ and
	$\cP_{S*}(\omega)$, respectively; that is, $\omega=aa^\dagger$ and $\cP_{S*}(\omega)=a'(a')^\dagger$. Moreover, we can rewrite their difference as
	\[
	a'-a=[W,O](\one_S\otimes\sigma_{S^c})^{1/2}
	=[W,O]\sigma^{1/2}(W^\dagger)^{-1}.
	\]
	At this point, notice that if $[J,O]=0$ and $\xi=O\sigma^{1/2}$, then
	\begin{equation}\label{loc-amplitude-commutator}
		\begin{split}
			[X,O]\sigma^{1/2}
			&=(X-J)\xi-\xi\Delta_\sigma^{-1/2}(X-J),\qquad\qquad\\
			&\implies\qquad\qquad\hs{[X,O]\sigma^{1/2}}
			\le\bigl(\norm{X-J}+
			\norm{\Delta_\sigma^{-1/2}(X-J)}\bigr)\hs\xi.
		\end{split}
	\end{equation}

	Let us choose $X=W$ and
	$Y=W^\dagger$. Thanks to Lemma~\ref{loc-ordinary}, these satisfy the hypotheses of Lemma~\ref{loc-two-norms}.
	The resulting approximation $J= E^{\mathrm{tr}}_{S[b]}W$ commutes with $O$. Then, we can use Eq.~\eqref{loc-amplitude-commutator} together with the inverse bound in
	\eqref{loc-amplitude-bounds} to obtain:
	\[
	\hs{a'-a}\le Ce^{c_1|S|} \hs{[W,O]\sigma^{1/2}}\leq Ce^{c|S|-\kappa b} \hs{a}.
	\]
	Here the final $c$ includes the growth from both the approximation and $\norm{W^{-1}}$. This allows us to prove the theorem for $t=0$. Indeed, we have
	\begin{equation}\label{loc-exterior-unrotated}
		\frac12\|\cP_{S*}(O\sigma O^\dagger)-O\sigma O^\dagger\|_1 = \frac12\|a'a'^\dagger-aa^\dagger\|_1\leq \frac12(\hs a+\hs{a'})\hs{a-a'}\le Ce^{c|S|-\kappa b}\hs{a}^2.
	\end{equation}	
	Now let $|t|\le1$. Cut the Hamiltonian halfway through the buffer and define
	\[
	\ell=\lfloor b/2\rfloor,\qquad I=S[\ell],\qquad
	K_0=K_I+K_{I^c},\qquad
	U=e^{itK},\quad U_0=e^{itK_0},\quad V=U_0^\dagger U.
	\]
	With this notation, $v=UOU^\dagger\sigma^{1/2}$. Moreover, define $O_0=U_0OU_0^\dagger$ and $v_0=O_0\sigma^{1/2}$. Since $U$
	commutes with $\sigma^{1/2}$, we obtain the exact identity
	\[
	v-v_0=U_0[V,O]\sigma^{1/2}U^\dagger
	\]
	so that, once again, we can use Eq.~\eqref{loc-amplitude-commutator} to obtain the bound
	\[
		\hs{v-v_0} = \hs{U_0[V,O]\sigma^{1/2}U^\dagger}\leq \hs{[V,O]\sigma^{1/2}} \leq \bigl(\norm{V-J}+
		\norm{\Delta_\sigma^{-1/2}(V-J)}\bigr)\hs\xi.
	\]
	At this point set $X = V$, and $Y = \Delta_\sigma^{-1/2}(X)$. Then, we have the expansional identities
	\begin{equation}\label{loc-rotation-cocycles}
		V=F_I(-it)^{-1},\qquad
		\Delta_\sigma^{-1/2}(V)
		=F_I(1/2)F_I(1/2-it)^{-1}.
	\end{equation}
	Together with \eqref{loc-boundary-cocycle-tail}, these show that both operators have
	uniformly bounded norms and tails $Ce^{-ar}$ outside $I[r]$. This allows us to apply Lemma~\ref{loc-two-norms} with the base interval $I$ and buffer $\ell$.
	It gives an approximation supported on $I[\ell]\subseteq S[b]$,
	with error $Ce^{-\kappa b}$ in both norms, after decreasing
	$\kappa$. Putting everything together, we obtain that
	$\delta:=\hs{v-v_0}\le Ce^{-\kappa b} \hs v$. We used $\hs\xi=\hs{O\sigma^{1/2}}=\hs v$, since the full modular rotation is unitary in Hilbert--Schmidt norm.

	Notice that the operator $O_0$ remains supported on $I^c$, so
	\eqref{loc-exterior-unrotated} applies with buffer $\ell$.
	Writing $\rho=vv^\dagger$ and $\rho_0=v_0v_0^\dagger$, trace-norm
	contraction of $\cP_{S*}$ gives
	\begin{align*}
		\frac12\|\cP_{S*}(\rho)-\rho\|_1
		&\le		\frac12\left(\|\rho-\rho_0\|_1+\|\cP_{S*}(\rho-\rho_0)\|_1+\|\cP_{S*}(\rho_0)-\rho_0\|_1\right)\\
		&\le( 2\hs{v}+\delta)\delta+Ce^{c|S|-\kappa \ell} ( \hs{v}+\delta)^2
		\le \tilde{C}e^{\tilde{c}|S|-\tilde{\kappa }b}\hs{v}^2,
	\end{align*}
	after adjusting the constants. This proves the theorem.
\end{proof}

\subsection{Locality of the conditional amplitude (Lemma~\ref{loc-ordinary})}\label{sec:locCondAmpl}

In order to prove Lemma~\ref{loc-ordinary}, we first need some auxiliary results. First, we report the following lemma from~\cite{bcp} without proof.
\begin{lemma}[{\cite[Corollaries 3.4 and 4.4]{bcp}}]
	For adjacent intervals $I,J$, set
	\[
	E_{I,J}=e^{-K_{IJ}}e^{K_I+K_J},\qquad
	\omega_I=e^{-K_I}/\Tr[e^{-K_I}],\qquad
	T_{I,J}=\Tr_I[\omega_I E_{I,J}].
	\]
	Then $\norm{E_{I,J}^{\pm1}},\norm{T_{I,J}},\norm{T_{I,J}^{-1}}\le C$.
\end{lemma}

Moreover, we need the characterization of effective Hamiltonians from~\cite{kuwahara2025clustering}:
\begin{lemma}[Boundary correction to the marginal logarithm]
	\label{loc-meanforce}
	For an interval $S$, there is a Hermitian $V_{S^c}$ such that
	\begin{equation}\label{loc-meanforce-identity}
		\log\sigma_{S^c}=-K_{S^c}+(\log Z_S-\log Z_\Lambda)\one+V_{S^c},
		\qquad Z_S=\Tr[e^{-K_S}],
	\end{equation}
	and
	\begin{equation}\label{loc-meanforce-tail}
		\norm{V_{S^c}}\le C,\qquad
		\norm{(\id- E^{\mathrm{tr}}_{S[b]})V_{S^c}}\le Ce^{-ab}.
	\end{equation}
	Moreover, $ E^{\mathrm{tr}}_{S[b]}(V_{S^c})$ is supported on $S^c\cap S[b]$.
\end{lemma}
\begin{proof}
	We begin by identifying $V_{S^c}$ by direct computation:
	\begin{align}
		\log\sigma_{S^c} &=\log \Tr_S[e^{-K}] - \mathbf{1}\log Z_\Lambda \nonumber
		\\
		&= - K_{S^c} + (\log Z_S - \log Z_\Lambda )\mathbf{1} + \left( \log \Tr_S[e^{-K}] - \log\Tr_S[e^{-(K_S + K_{S^c})}]\right)\,,
	\end{align}
	where $V_{S^c}$ is given by the term inside the parentheses. In order to bound its norm, let us first introduce the matrix $F:=F_{S}(-1/2)=e^{-K/2}e^{(K_S+K_{S^c})/2}$. Thanks to \eqref{loc-boundary-cocycle-tail}, we have that
	\[
	G^{-2}  \id\leq F^\dagger F \leq G^2\id \qquad\implies \qquad
	G^{-2}e^{-(K_S+K_{S^c})}\le e^{-K}\le G^2e^{-(K_S+K_{S^c})}.
	\]
	Since the partial trace and the logarithm preserve the ordering (as the latter is an operator monotone function), this also implies the inequality:
	\begin{align}
		   \log \Tr_S[e^{-(K_S + K_{S^c})}]-2\log G \leq \log \Tr_S[e^{-K}] \leq \log \Tr_S[e^{-(K_S + K_{S^c})}]+2\log G .
	\end{align}
	This is enough to prove the bound on $V_{S^c}$, as we have:
	\[
		\norm{V_{S^c}} = \norm{\log \Tr_S[e^{-K}] - \log\Tr_S[e^{-(K_S + K_{S^c})}]}\leq 2 \log G.
	\]

	For the tail we use the locality estimate for the marginal logarithm in
	\cite[Supplemental Material]{kuwahara2025clustering},
	applied with Hamiltonian $-K$ and inverse temperature one. It gives
	\begin{align}
		\sup_{u_x}\norm{[V_{S^c},u_x]}
		\le C e^{-a\operatorname{dist}(x,S)},\qquad x\in  S^c,\label{eq:kuwaharaBound}
	\end{align}
	where $u_x$ ranges over on-site unitaries. For an interior $S$,
	fold the chain about its midpoint: $S$ and ${S^c}$ become adjacent
	intervals, with local dimension at most $d^2$ and uniformly bounded
	range and strength. This also justifies subsequent applications of
	the adjacent-interval estimates to $S,{S^c}$. Endpoint versions follow
	by adjoining uncoupled spins.\footnote{If the cited statement is
		formulated with two nonempty traced-out tails, append a retained
		uncoupled segment longer than the original chain, followed by a
		traced-out spin. The artificial boundary is farther away than $S$
		from every original retained site. Removing the added spins changes
		the marginal logarithm only by a scalar.}

	Then, since the normalized partial trace can be realized by averaging over all unitaries on the region that one wants to trace out, we can obtain the tail bound as follows.
	Set $\mathcal T_x(V)=\int u_xVu_x^\dagger\,d\mu(u_x)$. The identity
	\[
	V-\mathcal T_x(V)=\int[V,u_x]u_x^\dagger\,d\mu(u_x)
	\]
	and a telescoping expansion of $\prod_{x\notin D}\mathcal T_x=E_D^{\mathrm{tr}}$, using contractivity of the other factors, give
	\[
	\norm{V-E_D^{\mathrm{tr}}V}\le\sum_{x\notin D}\sup_{u_x}\norm{[V,u_x]}.
	\]
	Hence
	\begin{align*}
	\norm{(\id- E^{\mathrm{tr}}_{S[b]})V_{S^c}}
	&\le\sum_{x\in S^c\setminus S[b]}\sup_{u_x}\norm{[V_{S^c},u_x]}\\
	&\le 2C\sum_{r>b}e^{-ar}\le C'e^{-a'b},
	\end{align*}
	where we used the bound in Eq.~\eqref{eq:kuwaharaBound}, together with the fact that
	there are at most two sites at each distance from $S$. This proves the claim.
\end{proof}

We are ready to prove Lemma~\ref{loc-ordinary}.

\begin{proof}[Proof of Lemma~\ref{loc-ordinary}]
We start by bounding $Q=\sigma(\mathbf{1}_S \otimes \sigma_{S^c})^{-1}$. Define 
\[
 E:= F_{S}(-1)  = e^{-K}e^{(K_S+K_{S^c})},\qquad
 \omega_S=e^{-K_S}/Z_S,\qquad T=\Tr_S[\omega_SE]=\Tr_S[e^{-K}]e^{K_{S^c}}/Z_S.
\]
Notice that thanks to Lemma~\ref{loc-input-expansionals}  and \cite[Corollary 4.4]{bcp}, both $\norm{E^{\pm1}}, \norm{T^{\pm1}}\leq G$ after enlarging $G$.

The weighted partial trace is contractive, which bounds $\norm T$ only. The inverse bound is \cite[Corollary 4.4, Eq.~(14)]{bcp}, applied to the adjacent pair $S,S^c$ after the folding construction above. Folding changes the local dimension, range and strength only by fixed factors, so its constants remain independent of $|S|$ and $n$.

Thanks to cyclicity of the trace, we have the identity $\sigma_{S^c}=(Z_S/Z_\Lambda)Te^{-K_{S^c}}$. Moreover, it also holds that $Q=E\omega_ST^{-1}$, which directly yields that $\norm Q\le G^2$. On the other hand, $Q^{-1}=T\omega_S^{-1}E^{-1}$. By using the fact that
$\norm{\omega_S^{-1}}\le d^{|S|}e^{2j|S|}$, we also have $\norm{Q^{-1}}\le G^2 e^{c|S|}$.

At this point, notice that $W^\dagger W = (\mathbf{1}_S \otimes \sigma_{S^c})^{-1/2}\sigma(\mathbf{1}_S \otimes \sigma_{S^c})^{-1/2} = (\mathbf{1}_S \otimes \sigma_{S^c})^{-1/2}Q(\mathbf{1}_S \otimes \sigma_{S^c})^{1/2}$. Given that this is a similarity transformation, this also means that $\norm{W}^2=r(Q)\le G^2$, where $r$ is the spectral radius. Similarly, we also have $\norm{W^{-1}}^2\le r(Q^{-1})\le G^2e^{c|S|}$.

In order to study the locality properties of $W$, define $W(z)=\sigma^z(\mathbf{1}_S \otimes \sigma_{S^c})^{-z}$ and the time evolution $\alpha_t(X)=e^{-itK}Xe^{itK}$. Then, the evolution $W(it)$ solves the differential equation $\frac{d}{dt}W(it)=i\alpha_t(V)W(it)$, with $V$ given by:
\[
 V=\log\sigma-\log(\mathbf{1}_S \otimes \sigma_{S^c})  
   =-(K-K_{S^c})-\log Z_S\,\one-V_{S^c}\,,
\]
where we applied Lemma~\ref{loc-meanforce} in the second equality. This operator $V$ is a sum of $\mathcal{O}(1+|S|)$ bounded finite-range terms, a scalar, and the quasilocal bounded correction $-V_{S^c}$. Then, it directly follows from the Lieb--Robinson bound~\cite{lr} (for $X$ supported on an interval $I\subset D$),
\[
	\norm{(\id- E^{\mathrm{tr}}_D)\alpha_t(X)}
	\le C_\mu |I|\norm X
	e^{-\mu\operatorname{dist}(I,D^c)+v_\mu|t|},
\]
that $V$ has the locality property:
\begin{equation}\label{loc-generator-tail}
	\norm{(\id- E^{\mathrm{tr}}_{S[b]})\alpha_t(V)}
	\le C(1+|S|)e^{v|t|-ab}.
\end{equation}
Indeed, this can be obtained by applying the Lieb--Robinson bound separately to the finite-range terms and to
$ E^{\mathrm{tr}}_{S[\lfloor b/2\rfloor]}V_{S^c}$ (notice that we can absorb the support factor $|S|+b$ by
decreasing the spatial exponent $a$). Then, it follows from Lemma~\ref{loc-meanforce} that the error of approximating $V_{S^c}$ with $ E^{\mathrm{tr}}_{S[\lfloor b/2\rfloor]}V_{S^c}$ is upper-bounded by
$\mathcal{O}(e^{-ab/2})$. 

Then, using Duhamel's formula and
\eqref{loc-generator-tail}, we get that
\begin{align}
	 \norm{(\id- E^{\mathrm{tr}}_{S[b]})W(it)}
	\le C(1+|S|)e^{v'|t|-ab}.\label{eq:21}
\end{align}
Now, consider the evolution $W(1+it)=\sigma^{it}Q(\one_S\otimes\sigma_{S^c})^{-it}$. This has norm at most $G^2$. Moreover,
$W(x+it)$ is bounded uniformly in $t$ for $0\le x\le1$. Then, we can apply the three-lines theorem to 
$e^{z^2}(\id- E^{\mathrm{tr}}_{S[b]})W(z)$ for $z=1/2$ to obtain
\begin{equation}\label{loc-ordinary-tail}
 \norm{(\id- E^{\mathrm{tr}}_{S[b]})W}
       \le G\sqrt{1+|S|}\,e^{-ab/2}.
\end{equation}
It should be noticed that the Gaussian factor $e^{z^2}$ was introduced to control the growth in $|t|$ of the exponential on the right-hand side of Eq.~\eqref{eq:21}. By absorbing the dependence on the size of $S$ into the exponent, we obtain the claim for $W$. The same result also holds for $W^\dagger$, since taking the adjoint does not change the argument.
\end{proof}

\subsection{Quasi-locality of the heat-bath sampler}
\label{sec:heatbath-quasilocal}

The conditional-amplitude estimate of Lemma~\ref{loc-ordinary} also
provides local channel approximations to each heat-bath update, with
exponentially decaying error.

\begin{proposition}[Quasi-locality of the heat-bath sampler]
	\label{loc:heatbath-quasilocal}
	There are constants $C_\beta,\kappa_\beta>0$, independent of the
	chain length, such that for every site $i$ and integer $b\ge0$ there
	is a channel $\cP_{i*}^{(b)}$ supported on
	$D=\{i\}[b]$ satisfying
	\begin{equation}\label{loc:heatbath-channel-tail}
		\|\cP_{i*}-\cP_{i*}^{(b)}\|_\diamond
		\le C_\beta e^{-\kappa_\beta b}.
	\end{equation}
	Here supported on $D$ means that the channel is a tensor product
	of a channel on $D$ with the identity channel on $D^c$, and
	$\|\cdot\|_\diamond$ is the trace norm of a map maximised over
	inputs of trace norm one and arbitrary finite-dimensional memories.
	Consequently the heat-bath generator
	$\cL_*^{\mathrm{HB}}=\sum_i(\cP_{i*}-\id)$ is quasi-local, with
	uniform exponential tails for its individual terms.
\end{proposition}
\begin{proof}
	In Eq.~\eqref{eq:conditionalAmplitude}, choose $S=\{i\}$, and let
	$W_b= E^{\mathrm{tr}}_D W$. Then, thanks to Lemma~\ref{loc-ordinary},  it holds that
	\[
	\|W-W_b\|_\infty\le C_\beta e^{-\kappa_\beta b},
	\qquad \|W\|_\infty,\|W_b\|_\infty\le C_\beta.
	\]
	Moreover, the amplitude satisfies
	\[
	\Tr_i[W^\dagger W]
	=\sigma_{i^c}^{-1/2}\Tr_i[\sigma]\sigma_{i^c}^{-1/2}
	=\one_{i^c}.
	\]
	Set $A_b=\Tr_i[W_b^\dagger W_b]$. This is supported on
	$D\setminus\{i\}$, and it satisfies
	\[
	\|A_b-\one_{i^c}\|_\infty
	\le d(\|W_b\|_\infty+\|W\|_\infty)\|W_b-W\|_\infty
	\le C_\beta e^{-\kappa_\beta b}.
	\]
	Choose $b$ large enough so that the right-hand side is bounded above by $1/2$. Since $A_b$ is positive, this also means that its spectrum is contained in the interval $[1/2,3/2]$. Then, define
	\[
	R_b=W_b(\one_i\otimes A_b^{-1/2}),\qquad
	\cP_{i*}^{(b)}(\rho)
	=R_b(\one_i\otimes\Tr_i[\rho])R_b^\dagger.
	\]
	The map is completely positive and supported on $D$.
	Moreover, since $\Tr_i[R_b^\dagger R_b]=\one_{i^c}$, it is trace preserving. Then, thanks to the properties of the spectrum of $A_b$, and the fact that $f(x) = x^{-1/2}$ is Lipschitz, it follows that $\|A_b^{-1/2}-\one_{i^c}\|_\infty
	\le C\|A_b-\one_{i^c}\|_\infty$. This implies that:
	\[
	\|R_b-W\|_\infty\le \|W_b- W\| + \|W_b (\one_i\otimes (A_b^{-1/2}-\one_{i^c})) \|\le C_\beta e^{-\kappa_\beta b}.
	\]
	Also notice that the map $\rho\mapsto\one_i\otimes\Tr_i[\rho]$ has diamond
	norm $d$. Expanding the difference between the two amplitude
	formulas, also after tensoring with any memory, therefore yields
	\[
	\|\cP_{i*}-\cP_{i*}^{(b)}\|_\diamond
	\le d(\|W\|_\infty+\|R_b\|_\infty)\|W-R_b\|_\infty
	\le C_\beta e^{-\kappa_\beta b}.
	\]
	This proves the bound for the case of large $b$. On the other hand, if $b$ is small, use the on-site reset
	channel $\rho\mapsto(\one_i/d)\otimes\Tr_i[\rho]$.
	The diamond distance between two channels is at most two, so
	enlarging $C_\beta$ proves the same bound for all $b\ge0$.
	Subtracting the identity channel gives the identical approximation
	bound for each generator term $\cP_{i*}-\id$.
\end{proof}

\begin{remark}[Regularisation need not preserve quasi-locality]
	\label{loc:rhb-nonlocal}
	The limit $\mathbb{E}_i=\lim_{m\to\infty}\cP_i^m$ need not inherit the
	quasi-locality of $\cP_i$. For a concrete example, consider the
	nearest-neighbour qubit chain
	\[
	H=\sum_{j=1}^{n-1}(X_jX_{j+1}+Z_jZ_{j+1}),\qquad n\ge3,
	\]
	where $X_j,Z_j$ are Pauli matrices, and fix any $\beta>0$.
	We claim that $\mathcal N_1=\mathbb C\one$ (where $\mathcal{N}_1$ is the fixed algebra of $\mathbb{E}_1$), so that the regularised
	update at the first site resets the entire chain:
	\[
	\mathbb{E}_1(A)=\Tr[\sigma A]\one,\qquad \mathbb{E}_{1*}(\rho)=\sigma
	\quad\text{for every state }\rho.
	\]
	Indeed, it follows from Proposition~\ref{pre:properties} that $\mathcal N_1\subseteq\mathcal A_{\{2,\ldots,n\}}$. Since $\mathbb{E}_1$ is a
	$\sigma$-preserving conditional expectation, its range is invariant
	under the modular group
	\cite[Proposition 1(iv)]{BardetCapelRouze2021NoncommutingAT}.
	Differentiating this invariance gives
	$[H,A]\in\mathcal N_1$ for every $A\in\mathcal N_1$.
	Suppose inductively that
	$\mathcal N_1\subseteq\mathcal A_{\{k+1,\ldots,n\}}$.
	The only terms of $[H,A]$ acting nontrivially on site $k$ are
	\[
	X_k[X_{k+1},A]+Z_k[Z_{k+1},A].
	\]
	Since $[H,A]\in\mathcal N_1$ acts trivially on site $k$, linear
	independence of $\one_k,X_k,Z_k$ forces both commutators to vanish.
	Thus $A$ also acts trivially on site $k+1$. Induction from $k=1$
	to $n-1$ proves the claim.

	For any channel $T_*$ supported on $\{1\}[b]$ with $b<n-1$,
	its adjoint satisfies $T(Z_n)=Z_n$. Trace-norm duality therefore gives
	\[
	\|\mathbb{E}_{1*}-T_*\|_\diamond
	\ge\|\mathbb{E}_1(Z_n)-T(Z_n)\|_\infty
	=\|\Tr[\sigma Z_n]\one-Z_n\|_\infty\ge1.
	\]
	Hence no approximation error tending to zero with $b$, uniformly
	in the chain length, is possible for this regularised update or
	its generator term $\mathbb{E}_{1*}-\id$. Taking infinitely many iterations
	can therefore destroy the uniform locality of the original update.
\end{remark}

\subsection{Conditional entropy of exterior changes}
\label{sec:kernel}

Theorem~\ref{loc-exterior} proves that the Petz update approximately fixes states with an amplitude $v\in\mathcal{V}_{S[b]}(t)$. We now convert this trace-distance estimate into a conditional entropy bound by data processing and
continuity of the von Neumann entropy.

\begin{lemma}[Entropy of the exterior-amplitude subspaces]
	\label{lem:kernel-entropy}
	Fix an interval $S$, a buffer $S[b]$ and $|t|\le1$, and consider a normalised amplitude $v\in\mathcal{V}_{S[b]}(t)$, i.e., an operator satisfying
	\[
	v=\Delta_\sigma^{-it}(O)\sigma^{1/2},
	\qquad O\in\mathcal A_{S[b]^c},\qquad \hs v=1.
	\]
	There are constants $C,c,\kappa>0$ depending only on the fixed
	interaction parameters such that, if
	\begin{equation}
		q_b:=C e^{c|S|-\kappa b}\le\tfrac12,                   \label{eq:qb}
	\end{equation}
	then it holds that
	\begin{equation}
		D_S(vv^\dagger\Vert\sigma)\le c_\beta^{\rm ent} n q_b+h(q_b),                  \label{eq:kernel-entropy}
	\end{equation}
	where $c_\beta^{\rm ent}=(\log d+2\beta J)$ and $h(q)=-(q\log q+(1-q)\log(1-q))$. If $S[b]=\Lambda$, the left side is zero.
\end{lemma}
\begin{proof}
	Let $\omega=vv^\dagger$ and $\omega'=\cP_{S*}(\omega)$. Theorem~\ref{loc-exterior} implies that 
	\[
	q=\tfrac12\|\omega-\omega'\|_1\le q_b.
	\]
	Notice that the Petz recovery channel $P_{S,\sigma}$ maps $\omega_{S^c}$ to $\omega'$ and
	$\sigma_{S^c}$ to $\sigma$. Thanks to the data processing inequality, it holds that
	$D(\omega_{S^c}\Vert\sigma_{S^c})\ge D(P_{S,\sigma}(\omega_{S^c})\Vert P_{S,\sigma}(\sigma_{S^c}))= D(\omega'\Vert\sigma)$. Then, we can use this bound to give the estimate
	\begin{align*}
		D_S(\omega\Vert\sigma)&= D(\omega\Vert\sigma)-D(\omega_{S^c}\Vert\sigma_{S^c})
		\le D(\omega\Vert\sigma)-D(\omega'\Vert\sigma)
		\\
		&=S(\omega')-S(\omega)+\beta\Tr[(\omega-\omega')H_\Lambda] + \cancel{\log Z_{\Lambda}\Tr[(\omega-\omega')]}
		\\
		&\le qn\log d+h(q)+\beta \|\omega-\omega'\|_1\norm{H_\Lambda}
		\\
		&\le (\log d+2\beta J)nq+h(q),
	\end{align*}
	where in the second line we used the fact that $\omega' = \cP_{S*}(\omega)$ and $\cP_{S*}$ is trace preserving, while the following inequality follows from Hölder's inequality, together with Audenaert's continuity
	bound \cite[Theorem 1]{entropy-continuity}, which states
	\[
	|S(\omega')-S(\omega)|\le q\log(d^n-1)+h(q)\le qn\log d+h(q).
	\]
	Finally, the last line uses
	$\norm{H_\Lambda}\le Jn$. Since $q\le q_b\le1/2$, monotonicity
	of $h(q)$ on $[0,1/2]$ proves Eq.~\eqref{eq:kernel-entropy}.
	If $S[b]=\Lambda$, $\mathcal{V}_{S[b]}(t) = \mathbb{C} \sigma^{1/2}$, so that 
	$vv^\dagger=\sigma$. Then, $D_S(\omega\Vert\sigma)=0$.
\end{proof}

\begin{figure}
	\centering
	\begin{tikzpicture}[x=.82cm,y=.7cm,font=\small]
		\fill[gray!12] (0,0) rectangle (2,1);
		\fill[blue!12] (2,0) rectangle (5,1);
		\fill[orange!25] (5,0) rectangle (9,1);
		\fill[blue!12] (9,0) rectangle (12,1);
		\fill[gray!12] (12,0) rectangle (14,1);
		\draw (0,0) rectangle (14,1);
		\foreach \x in {2,5,9,12} {\draw (\x,0)--(\x,1);}
		\node at (1,.5) {$O_{(S[b])^c}$};
		\node at (3.5,.5) {buffer};
		\node at (7,.5) {core $S$};
		\node at (10.5,.5) {buffer};
		\node at (13,.5) {$O_{(S[b])^c}$};
		\draw[decorate,decoration={brace,mirror,amplitude=5pt}]
		(2,-.2)--(12,-.2) node[midway,below=6pt] {$S[b]$: centres of the updates used};
		\draw[decorate,decoration={brace,amplitude=4pt}]
		(2,1.2)--(5,1.2) node[midway,above=5pt] {$b$ sites};
		\draw[decorate,decoration={brace,amplitude=4pt}]
		(9,1.2)--(12,1.2) node[midway,above=5pt] {$b$ sites};
	\end{tikzpicture}
	\caption{The buffer separates the core from the exterior change.
		A unit amplitude $v=O_{(S[b])^c}\sigma^{1/2}$ represents the state
		$O_{(S[b])^c}\sigma O_{(S[b])^c}^\dagger$. Its conditional entropy in $S$
		is small even if the exterior change is large. The same estimate
		holds after bounded modular rotations.}
\end{figure}
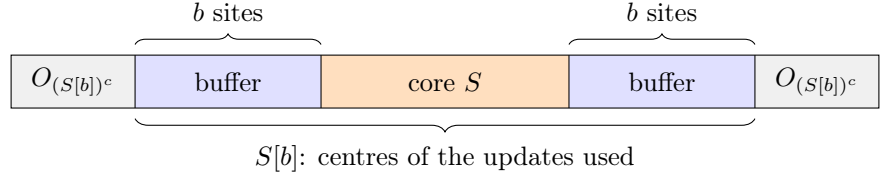

\begin{corollary}[Shielding at logarithmic scales]\label{main:shielding}
	For $n\ge2$, an interval $S$, $L\ge\log n$, $|S|\le6L$, and a sufficiently large fixed
	$k_\beta$, put $b=\lceil k_\beta L\rceil$. Every normalised amplitude
	$v\in\mathcal V_{S[b]}(t)$, $|t|\le1$, satisfies
	\begin{equation}\label{eq:shielding-statement}
		D_S(vv^\dagger\Vert\sigma)\le a_\beta/n.
	\end{equation}
	The constant $a_\beta$ is independent of $n$, $S$, and $t$.
\end{corollary}
\begin{proof}
	This follows directly from Lemma~\ref{lem:kernel-entropy}. Indeed, notice that since $|S|\le6L$ and $L\ge\log n$, there exists
	a fixed choice of $b=\lceil k_\beta L\rceil$ such that
	$ q_b\le n^{-2}$.
	Specifically, for the constants in Eq.~\eqref{eq:qb}, choose
	\[
	k_\beta\ge\frac{6c+2+\log(\max\{C,1\})/\log2}{\kappa}.
	\]
	Then, for every $n\ge2$,
	\[
	q_b\le C e^{-(\kappa k_\beta-6c)L}
	\le Cn^{-2-\log(\max\{C,1\})/\log2}\le n^{-2}\le\tfrac14.
	\]
	Then, Eq.~\eqref{eq:kernel-entropy} implies:
	\[
	D_S(vv^\dagger\Vert\sigma)\le
	c_\beta^{\rm ent} n q_b+h(q_b)
	\le \frac{c_\beta^{\rm ent}+3}{n},
	\]
	where we used the fact that $h(q)\le q\log(e/q)$ and the upper bound $\log n\le n$. 
	This proves the claim.
\end{proof}
 %
\section{Conditional entropy controlled by buffered energy}\label{sec:buffered}
In this section we show how to combine the ingredients from Sections~\ref{gap:section} and~\ref{sec:loc} to prove that the conditional relative entropy in an interval is
controlled by the quadratic energy of either sampler's updates centred in a larger interval,
with an error of order $1/n$.

\begin{theorem}[Conditional entropy controlled by buffered energy]
\label{ent:buffered}
We fix the interaction parameters and $\beta>0$ as in
Section~\ref{sec:model}. There are constants
$k_\beta,C_\beta>0$, independent of the chain length and the position
of the interval, with the following property. For every $n\ge2$,
$L\ge\log n$, and nonempty interval $S\subseteq\Lambda$ with
$|S|\le6L$, every state $\rho$
satisfies
\begin{equation}\label{eq:buffered}
 D_S(\rho\Vert\sigma)
 \le C_\beta L\,\mathcal Q_{S[b]}^{a}(\rho)
                                      +\frac{C_\beta}{n},\qquad a\in\{\mathrm{HB},\mathrm{RHB}\} 
\end{equation}
with $b=\lceil k_\beta L\rceil$.
The energy is the sum of the single-site quadratic forms of sampler $a$ with
centres in $S[b]$, all using the original global Gibbs state $\sigma$.
\end{theorem}

Before proving this result, we need one last ingredient:

\begin{lemma}[Entropy of a superposition]\label{lem:superposition}
	Assume that $S$ is a nonempty interval and $\eta\ge0$.
	Let $\Pi$ be a Hilbert--Schmidt orthogonal projection such that
	$D_S(vv^\dagger\Vert\sigma)\le\eta$ for every unit vector
	$v\in\Ran\Pi$. For every unit Hilbert--Schmidt amplitude $\psi$
	and $0<\varepsilon<1$,
	\begin{equation}
		D_S(\psi\psi^\dagger\Vert\sigma)
		\le2\eta+2(M_S+\log\varepsilon^{-1})\hs{(1-\Pi)\psi}^2
		+2\log(1-\varepsilon)^{-1},         \label{eq:superposition}
	\end{equation}
	where $M_S=c_\beta(|S|+1)$ and $c_\beta>0$ does not depend on the size of the chain.
\end{lemma}

\begin{proof}
	We use a construction akin to the one used in \cite{superpositions}.
		First suppose $0<x:=\hs{(1-\Pi)\psi}^2<1$.
		Let us decompose any unit Hilbert--Schmidt amplitude $\psi$ as
	$\psi=\sqrt{1-x}\,v+\sqrt x\,w$, where $v,w$ are
		orthogonal unit amplitudes, chosen as $v=\Pi\psi/\sqrt{1-x}$ and $w=(1-\Pi)\psi/\sqrt x$, with $v\in\Ran\Pi$, and
	$x=\hs{(1-\Pi)\psi}^2$. Introduce a purifying space $E$, and a qubit register $F$, and define
	\[
	|\Omega\rangle=\sqrt{1-x}\,|0\rangle_F|v\rangle_{\Lambda E}
	+\sqrt x\,|1\rangle_F|w\rangle_{\Lambda E},
	\]

		Here $E\simeq\mathcal H_\Lambda$, and both vectors use the same canonical vectorisation $|a\rangle_{\Lambda E}:=\sum_j a|j\rangle_\Lambda\otimes|j\rangle_E$ in a fixed orthonormal basis. Thus
		\[
		\Tr_E[|a\rangle\langle b|]=ab^\dagger,
		\qquad \langle a|b\rangle=\Tr[a^\dagger b].
		\]
		This preserves both the coherence $vw^\dagger$ and the orthogonality of $v,w$; in particular, $\Omega_F=(1-x)|0\rangle\langle0|+x|1\rangle\langle1|$.

	Define $\Omega_{\Lambda F} = \Tr_E[|\Omega\rangle\langle \Omega|]$ and $\Omega_F = \Tr_{\Lambda}[\Omega_{\Lambda F}]$. Let us explicitly expand the conditional relative entropy:
	\begin{align}
		D_S(\Omega_{\Lambda F}\Vert\sigma\otimes\Omega_F)
		= D_S(\Omega_\Lambda || \sigma) + I(S:F|S^c)_{\Omega_{\Lambda F}}\,, \label{eq:63}
	\end{align}
	where we have introduced the conditional mutual information. Since the latter satisfies a data processing inequality, for any channel of the form  $(\id_\Lambda\otimes\mathcal{N}_F)$ we have:
	\begin{align}
		D_S((\id_\Lambda\otimes\mathcal{N}_F)(\Omega_{\Lambda F})\Vert\sigma\otimes\mathcal{N}_F(\Omega_F))\leq D_S(\Omega_{\Lambda F}\Vert\sigma\otimes\Omega_F).\label{memory-data-processing}
	\end{align}
	Then, choose $\mathcal{N}_F(\rho) :=|+\rangle\langle +|\rho|+\rangle\langle +| + |-\rangle\langle -|\rho |-\rangle\langle -|$, where $|\pm\rangle := \frac{1}{\sqrt{2}} (|0\rangle  \pm |1\rangle )$. Applying this channel to register $F$ gives:
	\begin{align}
		(\id_\Lambda\otimes\mathcal{N}_F)(\Omega_{\Lambda F}) = \tfrac 12 (|+\rangle \langle +|\otimes \psi_+\psi_+^\dagger + |-\rangle \langle -|\otimes \psi_-\psi_-^\dagger )
	\end{align}
	where $\psi_{\pm}:=\sqrt{1-x}\,v\pm\sqrt x\,w$, which directly implies that:
	\begin{align}
		\frac{1}{2}D_S(\psi\psi^\dagger\Vert\sigma) &\leq \frac{1}{2}(D_S(\psi_+\psi_+^\dagger\Vert\sigma)+D_S(\psi_-\psi_-^\dagger\Vert\sigma)) =  D_S((\id_\Lambda\otimes\mathcal{N}_F)(\Omega_{\Lambda F})\Vert\sigma\otimes\mathcal{N}_F(\Omega_F))  \nonumber
		\\
		&\leq D_S(\Omega_{\Lambda F}\Vert\sigma\otimes\Omega_F)\label{eq:6.6}
	\end{align}
	This implies that bounding the conditional entropy in Eq.~\eqref{eq:6.6} is enough to prove the claim.

	Now rewrite Eq.~\eqref{eq:6.6} as:
	\[
	D_S(\Omega_{\Lambda F}\Vert\sigma\otimes\Omega_F)=- S(\Omega_{\Lambda F})+S(\Omega_{S^c F})
	-\Tr\bigl[\Omega_{\Lambda}(\log\sigma-\log\sigma_{S^c})]
	\]

		Dephasing $F$ in the computational basis increases $S(\Omega_{\Lambda F})$ by at most the binary entropy $h(x)$ and does not decrease $S(\Omega_{S^cF})$. It leaves $\Omega_\Lambda$ unchanged, and its conditional relative entropy is the average of those of the two branches. Therefore
	\[
	D_S(\Omega_{\Lambda F}\Vert\sigma\otimes\Omega_F)
	\le(1-x)D_S(vv^\dagger\Vert\sigma)
	+xD_S(ww^\dagger\Vert\sigma)+h(x).
	\]
		By assumption $D_S(vv^\dagger\Vert\sigma)\leq \eta$. On the other hand, we can establish that $D_S(ww^\dagger\Vert\sigma)\leq M_S$.

		Indeed, for any state $\omega$, write $\omega_{S^c}:=\Tr_S\omega$. The entropy-difference bound and Lemma~\ref{loc-meanforce} give
		\begin{align*}
		D_S(\omega\Vert\sigma)
		&=-S(\omega)+S(\omega_{S^c})
		-\Tr[\omega(\log\sigma-\one_S\otimes\log\sigma_{S^c})]\\
		&\le |S|\log d+\|\log\sigma-\one_S\otimes\log\sigma_{S^c}\|_\infty\\
		&\le c_\beta(|S|+1)=M_S.
		\end{align*}
		For the last inequality, use $K=\beta H_\Lambda$ and the mean-force identity to write the logarithmic difference as $-(K-K_{S^c})-\log Z_S\,\one-V_{S^c}$. Here $\|K-K_{S^c}\|_\infty\le\beta J|S|$, $|\log Z_S|\le(\log d+\beta J)|S|$, and $\|V_{S^c}\|_\infty$ is bounded independently of the chain. Enlarging $c_\beta$ gives the stated $M_S$. Applying this to $\omega=ww^\dagger$ proves the claim, including singular $ww^\dagger$: only the faithful reference states occur inside logarithms.

		Finally, notice that nonnegativity of binary relative entropy gives
		$h(x)\le x\log\varepsilon^{-1}+\log(1-\varepsilon)^{-1}$. Restoring $x=\hs{(1-\Pi)\psi}^2$ gives the claim for $0<x<1$. For $x=0$, the hypothesis directly gives $D_S(\psi\psi^\dagger\Vert\sigma)\le\eta$; for $x=1$, the universal bound just proved gives $D_S(\psi\psi^\dagger\Vert\sigma)\le M_S$. Both imply \eqref{eq:superposition}, including when $\Pi$ is zero or the identity.
\end{proof}

\begin{proof}[Proof of Theorem~\ref{ent:buffered}]
		Fix $a\in\{\mathrm{HB},\mathrm{RHB}\}$. We introduce the notation $\Pi_t:=\Pi_{S[b],t}$ and $x_t :=\|(1-\Pi_t)\sqrt\rho\|_2^2 $.  For $0<x_t<1$, we can split the canonical amplitude $\sqrt{\rho}$ as:
	\[
	\sqrt\rho
	=\sqrt{1-x_t}\,v_t+\sqrt{x_t}\,w_t,
	\]
	where
	\[
	v_t
	:=\frac{\Pi_t\sqrt\rho}{\sqrt{1-x_t}}
	\in\Ran\Pi_t=\mathcal V_{S[b]}(t),
	\qquad
	w_t
	:=\frac{(1-\Pi_t)\sqrt\rho}{\sqrt{x_t}},
	\]
		and $v_t\perp w_t$ are orthogonal unit Hilbert--Schmidt amplitudes. If $x_t=0$ or $x_t=1$, we use the endpoint cases of Lemma~\ref{lem:superposition}, without defining the amplitude whose coefficient vanishes. Notice that thanks to Corollary~\ref{main:shielding}, every unit amplitude
	$v\in\mathcal V_{S[b]}(t)$, for $|t|\le1$, satisfies
	\[
	D_S(vv^\dagger\Vert\sigma)\le\frac{a_\beta}{n}.
	\]
	Then, we can apply Lemma~\ref{lem:superposition} with $\eta=a_\beta/n$ and $\varepsilon=1/n$, to obtain
	\begin{align}
		D_S(\rho\Vert\sigma)
		\le
		\frac{2a_\beta}{n}
		+2\bigl(M_S+\log n\bigr)x_t
		+2\log\left(1-\frac1n\right)^{-1}.
		\label{eq:buffered-superposition-expanded}
	\end{align}
	We can now bound the two coefficients in
	\eqref{eq:buffered-superposition-expanded} in terms of $\mathcal Q_{S[b]}^{a}(\rho)$ plus a small correction. 

	\begin{figure}
		\centering
		\begin{tikzpicture}[x=1cm,y=1cm,font=\small,>=Stealth]
			\fill[blue!6] (-.45,-.18) rectangle (4.35,.18);
			\draw[blue!45] (-.45,0)--(4.35,0);

			\draw[->,very thick,blue!65!black]
			(0,0)--(3.8,0);

			\draw[->,very thick,orange!85!black]
			(3.8,0)--(3.8,2.1);

			\draw[->,very thick]
			(0,0)--(3.8,2.1)
			node[pos=.48,above=5pt,sloped]
			{$\psi=\sqrt{\rho}$};

			\node[blue!65!black,anchor=north]
			at (2.15,-.28)
			{$\Pi_{S[b],t}\psi=\sqrt{1-x_t}\,v_t$};

			\node[blue!65!black,anchor=north]
			at (2.15,-.90)
			{$\mathcal V_{S[b]}(t)$};

			\node[anchor=north east] at (-.10,-.18) {$0$};

			\draw[black!45]
			(3.50,0)--(3.50,.28)--(3.8,.28);

			\fill (0,0) circle (1.4pt);
			\node[anchor=north east] at (-.10,-.12) {$0$};

			\node[anchor=west,align=left,text width=6.2cm]
			at (5.05,1.25)
			{%
				\textcolor{red!85!black}
				{Remainder: \(\sqrt{x_t}\,w_t\)}\\[4pt]
				Squared length: \(x_t\).\\[6pt]
				The GNS gap controls its weighted average over \(t\).
			};
		\end{tikzpicture}
		\caption{A two-dimensional view of the Hilbert--Schmidt amplitude
			decomposition. Locality bounds the entropy of $v_tv_t^\dagger$;
			the gap controls the weight of the orthogonal remainder. Squaring
			$\psi$ includes cross terms, whose entropy cost is bounded by the
			binary-entropy term in Lemma~\ref{lem:superposition}.}
		\label{fig:summary-projection}
	\end{figure}

	First, notice that from the explicit expression of $M_S$ and the fact that $L\geq \log n$, it follows that there exists a constant $c'_\beta$ such that $2(M_S+\log n)\le c_{\beta}'L$. Moreover, direct Taylor expansion gives the upper bound $2\log\left(1-\frac1n\right)^{-1}\leq 4/n$. In this way we arrive at the uniform estimate
	\begin{equation}
		D_S(\rho\Vert\sigma)
		\le C_{\beta,1}Lx_t+\frac{C_{\beta,1}}{n},
		\label{eq:buffered-pointwise}
	\end{equation}
	where
	\[
	C_{\beta,1}
	:=
		\max\{c'_\beta,\,2a_\beta+4\}.
	\]
	Importantly, the constant $C_{\beta,1}$ is independent of
	$n$, $L$, the chain length, the position of $S$, and $t$.

	At this point, we can average Eq.~\eqref{eq:buffered-pointwise} over $t$, for $|t|\leq 1$, according to the probability distribution $f(t):=2\operatorname{sech}(2\pi t)$ to obtain:
	\begin{align}
		w_0 D_S(\rho\Vert\sigma)
		&\le
		C_{\beta,1}L
		\left(\int_{-1}^1 f(t)x_t\,dt\right)
		+\frac{C_{\beta,1} w_0}{n}
		\label{eq:buffered-integrated}
	\end{align}
	where we defined 
	\[
		w_0:=\int_{-1}^1 2\operatorname{sech}(2\pi t)\,dt = \frac{4}{\pi } \tan ^{-1}(\tanh (\pi )).
	\]
	It remains to control the integral of $x_t$. Since the integrand is non-negative, we can extend the integration range to $\pm \infty$. Then, it directly follows from Lemma~\ref{lem:average} that
	\begin{align}
		\int_\mathbb{R} f(t)\left\langle
		\sqrt\rho,\,
		(1-\Pi_{t})\sqrt\rho
		\right\rangle_{\mathrm{HS}}\,dt
		&\leq
		\gamma_\beta^{-1}
		\left\langle
		\sqrt\rho,\,
		(-\Gamma_\sigma^{1/2}
		\mathcal L^{a}_{S[b]}
		\Gamma_\sigma^{-1/2})(\sqrt\rho)
		\right\rangle_{\mathrm{HS}}\nonumber\\
		&=\gamma_\beta^{-1}\mathcal Q^{a}_{S[b]}(\rho).
	\end{align}
		Plugging this estimate back into Eq.~\eqref{eq:buffered-integrated} gives the desired bound, for example with $C_\beta=C_{\beta,1}\max\{1,(w_0\gamma_\beta)^{-1}\}$. This holds for every state, since Lemma~\ref{lem:superposition} allows singular amplitudes and the energy is defined through $\sqrt\rho$.
\end{proof}
 %
\section{Removal of the bounded entropy defect}
\label{app:tightening}

We now prove the abstract implication that completes the entropy
argument. It combines scalar centering with weighted norm interpolation,
following \cite[Lemma 2.3, Theorem 2.4, and Proposition 2.2]{norms}.
All constants are tracked below; no $L_p$ regularity assumption is used.

\begin{theorem}[Removing a bounded entropy defect]\label{main:tightening}
Let $\cL$ be a primitive KMS-symmetric Markov generator with invariant
state $\sigma$ and KMS gap $\gamma>0$. Put
\[
 \mathcal Q(\rho)
 =-\ip{\sqrt\rho}{\Gamma_\sigma^{1/2}\cL
                         \Gamma_\sigma^{-1/2}(\sqrt\rho)}.
\]
Suppose $D(\rho\Vert\sigma)\le C\mathcal Q(\rho)+b$ for every state,
where $C>0$ and $b\ge0$. Then
\begin{equation}\label{eq:tightening}
 D(\rho\Vert\sigma)
 \le\left(\frac C2+\frac{2+b}{2\gamma}\right)\EP(\rho).
\end{equation}
\end{theorem}

We separate the proof into three lemmas: the first converts the entropy
hypothesis into a defective norm bound, the second removes the norm
defect using the KMS gap, and the third converts norm contraction into
entropy decay. This result is of independent interest and does not necessarily require $\cL$ to be the heat-bath or regularised heat-bath generator. Indeed, throughout $\cL$ is a primitive KMS-symmetric Markov
generator with faithful invariant state $\sigma$.

We write $P_t=e^{t\cL}$ and, for $1\le q<\infty$, we set
\[
 \|X\|_{q,\sigma}
   :=\|\Gamma_\sigma^{1/q}(X)\|_q,
 \qquad E_\sigma X:=\Tr[\sigma X]\one.
\]
These are the usual weighted Schatten norms.  Their interpolation and
duality, and the norm centering inequality used below, are recalled in
\cite[Section 2]{norms}.  In the notation of Section~\ref{pre:quadratic-objects}, we set
$\mathcal E_{\mathrm{KMS}}(X)=-\langle X,\cL(X)\rangle_{\mathrm{KMS}}$.
For positive $X$, we use the $L_2$ relative entropy of
\cite[Definition~3]{kastoryano-temme2013}:
\begin{align*}
 \Ent_{2,\sigma}(X)
 &=\frac12\Tr\!\left[(\Gamma_\sigma^{1/2}(X))^2
 \left(\log[(\Gamma_\sigma^{1/2}(X))^2]-\log\sigma\right)\right]\\
 &\quad-\frac12\|X\|_{2,\sigma}^{2}\log\|X\|_{2,\sigma}^{2}.
\end{align*}

\subsection{From an entropy defect to a norm bound}

We now express the relative entropy on the left-hand side of
\eqref{eq:tightening} in the Heisenberg picture through
$\Ent_{2,\sigma}$ and use the following lemma to begin the norm
estimates in that picture, before returning to the Schr\"odinger
picture to obtain the desired entropy bound.

\begin{lemma}
\label{tightening:defective-norm}
Suppose that, for some $\alpha>0$ and $\delta\geq0$,
\begin{equation}
 \alpha\Ent_{2,\sigma}(X)
 \leq\mathcal E_{\mathrm{KMS}}(X)+\delta\|X\|_{2,\sigma}^2
 \qquad\text{for every positive }X.
 \label{app:ls2-defect}
\end{equation}
Then, with $q(t)=2+2\alpha t$, there is $M(t)\geq1$ such that
\[
 \|P_t\|_{2\to q(t),\sigma}\leq M(t),
 \qquad \log M(t)=\delta t+o(t)\quad(t\downarrow0).
\]
The remainder is uniform over inputs in each fixed finite system.
\end{lemma}
\begin{proof}
For positive $X$ with
$\|X\|_{2,\sigma}=1$, differentiation gives
\begin{equation}
 \left.\frac d{dt}\log\|P_tX\|_{q(t),\sigma}\right|_{t=0^+}
       =\alpha\Ent_{2,\sigma}(X)-\mathcal E_{\mathrm{KMS}}(X)\le\delta.
 \label{app:norm-derivative}
\end{equation}
Here is a justification that also covers singular inputs and the
supremum over $X$.  The positive unit $L_2(\sigma)$ sphere is compact.
For $Y_t=\Gamma_\sigma^{1/q(t)}(P_tX)$,
\[
 Y'_0=\Gamma_\sigma^{1/2}(\cL(X))
       -\frac\alpha4\{(\log\sigma)Y_0+Y_0\log\sigma\}.
\]
On every bounded scalar interval $x\ge0$, the expansion
$x^{2+h}=x^2+h x^2\log x+O(h^2)$ is uniform for $h\downarrow0$:
$x^{2+h}(\log x)^2$ stays bounded, with its value at zero defined
by continuity.  Together with continuity of $x^2\log x$, this proves
a uniform first-order expansion of $\Tr[Y_t^{q(t)}]$ and hence
\eqref{app:norm-derivative} uniformly on that sphere.

We first take the supremum over positive inputs.  Since the remainder
in the preceding expansion is uniform, there is a function
$r(t)\geq0$ with $r(t)=o(t)$ such that
\[
 \log\|P_tX\|_{q(t),\sigma}\leq\delta t+r(t)
 \qquad\text{for every }X\geq0\text{ with }\|X\|_{2,\sigma}=1.
\]
Thus, if we denote the norm restricted to positive inputs by
\[
 N_+(t):=\sup_{\substack{X\geq0\\\|X\|_{2,\sigma}=1}}
                  \|P_tX\|_{q(t),\sigma},
\]
then $N_+(t)\leq e^{\delta t+r(t)}$.

We next show that the same bound holds for arbitrary inputs.
For fixed $t$, we write $q=q(t)$ and pass to the ordinary Schatten
norms using the map
\[
 \Phi_t:=\Gamma_\sigma^{1/q}\circ P_t\circ\Gamma_\sigma^{-1/2}.
\]
Each factor is completely positive.  Moreover,
$\|\Phi_t\|_{2\to q}=\|P_t\|_{2\to q,\sigma}$, and
\[
 \|\Phi_t(B)\|_q\leq N_+(t)\|B\|_2
 \qquad(B\geq0),
\]
because $X=\Gamma_\sigma^{-1/2}(B)$ is positive and
$\|X\|_{2,\sigma}=\|B\|_2$.
For an arbitrary matrix $Z$, the polar decomposition gives
\[
 \begin{pmatrix}|Z^\dagger|&Z\\ Z^\dagger&|Z|\end{pmatrix}\geq0.
\]
Complete positivity therefore implies
\[
 \begin{pmatrix}
  \Phi_t(|Z^\dagger|)&\Phi_t(Z)\\
  \Phi_t(Z)^\dagger&\Phi_t(|Z|)
 \end{pmatrix}\geq0.
\]
Writing $A=\Phi_t(|Z^\dagger|)$ and $B=\Phi_t(|Z|)$, positivity
of this block matrix gives $\Phi_t(Z)=A^{1/2}KB^{1/2}$ for a
contraction $K$.  Schatten H\"older and the positive-input bound yield
\begin{align*}
 \|\Phi_t(Z)\|_q
 &\leq\|A^{1/2}\|_{2q}\|K\|_\infty\|B^{1/2}\|_{2q}\\
 &\leq\|A\|_q^{1/2}\|B\|_q^{1/2}\\
 &\leq N_+(t)\,
       \bigl\||Z^\dagger|\bigr\|_2^{1/2}
       \bigl\||Z|\bigr\|_2^{1/2}
 =N_+(t)\|Z\|_2.
\end{align*}
The last equality uses that $Z$, $|Z|$, and $|Z^\dagger|$ have the
same singular values.  Taking the supremum over arbitrary $Z$ shows
that the full induced norm equals $N_+(t)$.  Consequently, choosing
$M(t):=e^{\delta t+r(t)}$ gives
\begin{equation}
 \|P_t\|_{2\to q(t),\sigma}
       \le M(t),\qquad \log M(t)=\delta t+o(t),\quad M(t)\ge1.
 \label{app:small-hb}
\end{equation}
All remainders here are uniform within each fixed finite system;
their uniformity in system size is not required.
\end{proof}

\subsection{Removing the norm defect using the gap}

We next remove the factor $M(t)$ from the preceding norm bound.
We separate each observable into its equilibrium mean, which the
semigroup preserves, and a centred part, which decays at the rate
given by the spectral gap. By allowing the output exponent to grow
more slowly, we use this decay to compensate for the defect and
obtain a norm bound with constant one.

\begin{lemma}
\label{tightening:remove-norm-defect}
Suppose $P_t$ has KMS gap $\gamma>0$ and satisfies
\[
 \|P_t\|_{2\to 2+2\alpha t,\sigma}\leq M(t),
 \qquad \log M(t)=\delta t+o(t)\quad(t\downarrow0),
\]
for some $\alpha>0$ and $\delta\geq0$.
For every
\[
 0<c<\frac{\alpha\gamma}{\alpha+\delta+\gamma},
\]
we then have $\|P_t\|_{2\to 2+2ct,\sigma}\leq1$ for all sufficiently
small $t>0$.
\end{lemma}
\begin{proof}
We write $q(t)=2+2\alpha t$.
KMS symmetry and the gap give
\[
 \|P_tZ\|_{2,\sigma}\le r(t)\|Z\|_{2,\sigma},
 \qquad r(t)=e^{-\gamma t},\qquad E_\sigma Z=0.
\]
We also use the weighted norm centering inequality
\begin{equation}
 \|a\one+Z\|_{p,\sigma}^2
       \le |a|^2+(p-1)\|Z\|_{p,\sigma}^2,
 \qquad E_\sigma Z=0,\quad p\ge2.
 \label{app:centering}
\end{equation}
This follows from the Schatten uniform-smoothness inequality of
\cite{BallCarlenLieb1994}, evaluated at $a\sigma^{1/p}$ in the direction
$\Gamma_\sigma^{1/p}(Z)$; its linear term vanishes because
$\Tr[\sigma Z]=0$.  Equivalently it is
\cite[Lemma 2.3, equation (2.15)]{norms}.

We choose any
\[
 0<c<\frac{\alpha\gamma}{\alpha+\delta+\gamma},
 \qquad p(t)=2+2ct.
\]
Since $c<\alpha$, we have $2<p(t)<q(t)$. For centred $Z$, the
gap controls the output in $L_2(\sigma)$, while the hypothesis
controls it in $L_{q(t)}(\sigma)$. Interpolating these two bounds gives
\begin{align*}
 \|P_tZ\|_{p(t),\sigma}
 &\le\|P_tZ\|_{2,\sigma}^{1-\theta(t)}
       \|P_tZ\|_{q(t),\sigma}^{\theta(t)}\\
 &\le r(t)^{1-\theta(t)}M(t)^{\theta(t)}\|Z\|_{2,\sigma},
\end{align*}
where $1/p(t)=(1-\theta(t))/2+\theta(t)/q(t)$, so
\[
 \theta(t)=\frac{1/2-1/p(t)}{1/2-1/q(t)}
          =\frac c\alpha+O(t).
\]
Using $\log(p(t)-1)=2ct+O(t^2)$, $\log r(t)=-\gamma t$, and
$\log M(t)=\delta t+o(t)$, we find that the coefficient needed in
\eqref{app:centering} is strictly below one for sufficiently small $t>0$:
\[
 \log\!\left[(p(t)-1)r(t)^{2(1-\theta(t))}
                         M(t)^{2\theta(t)}\right]
 =2\left[c-\gamma+\frac c\alpha(\gamma+\delta)\right]t+o(t)<0.
\]
The strict negativity follows from
$c(\alpha+\delta+\gamma)<\alpha\gamma$.

For an arbitrary observable $X$, we now set $a=\Tr[\sigma X]$ and
$Z=X-a\one$. Then $E_\sigma Z=0$, and orthogonality gives
$\|X\|_{2,\sigma}^2=|a|^2+\|Z\|_{2,\sigma}^2$.
Unitality and invariance of $\sigma$ imply
$P_tX=a\one+P_tZ$ and $E_\sigma(P_tZ)=0$. We may therefore apply
\eqref{app:centering} to obtain
\begin{align*}
 \|P_tX\|_{p(t),\sigma}^2
 &\le |a|^2+(p(t)-1)\|P_tZ\|_{p(t),\sigma}^2\\
 &\le |a|^2+(p(t)-1)r(t)^{2(1-\theta(t))}
                   M(t)^{2\theta(t)}\|Z\|_{2,\sigma}^2\\
 &\le |a|^2+\|Z\|_{2,\sigma}^2
 =\|X\|_{2,\sigma}^2.
\end{align*}
Taking the supremum over $\|X\|_{2,\sigma}=1$ proves
\begin{equation}
             \|P_t\|_{2\to p(t),\sigma}\le1
             \qquad\text{for all sufficiently small }t>0.
 \label{app:small-hc}
\end{equation}
\end{proof}

\subsection{From norm contraction to entropy decay}

We now return to the Schr\"odinger picture. The following lemma
translates the norm contraction for observables into a contraction
of relative entropy for states, using duality and interpolation
towards exponent one. Applied to the semigroup at small times,
this gives the entropy-production inequality needed to conclude MLSI.

\begin{lemma}
\label{tightening:norm-to-entropy}
Let $\Phi$ be a unital completely positive map with
$\Phi_*(\sigma)=\sigma$. If $\|\Phi\|_{2\to p,\sigma}\leq1$ for
some $p>2$, then every state $\rho$ satisfies
\[
 D(\Phi_*\rho\Vert\sigma)\leq\frac2pD(\rho\Vert\sigma).
\]
Consequently, if $\|P_t\|_{2\to 2+2ct,\sigma}\leq1$ for all
sufficiently small $t>0$ and some $c>0$, then
\[
 \EP(\rho)\geq cD(\rho\Vert\sigma),\qquad
 D(P_{t*}\rho\Vert\sigma)\leq e^{-ct}D(\rho\Vert\sigma)
 \quad(t\geq0).
\]
\end{lemma}
\begin{proof}
We set $p'=p/(p-1)$ and write the KMS adjoint as
\[
 \Phi^\#=\Gamma_\sigma^{-1}\Phi_*\Gamma_\sigma.
\]
Weighted H\"older duality gives
\begin{align*}
 \|\Phi^\#(X)\|_{2,\sigma}
 &=\sup_{\|Y\|_{2,\sigma}=1}
       |\langle\Phi^\#(X),Y\rangle_{\mathrm{KMS}}|\\
 &=\sup_{\|Y\|_{2,\sigma}=1}
       |\langle X,\Phi(Y)\rangle_{\mathrm{KMS}}|
 \le\|X\|_{p',\sigma}.
\end{align*}
At the other endpoint, $\Gamma_\sigma\Phi^\#=\Phi_*\Gamma_\sigma$
and trace-norm contraction of $\Phi_*$ imply
\[
 \|\Phi^\#(X)\|_{1,\sigma}
 =\|\Phi_*(\Gamma_\sigma(X))\|_1
 \le\|\Gamma_\sigma(X)\|_1=\|X\|_{1,\sigma}.
\]
Interpolating these bounds, for $0\le s\le1$ we obtain
\[
 \|\Phi^\#(X)\|_{b_s,\sigma}\le\|X\|_{a_s,\sigma},
 \qquad
 \frac1{a_s}=1-s+\frac{s}{p'}=1-\frac{s}{p},\qquad
 \frac1{b_s}=1-s+\frac{s}{2}=1-\frac{s}{2}.
\]
In particular,
\[
 a_0=b_0=1,\qquad a'_0=\frac1p,\qquad b'_0=\frac12.
\]
For $f_\rho=\Gamma_\sigma^{-1}(\rho)$, we have
$\Phi^\#(f_\rho)=f_{\Phi_*\rho}$, so
\[
 \log\|f_{\Phi_*\rho}\|_{b_s,\sigma}
 \le\log\|f_\rho\|_{a_s,\sigma}.
\]
Both sides vanish at $s=0$.

For faithful $\rho$, we write
\[
 A(u)=\Gamma_\sigma^{1/u-1}(\rho),\qquad
 F(u)=\Tr[A(u)^u],\qquad
 \log\|f_\rho\|_{u,\sigma}=\frac1u\log F(u).
\]
At $u=1$ we have $A(1)=\rho$, $F(1)=1$, and
\[
 A'(1)=-\tfrac12\bigl((\log\sigma)\rho+\rho\log\sigma\bigr).
\]
Differentiating under the trace gives
\begin{align*}
 F'(1)
 &=\Tr[\rho\log\rho]+\Tr[A'(1)]\\
 &=\Tr[\rho\log\rho]-\Tr[\rho\log\sigma]
 =D(\rho\Vert\sigma).
\end{align*}
Consequently,
\[
 \left.\frac d{du}\log\|f_\rho\|_{u,\sigma}\right|_{u=1}
 =D(\rho\Vert\sigma).
\]
The same calculation applies to $\Phi_*\rho$: if $\rho\ge\varepsilon\sigma$
for some $\varepsilon>0$, positivity and invariance give
$\Phi_*\rho\ge\varepsilon\sigma$. Differentiating the interpolated
norm inequality at $s=0$ yields
\[
 \frac12D(\Phi_*\rho\Vert\sigma)
 \le\frac1pD(\rho\Vert\sigma).
\]

Approximation by $(1-\varepsilon)\rho+\varepsilon\sigma$ and
continuity of relative entropy extend the channel inequality to all states.

Applying this estimate with $\Phi=P_t$ and $p=2+2ct$ gives, for
all sufficiently small $t>0$,
\[
 D(P_{t*}\rho\Vert\sigma)\le\frac{1}{1+ct}D(\rho\Vert\sigma).
\]
Differentiating at $t=0^+$ and using \eqref{eq:ep} gives
$\EP(\rho)\ge cD(\rho\Vert\sigma)$. The admissible time interval
need not be uniform in the system size.

Applying this inequality along $\rho_t=P_{t*}\rho$ gives
\[
 \frac d{dt}D(\rho_t\Vert\sigma)
 =-\EP(\rho_t)\le-cD(\rho_t\Vert\sigma),\qquad
 D(\rho_t\Vert\sigma)\le e^{-ct}D(\rho\Vert\sigma).
\]
Finally, approximation by $(1-\varepsilon)\rho+\varepsilon\sigma$
extends the entropy contraction to all states. Its right derivative
at zero gives the same entropy-production inequality, with the usual
extended value allowed for $\EP$.
\end{proof}

\subsection{Proof of the entropy defect removal theorem}

\begin{proof}[Proof of Theorem~\ref{main:tightening}]
For positive nonzero $X$, we apply the assumed entropy bound to
\[
 \rho=\frac{(\Gamma_\sigma^{1/2}(X))^2}{\|X\|_{2,\sigma}^2}.
\]
The identity
$\sqrt\rho=\Gamma_\sigma^{1/2}(X)/\|X\|_{2,\sigma}$ gives
$\mathcal Q(\rho)=\mathcal E_{\mathrm{KMS}}(X)/\|X\|_{2,\sigma}^2$,
while the definition of $\Ent_{2,\sigma}$ gives
$D(\rho\Vert\sigma)=2\Ent_{2,\sigma}(X)/\|X\|_{2,\sigma}^2$.
Thus the hypothesis becomes
\[
 2\Ent_{2,\sigma}(X)
 \leq C\mathcal E_{\mathrm{KMS}}(X)+b\|X\|_{2,\sigma}^2.
\]
This is the Heisenberg-picture inequality \eqref{app:ls2-defect} with
\[
 \alpha=\frac2C,\qquad\delta=\frac bC.
\]
Lemma~\ref{tightening:defective-norm} therefore gives
\[
 \|P_t\|_{2\to 2+2\alpha t,\sigma}\leq M(t),
 \qquad \log M(t)=\delta t+o(t).
\]
Together with the KMS gap $\gamma$, this verifies the hypotheses of
Lemma~\ref{tightening:remove-norm-defect}. Hence, for every
$0<c<c_*$, we have $\|P_t\|_{2\to 2+2ct,\sigma}\leq1$ for
sufficiently small $t>0$, where
\[
 c_*:=\frac{\alpha\gamma}{\alpha+\delta+\gamma}
      =\frac{2\gamma}{2+b+C\gamma}.
\]
We now return to the Schr\"odinger picture using
Lemma~\ref{tightening:norm-to-entropy}, which gives
$\EP(\rho)\geq cD(\rho\Vert\sigma)$ for every state $\rho$.
Letting $c\uparrow c_*$ proves
\[
 D(\rho\Vert\sigma)
 \leq\frac1{c_*}\EP(\rho)
 =\left(\frac C2+\frac{2+b}{2\gamma}\right)\EP(\rho),
\]
as asserted.
\end{proof}
 %
\section{Proof of the global MLSI and rapid mixing}\label{sec:conclusion}
Our goal is to prove Theorem~\ref{thm:main}: for the heat-bath and
regularised heat-bath samplers, with one rate-one update per site,
\[
 D(\rho\Vert\sigma)\le C_\beta\log n\,\EP^a(\rho),
 \qquad a\in\{\mathrm{HB},\mathrm{RHB}\},\quad n\ge2,
\]
for every faithful state $\rho$, with $C_\beta$ independent of $n$.
We then deduce the entropy decay and rapid mixing of
Corollary~\ref{cor:mixing}, in particular
\[
 \max\{t_{{\rm mix},D}^a(\epsilon),t_{\rm mix}^a(\epsilon)\}
 \le C_\beta'\log n\,[\log n+\log\epsilon^{-1}],
 \qquad 0<\epsilon<1.
\]
In order to prove this, we first combine weak quasi-factorisation and the buffered
entropy--energy bound to obtain a bounded global entropy defect.
The uniform gap then removes this defect and yields the MLSI.
The transfer between the two samplers uses the quadratic-form order
\eqref{eq:energy-comparison}.

\subsection{A bounded global entropy defect}

We combine the weak quasi-factorisation of Theorem~\ref{stat:weak-at}
with the buffered entropy--energy estimate of Theorem~\ref{ent:buffered},
using the bounded overlap of the buffered intervals to obtain a global
bound with an entropy defect independent of $n$.

\begin{proposition}\label{global:defective}
For both $a\in\{\mathrm{HB},\mathrm{RHB}\}$ there are constants
$C_\beta,B_\beta$ independent of $n$ such that
\begin{equation}\label{eq:global-defective}
 D(\rho\Vert\sigma)
 \le C_\beta\log n\,\mathcal Q_\Lambda^a(\rho)+B_\beta.
\end{equation}
\end{proposition}
\begin{proof}
Choose
\[
 L=\lceil K_\beta\log n\rceil,
 \qquad b=\lceil k_\beta L\rceil.
\]
Take $K_\beta\ge1$ large enough for Theorem~\ref{stat:weak-at} to give
an error at most one, and choose $k_\beta$ as in
Theorem~\ref{ent:buffered}. Let
$\mathcal G=\{A_i\}_i\cup\{B_j\}_j$ be the cover from
Theorem~\ref{stat:weak-at}, illustrated in Figure~\ref{fig:summary-cover}.
For $n\ge4L$, its intervals have length in $[L,6L]$ and multiplicity
at most two. For $n<4L$, the single interval $\Lambda$ suffices and
the static error is zero.

The enlarged intervals have bounded multiplicity. For $n\ge4L$, if
$x\in S[b]$, then $S$ meets $[x-b,x+b]\cap\Lambda$.
At most four such cores are not contained in this latter interval:
each contains an endpoint, and the original multiplicity is two.
Every remaining core has length at least $L$, and the sum of their
lengths is at most $2(2b+1)$. Thus, for $n\ge4L$,
\begin{equation}\label{global:overlap}
 \nu:=\sup_x|\{S\in\mathcal G:x\in S[b]\}|
 \le4+\frac{2(2b+1)}L\le C_\beta.
\end{equation}
For $n<4L$, the cover is $\mathcal G=\{\Lambda\}$ and $\nu=1$.
Positivity and additivity of the
site energies yield
\[
 \sum_{S\in\mathcal G}\mathcal Q_{S[b]}^{\mathrm{HB}}(\rho)
 \le\nu\mathcal Q_\Lambda^{\mathrm{HB}}(\rho).
\]
There are at most $2n$ cores. Combining weak quasi-factorisation with
Theorem~\ref{ent:buffered} therefore gives
\begin{align*}
 D(\rho\Vert\sigma)
 &\le\sum_{S\in\mathcal G}D_S(\rho\Vert\sigma)+1\\
 &\le C_\beta L\sum_{S\in\mathcal G}\mathcal Q_{S[b]}^{\mathrm{HB}}(\rho)
                       +\frac{C_\beta|\mathcal G|}{n}+1\\
 &\le C_\beta'\log n\,\mathcal Q_\Lambda^{\mathrm{HB}}(\rho)+B_\beta.
\end{align*}
We used $L\le(K_\beta+1/\log2)\log n$. When $n<4L$, the bound follows from
the same buffered estimate with $S=S[b]=\Lambda$ and $|S|=n<4L$.
Finally, \eqref{eq:energy-comparison} gives
\[
 \mathcal Q_\Lambda^{\mathrm{HB}}(\rho)
 \le \mathcal Q_\Lambda^{\mathrm{RHB}}(\rho),
\]
so the bound transfers from the heat-bath energy to the regularised
heat-bath energy by replacing the former with the latter on the
right-hand side.
\end{proof}

\subsection{Proof of the main result}

\begin{proof}[Proof of Theorem~\ref{thm:main}]
Proposition~\ref{pre:properties} proves primitivity and KMS symmetry
for both generators, with invariant state $\sigma$.
Proposition~\ref{loc:heatbath-quasilocal} proves
quasi-locality of the heat-bath generator by constructing local channel
approximations with exponentially decaying diamond-norm error.
It remains to prove the MLSI for $n\ge2$.

The whole-chain case of Lemma~\ref{lem:average}, namely
\eqref{eq:kms-gap}, gives a heat-bath KMS gap
$\gamma_\beta>0$ independent of $n$. Together with the quadratic-form
comparison \eqref{eq:energy-comparison}, this yields, for every observable $X$,
\[
 \mathcal E_{\mathrm{KMS}}^{\mathrm{RHB}}(X)
 \ge \mathcal E_{\mathrm{KMS}}^{\mathrm{HB}}(X)
 \ge \gamma_\beta\operatorname{Var}_\sigma^{\mathrm{KMS}}(X).
\]
Thus both generators have a KMS gap bounded below by the same
$\gamma_\beta$.

For either $a\in\{\mathrm{HB},\mathrm{RHB}\}$,
Proposition~\ref{global:defective} supplies the entropy bound
\eqref{eq:global-defective}, with $C=C_\beta\log n$ and $b=B_\beta$.
By \eqref{eq:quadratic}, $\mathcal Q_\Lambda^a$ is precisely the
quadratic energy appearing in Theorem~\ref{main:tightening}.
All hypotheses of that theorem therefore hold for $\cL^a$, and it gives
\[
 D(\rho\Vert\sigma)
 \le\left[\frac{C_\beta\log n}2
                +\frac{2+B_\beta}{2\gamma_\beta}\right]\EP^a(\rho).
\]
Since $\log n\ge\log2$, the coefficient is at most
$C_\beta'\log n$, where
\[
 C_\beta':=\frac{C_\beta}{2}
            +\frac{2+B_\beta}{2\gamma_\beta\log2}
\]
is independent of $n$. This proves \eqref{eq:mlsi} for every faithful
state $\rho$, and hence
$\alpha_{\rm MLSI}(\cL_*^a)\ge c_\beta/\log n$ with
$c_\beta=(C_\beta')^{-1}$.
\end{proof}

\begin{remark}[Alternative proof for the regularised heat-bath sampler]
\label{global:rhb-alternative}
Appendix~\ref{app:rhb-direct} gives a direct proof of the regularised
heat-bath MLSI. Combining the energy estimate of
Corollary~\ref{cor:rhb-direct-buffer} with
Theorem~\ref{stat:weak-at} and the overlap bound
\eqref{global:overlap} yields \eqref{eq:global-defective} for this
sampler. The whole-chain case of
Proposition~\ref{lem:rhb-direct-comparisons} gives its uniform KMS
gap, and Theorem~\ref{main:tightening} then proves the MLSI.
This route uses the conditional-expectation structure of the
regularised updates directly, without modular averaging or transfer
of the heat-bath bound through \eqref{eq:energy-comparison}.
\end{remark}

\begin{proof}[Proof of Corollary~\ref{cor:mixing}]
For a faithful initial state, set $\rho_t=e^{t\cL_*^a}\rho$.
Differentiation gives
\[
 \frac d{dt}D(\rho_t\Vert\sigma)=-\EP^a(\rho_t)
 \le-\frac{c_\beta}{\log n}D(\rho_t\Vert\sigma).
\]
Integration proves \eqref{eq:entropy-decay}. Since $\sigma$ is faithful,
continuity in the initial state extends this to all $\rho$ by replacing
$\rho$ with $(1-\delta)\rho+\delta\sigma$ and letting $\delta\downarrow0$.
The interaction bound gives $\|H_\Lambda\|_\infty\le Jn$, and hence
\[
 D(\rho\Vert\sigma)\le\|\log\sigma\|_\infty
                         \le n(\log d+2\beta J).
\]
Together with \eqref{eq:entropy-decay}, this gives
\[
 \sup_\rho D(e^{t\cL_*^a}\rho\Vert\sigma)
 \le n(\log d+2\beta J)e^{-c_\beta t/\log n}.
\]
Solving for a right-hand side at most $\epsilon$ bounds
$t_{{\rm mix},D}^a(\epsilon)$ as in \eqref{eq:mixing}.
By \eqref{eq:mixing-comparison}, replacing $\epsilon$ by $\epsilon^2/2$
also bounds $t_{\rm mix}^a(\epsilon)$, after enlarging the fixed constant.
\end{proof}

\begin{remark}[Normalisation, uniform families, and complete inequalities]
The stated bounds use the normalisation of the Lindbladians in
\eqref{eq:generators}, in which each update term appears with
coefficient one. Replacing $\cL_*^a$ by $\cL_*^a/n$ divides its
MLSI constant by $n$ and multiplies the mixing-time bound by $n$.

For the families on increasing volumes introduced in
Section~\ref{pre:maps}, all constants in the gap, MLSI, and mixing
estimates depend only on $d,r_0,J,\beta$. Since the updates are
channels, each individual generator term has diamond norm at most
two, uniformly in $m$ and the update site.
The heat-bath family also satisfies the local
approximation bound \eqref{loc:heatbath-channel-tail} uniformly in $m$.

The theorem concerns the two specified global-reference samplers.
The complete local statements are given in Proposition~\ref{cmlsi:local}.
Corollary~\ref{cmlsi:global} gives global CMLSI with weaker volume
dependence, including an order-$1/n$ rate for the
regularised heat-bath sampler. The
order-$1/\log n$ theorem remains ordinary MLSI, for the reason
explained in Remark~\ref{cmlsi:scope}.
\end{remark}
 %
\section{Discussion}
\label{sec:discussion}

We have obtained an MLSI constant of order at least $1/\log n$ for
the heat-bath and regularised heat-bath samplers of one-dimensional
finite-range Gibbs states at every fixed positive temperature.
The proof allows additive errors in both the quasi-factorisation of
relative entropy and the comparison with buffered energy. We control
their sum by a constant independent of the chain length and then use
the spectral gap to remove this defect. The abstract implication in
Theorem~\ref{main:tightening} may also be useful for other dynamics
for which entropy estimates with a bounded additive error are available.

The logarithmic loss in our bound comes from choosing cores and
buffers of size $O(\log n)$ to control the accumulated errors. The
argument does not establish that this loss is intrinsic to either
sampler. A natural question is whether a bound independent of $n$
can be obtained, for example through a sharper conditional entropy
estimate or a different treatment of the errors across the cover.
A similar improvement has already been achieved for Davies dynamics
of one-dimensional commuting Hamiltonians: the MLSI bounds of order
$1/\log n$ in
\cite{BardetEtAlRapidThermalization2023,BardetEtAlEntropyDecayDavies}
were subsequently improved to a positive constant independent of $n$
in \cite{KochanowskiAlhambraCapelRouze2024} through a different proof
strategy. We hope that a corresponding change of strategy will also
remove the logarithmic loss in the noncommuting setting considered here.

The two samplers also illustrate a distinction between entropy decay
and locality. The heat-bath updates admit local channel approximations
with exponentially decaying error, whereas the limiting conditional
expectations defining the regularised sampler can act on the entire
chain, as shown in Remark~\ref{loc:rhb-nonlocal}. The direct argument
in Appendix~\ref{app:rhb-direct} simplifies the entropy comparison
for the latter, but this simplification does not imply locality.

Another question is whether the same $1/\log n$ rate can be obtained
for complete MLSI. The local complete inequalities are already
available, but the scalar centering step in the removal of the global
defect does not extend directly in the presence of an untouched
memory; see Remark~\ref{cmlsi:scope}. A corresponding argument relative
to the fixed-point algebra would address this obstruction.

Extending the argument beyond one dimension is a central open direction.
For such an extension, the first step of the proof is already missing:
the quasi-factorisation of relative entropy into overlapping systems.
The one-dimensional construction uses two such systems, each composed
of separated intervals, whereas an analogous construction in higher
dimensions would require at least three overlapping systems. We would
therefore need a corresponding multipartite quasi-factorisation estimate
with a controlled additive error. Establishing this estimate is a first
obstacle, in addition to obtaining suitable higher-dimensional analogues
of the correlation and buffered locality estimates used here.
 
\begin{acknowledgments}
A.C. and D.P.G. would like to thank the Isaac Newton Institute for Mathematical Sciences, Cambridge, for support and hospitality during the programme Mathematics of Many-Body Entanglement, where work on this paper was undertaken. This work was supported by EPSRC grant EP/Z000580/1.
A.C. was financially supported by the DFG under grant TRR 352 -- Project-ID 470903074. M.S. acknowledges support from the Spanish Agencia Estatal de Investigación through the grants “PID2023-150847NA-I00” and “PCI2024-153448”, funded by MICIU/AEI/10.13039/501100011033 and co-funded by the EU. This publication has also been funded within the framework of the R\&D\&I Project CEX2025-001574-S, funded by MICIU/AEI/10.13039/501100011033. The research presented in this publication falls within the research line “Quantum Matter and Fields”. This project was funded within the QuantERA II Programme, which received funding from the European Union's Horizon 2020 research and innovation programme under Grant Agreement No.~101017733. 
\end{acknowledgments}

\section*{AI Statement}

The main conceptual contributions were developed by the authors. We acknowledge the use of LLMs for a proof of MLSI, instead of weak MLSI, and for proofreading, refining proof ideas and creating figures. 
 The authors take full responsibility for the final manuscript.

\bibliographystyle{plainurl}
\bibliography{MLSI_1D_noncommuting}

\appendix
%
\section{Direct regularised heat-bath comparison and local CMLSI}\label{app:rhb-direct}
The regularised heat-bath sampler permits a shorter conditional argument:
its updates are conditional expectations and commute with the modular
group. This removes the average over modular rotations. We give this
alternative proof of the regularised heat-bath MLSI explicitly; the main
argument in Section~\ref{sec:conclusion} instead transfers the heat-bath
estimate through the quadratic-form comparison
\eqref{eq:energy-comparison}.

\begin{proposition}[Direct regularised heat-bath estimates]\label{lem:rhb-direct-comparisons}
For every interval $S$ and buffer $S[b]$,
\begin{align}
 -\Gamma_\sigma^{1/2}\cL_{S[b]}^{\mathrm{RHB}}\Gamma_\sigma^{-1/2}
   &\ge\gamma_\beta(1-\Pi_{S[b],0}),\nonumber\\
 \intertext{and}
 \EP_{S[b]}^{\mathrm{RHB}}(\rho)&\ge4\mathcal Q_{S[b]}^{\mathrm{RHB}}(\rho).
 \label{eq:rhb-direct-comparisons}
\end{align}
Here $\EP_{S[b]}^{\mathrm{RHB}}$ is the entropy production of
$\sum_{i\in S[b]}(\mathbb{E}_i-\id)$, still with reference $\sigma$.
\end{proposition}
\begin{proof}
We first make the conditional-expectation structure explicit. For
a fixed site $i$, the finite-dimensional algebra decomposition gives
\begin{align}
 \mathcal N_i&=\bigoplus_j\mathcal B(\mathcal H_{j,L})\otimes\one_{j,R},
 &\sigma&=\bigoplus_j q_j\sigma_{j,L}\otimes\theta_{j,R},\nonumber\\
 \mathbb{E}_{i*}\rho&=\bigoplus_j\rho_{j,L}\otimes\theta_{j,R},
 &\rho_{j,L}&=\Tr_{j,R}[P_j\rho P_j].
 \label{rhb:blocks}
\end{align}
The $P_j$ are the central block projections; $q_j>0$ and
$\sigma_{j,L},\theta_{j,R}$ are density matrices. Bimodularity makes
$\mathbb{E}_i$ the identity on each left factor and a state functional on its
right factor. Stationarity then forces exactly the displayed product
form of $\sigma$. This also proves that $\mathbb{E}_i$ commutes with
$\Delta_\sigma^z$ for every $z$; cf.\
\cite[Proposition 1(iv), equation (2.4)]{BardetCapelRouze2021NoncommutingAT}.
It is orthogonal in the GNS inner product, since for
$A\in\mathcal N_i$,
\[
 \Tr\bigl[\sigma A^\dagger(X-\mathbb{E}_iX)\bigr]
 =\Tr\bigl[\sigma \mathbb{E}_i\bigl(A^\dagger(X-\mathbb{E}_iX)\bigr)\bigr]=0.
\]

Theorem~\ref{gap:conditional} gives the conditional GNS gap
$\gamma_\beta$ for the regularised heat-bath sampler. Since
$\Delta_\sigma^{-1/4}\Gamma_\sigma^{1/2}(X)=X\sigma^{1/2}$,
we can express this gap in Hilbert--Schmidt coordinates.
Modular covariance implies that $\cL_{S[b]}^{\mathrm{RHB}}$ commutes
with $\Delta_\sigma^{-1/4}$, as does $\Gamma_\sigma$. The gap becomes
\[
 -\Delta_\sigma^{-1/4}\Gamma_\sigma^{1/2}\cL_{S[b]}^{\mathrm{RHB}}
       \Gamma_\sigma^{-1/2}\Delta_\sigma^{-1/4}
 \ge\gamma_\beta
          \Delta_\sigma^{-1/4}(1-\Pi_{S[b],0})\Delta_\sigma^{-1/4}.
\]
Congruence by $\Delta_\sigma^{1/4}$ proves the first estimate in
\eqref{eq:rhb-direct-comparisons}.

For the entropy comparison, we take $S=\{i\}$ and a one-dimensional
memory in \eqref{cmlsi:rhb} of Proposition~\ref{cmlsi:local}.
The generator is then $\mathbb{E}_i-\id$, whose entropy production
is $\EP_i^{\mathrm{RHB}}$, so the identity gives
\[
 \EP_i^{\mathrm{RHB}}(\rho)
 =D(\rho\Vert\mathbb{E}_{i*}(\rho))+D(\mathbb{E}_{i*}(\rho)\Vert\rho).
\]
To obtain the quadratic lower bound, we write
$\rho=\sum_a p_a|u_a\rangle\langle u_a|$ and
$\mathbb{E}_{i*}(\rho)=\sum_b q_b|v_b\rangle\langle v_b|$.
Expanding the two relative entropies and applying
$(p-q)(\log p-\log q)\ge4(\sqrt p-\sqrt q)^2$ gives
\begin{align*}
 \EP_i^{\mathrm{RHB}}(\rho)
 &=\sum_{a,b}|\langle u_a,v_b\rangle|^2
       (p_a-q_b)(\log p_a-\log q_b)\\
 &\ge4\sum_{a,b}|\langle u_a,v_b\rangle|^2
       (\sqrt{p_a}-\sqrt{q_b})^2\\
 &=4\|\sqrt\rho-\sqrt{\mathbb{E}_{i*}(\rho)}\|_2^2.
\end{align*}
The decomposition \eqref{rhb:blocks} shows that
\[
 \sqrt{\mathbb{E}_{i*}(\rho)}=\Gamma_\sigma^{1/2}\!\left(
 \bigoplus_j\Gamma_{q_j\sigma_{j,L}}^{-1/2}
             (\sqrt{\rho_{j,L}})\otimes\one_{j,R}\right).
\]
It belongs to the range of the orthogonal projection
$\widehat{\mathbb{E}}_i=\Gamma_\sigma^{1/2}\mathbb{E}_i\Gamma_\sigma^{-1/2}$.
Since projection minimises distance,
\[
 \mathcal Q_i^{\mathrm{RHB}}(\rho)
 =\|(1-\widehat{\mathbb{E}}_i)\sqrt\rho\|_2^2
 \le\|\sqrt\rho-\sqrt{\mathbb{E}_{i*}(\rho)}\|_2^2.
\]
We sum over $i\in S[b]$. Faithful approximation gives the extended-value
statement for singular states.
\end{proof}

The following corollary converts the preceding proposition into a
bound on the conditional entropy of a core $S$ in terms of the
regularised heat-bath energy, or entropy production, in its buffer.
For cores and buffers of size $O(L)$ with $L\gtrsim\log n$, the
coefficient is $O(L)$ and the additive error is $O(1/n)$. Summing the
energy estimate over the overlapping cover and using weak
quasi-factorisation gives the bounded global entropy defect needed
for Theorem~\ref{main:tightening}. This supplies the direct alternative
to the quadratic-form transfer used in Section~\ref{sec:conclusion}.

\begin{corollary}[Direct buffered regularised heat-bath inequality]\label{cor:rhb-direct-buffer}
If every unit $v\in\mathcal V_{S[b]}(0)$ satisfies
$D_S(vv^\dagger\Vert\sigma)\le\eta_b$, then, for $0<\varepsilon<1$,
\begin{align}
 D_S(\rho\Vert\sigma)
 &\le\frac{2(M_S+\log\varepsilon^{-1})}{\gamma_\beta}
                 \mathcal Q_{S[b]}^{\mathrm{RHB}}(\rho)
       +2\eta_b+2\log(1-\varepsilon)^{-1},\label{rhb:buffer-energy}\\
 D_S(\rho\Vert\sigma)
 &\le\frac{M_S+\log\varepsilon^{-1}}{2\gamma_\beta}
                 \EP_{S[b]}^{\mathrm{RHB}}(\rho)
       +2\eta_b+2\log(1-\varepsilon)^{-1}.
 \label{eq:rhb-direct-buffer}
\end{align}
In particular, for $L\ge\log n$, $|S|\le6L$, and
$b=\lceil k_\beta L\rceil$ with sufficiently large $k_\beta$, both
right-hand sides have coefficient $C_\beta L$ and error $C_\beta/n$.
\end{corollary}
\begin{proof}
Use Lemma~\ref{lem:superposition} with $\Pi=\Pi_{S[b],0}$ and
substitute
\[
 \|(1-\Pi_{S[b],0})\sqrt\rho\|_2^2
 \le\frac1{\gamma_\beta}\mathcal Q_{S[b]}^{\mathrm{RHB}}(\rho)
 \le\frac1{4\gamma_\beta}\EP_{S[b]}^{\mathrm{RHB}}(\rho).
\]
For the last assertion take $\varepsilon=n^{-1}$ and use
Corollary~\ref{main:shielding} and $M_S\le c_\beta(|S|+1)$.
\end{proof}

The entropy-production form \eqref{eq:rhb-direct-buffer} directly
gives a weak MLSI after summation. The energy form
\eqref{rhb:buffer-energy} supplies the quadratic entropy inequality
required to remove the additive defect using
Theorem~\ref{main:tightening}.
The local CMLSI in \eqref{cmlsi:rhb} holds for every memory, but no
complete global defect-removal step is used here.
 \end{document}